\documentclass[12pt]{article}
\pdfoutput=1
\usepackage{jheppub}

\usepackage{amssymb}
\usepackage{color}

\usepackage{hyperref}
\usepackage{amsthm}
\usepackage{graphicx}
\usepackage{dcolumn}
\usepackage{bm}
\usepackage{amssymb}
\usepackage{verbatim}
\usepackage{amscd}
\usepackage{amssymb}
\usepackage{amsmath, amsfonts}
\usepackage{setspace}
\usepackage{amsthm}
\usepackage{enumerate}

\usepackage[caption=false]{subfig}

\newcommand{\tb}[1]{\textbf{#1}}

\theoremstyle{plain}
\newtheorem{theorem}{Theorem}
\theoremstyle{plain}
\newtheorem{lemma}{Lemma}
\theoremstyle{plain}
 
\theoremstyle{plain}

\theoremstyle{result}

\theoremstyle{remark}

\theoremstyle{conjecture}

\theoremstyle{observation}

\theoremstyle{definition}

\theoremstyle{corollary}

\theoremstyle{definition}

\theoremstyle{definition}
\newtheorem{definition}{Definition}
\theoremstyle{assumption}

\theoremstyle{definition}

\theoremstyle{problem}

\theoremstyle{fact}

\newcommand*{\cC}{\mathcal{C}}
\newcommand*{\cE}{\mathcal{E}}

\newcommand*{\cH}{\mathcal{H}}

\newcommand*{\cS}{\mathcal{S}}

\newcommand{\be}{\begin{equation}}
\newcommand{\ee}{\end{equation}}
\newcommand{\bfig}{\begin{figure}[htb!] \centering}
\newcommand{\efig}{\end{figure}}
\newcommand{\lan}{\langle}
\newcommand{\ran}{\rangle}
\newcommand{\HA}{\mathcal{H}_A}
\newcommand{\HB}{\mathcal{H}_B}

\newcommand{\CA}{\mathcal{C}[A]}
\newcommand{\EA}{\mathcal{E}[A]}

\begin{document}
\title{Holographic quantum error-correcting codes: Toy models for the bulk/boundary correspondence}
\author[a]{Fernando Pastawski,*}
\author[a]{Beni Yoshida*}
\author[b]{Daniel Harlow,} 
\author[a]{John Preskill,}
\affiliation[a]{Institute for Quantum Information \& Matter and Walter Burke Institute for Theoretical Physics, California Institute of Technology, Pasadena, California 91125, USA}
\affiliation[b]{Princeton Center for Theoretical Science, Princeton University, Princeton NJ 08540 USA}
\emailAdd{fernando.pastawski@gmail.com}
\emailAdd{rouge@caltech.edu}
\emailAdd{dharlow@princeton.edu}
\emailAdd{preskill@caltech.edu}

\affiliation{\emph{*These authors contributed equally to this work.}}
 \date{\currenttime \today}
\abstract{
We propose a family of exactly solvable toy models for the AdS/CFT correspondence based on a novel construction of quantum error-correcting codes with a tensor network structure. 
Our building block is a special type of tensor with maximal entanglement along any bipartition, which gives rise to an isometry from the bulk Hilbert space to the boundary Hilbert space. The entire tensor network is an encoder for a quantum error-correcting code, where the bulk and boundary degrees of freedom may be identified as logical and physical degrees of freedom respectively.
These models capture key features of entanglement in the AdS/CFT correspondence; in particular, the Ryu-Takayanagi formula and the negativity of tripartite information are obeyed exactly in many cases.
That bulk logical operators can be represented on multiple boundary regions mimics the Rindler-wedge reconstruction of boundary operators from bulk operators, realizing explicitly the quantum error-correcting features of AdS/CFT recently proposed in ~\cite{Almheiri14}. 
}

\maketitle

\section{Introduction}

The AdS/CFT correspondence, an exact duality between quantum gravity on a $(d{+}1)$-dimensional asymptotically-AdS space and a $d$-dimensional CFT defined on its boundary,
has significantly advanced our understanding of quantum gravity, as well as provided a powerful framework for studying strongly-coupled quantum field theories.  One aspect of this duality is a remarkable relationship between geometry and entanglement.  This notion first appeared in the proposal \cite{Maldacena03} that two entangled CFT's have a bulk dual connecting them through a wormhole, and was later quantified by Ryu and Takayanagi via their proposal that entanglement entropy in the CFT is computed by the area of a certain minimal surface in the bulk geometry \cite{Ryu06, Ryu06b}. This latter proposal, known as the Ryu-Takayanagi (RT) formula, has led to much further work on sharpening the connection between geometry and entanglement \cite{Hubeny07, Headrick07,Raamsdonk09, Raamsdonk10, Hayden13b, Lewkowycz13,Maldacena13, Lashkari14}

In the condensed matter physics community, improved understanding of quantum entanglement has led to significant progress in the numerical simulation of emergent phenomena in strongly-interacting systems. 
A key ingredient of such algorithms is the use of tensor networks to efficiently represent quantum many-body states~\cite{Vidal03b, Verstraete04b, Verstraete08}. 
Vidal combined this idea with entanglement renormalization to formulate the Multiscale Entanglement Renormalization Ansatz (MERA)~\cite{Vidal07, Vidal08}, a family of tensor networks that efficiently approximate wave functions with long-range entanglement of the type exhibited by ground states of local scale-invariant Hamiltonians~\cite{Evenbly09, Evenbly09b, Evenbly10}. 
The key idea is to represent entanglement at different length scales using tensors in a hierarchical array.

In the AdS/CFT correspondence, the emergent radial direction can be regarded as a renormalization scale \cite{Susskind1998}, and spatial slices have a hyperbolic geometry resembling the exponentially growing tensor networks of MERA.  This similarity between AdS/CFT and MERA was pointed out by Swingle, who argued that some physics of the AdS/CFT correspondence can be modeled by a MERA-like tensor network where quantum entanglement in the boundary theory is regarded as a building block for the emergent bulk geometry~\cite{Swingle12, Swingle12b}.

Recently it has been argued in ~\cite{Almheiri14} that the emergence of bulk locality in AdS/CFT can be usefully characterized in the language of quantum error-correcting codes.  
Certain paradoxical features of the correspondence arise naturally by interpreting bulk local operators as logical operators on certain subspaces of states in the CFT, whose entanglement structure protects these operators from boundary erasures.  Moreover, inspired by \cite{Swingle12, Swingle12b}, it was suggested that there should be tensor network models that concretely implement these ideas.  

In this paper, we propose such a family of exactly solvable toy models of the bulk/boundary correspondence based on a novel tensor-network construction of quantum error-correcting codes. Other authors have recently used holographic ideas~\cite{Yoshida2013,Latorre2015} and related tensor network constructions~\cite{Ferris14,Bacon14} to build quantum codes with interesting properties or toy models of the bulk/boundary correspondence~\cite{Qi13}, but our approach differs from previous work by combining the following properties, all of which are desirable for a model of AdS/CFT: 

\begin{itemize}
\item \tb{Exactly solvable:} Many of the properties of our models can be shown explicitly.  In particular, an exact prescription for mapping bulk operators to boundary operators can be obtained, and we can give examples where the Ryu-Takayanagi formula holds exactly for all connected boundary regions.
 
\item \tb{QECC:} Our models are quantum error-correcting codes, where the bulk/boundary legs of the tensor network correspond to input/outputs of an encoding quantum circuit. In this sense they realize explicitly the proposal of \cite{Almheiri14}. 

\item \tb{Bulk uniformity:} The tensor network is supported on a uniform tiling of a hyperbolic space, known as a hyperbolic tessellation. If the tiling is extended to an infinite system, the tensor network has no inherent directionality and all the locations in the bulk can be treated on an equal footing (see Fig.~\ref{fig:HolographicPentagonCode}). 
\end{itemize}

The rest of this paper is organized as follows: In section~\ref{sec:perfect}, we introduce a class of tensors called perfect tensors, which are associated with pure quantum states of many spins such that the entanglement is maximal across any partition of the spins into two sets of equal size. In section~\ref{sec:model}, we construct holographic states and codes by building networks of perfect tensors. These codes have properties reminiscent of the AdS/CFT correspondence, elucidated in the rest of the paper, where the code's logical/physical degrees of freedom are interpreted as the bulk/boundary degrees of freedom of a CFT with a gravitational dual. 

In section~\ref{sec:state}, we study the entanglement structure of holographic states, showing that the Ryu-Takayanagi formula is exactly satisfied for any connected boundary region, developing a graphical representation of multipartite entanglement, and confirming the negativity of tripartite information~\cite{Hayden13b}. In section~\ref{sec:code}, we investigate the dictionary relating bulk and boundary observables, define a lattice version of the causal wedge, and explain how bulk local operators in the causal wedge can be reconstructed on the boundary; we also define a lattice version of the entanglement wedge, and offer evidence supporting the  entanglement wedge hypothesis proposed in ~\cite{Headrick2014, Wall2012, Czech2012}, see also \cite{Jafferis2014}. We briefly discuss how to describe black holes using holographic codes in section~\ref{sec:black}. Section~\ref{sec:conclude} contains our conclusions, and many details appear in the appendices. 

\section{Isometries and perfect tensors}\label{sec:perfect}
In this section we review some tools which will be used in our constructions of holographic states and codes. 
We begin with a standard definition:
\begin{definition}
Say $\HA$ and $\HB$ are two Hilbert spaces, not necessarily of the same dimensionality.  An \textbf{isometry} from $\HA$ to $\HB$ is a linear map $T:\HA\mapsto\HB$ with the property that it preserves the inner product.
\end{definition}
If $\HA$ and $\HB$ have finite dimensionality, as we will assume throughout this paper, then it immediately follows that such a $T$ can exist only if their dimensionalities $\textrm{dim}(A)$ and $\textrm{dim}(B)$ obey $\textrm{dim}(A)\leq \textrm{dim}(B)$.  In the special case where $\textrm{dim}(A) = \textrm{dim}(B)$, $T$ is just a unitary transformation.  Clearly the composition of two isometries is also an isometry.

If $T:\HA\mapsto\HB$ is an isometry, then $T^\dagger T$ is the identity on $\HA$ and $TT^\dagger$ is a projector mapping $\HB$ to the range of $T$. We may represent the map $T$ as a two-index tensor acting as
\be
T:|a\ran \mapsto \sum_b |b\ran T_{ba},
\ee
where $\{|a\ran\}$ denotes a complete orthonormal basis for $\HA$ and $\{|b\ran\}$ for $\HB$.  Then $T$ is an isometry if and only if
\be\label{iso}
\sum_b T^{\dagger}_{a'b}T_{ba}=\delta_{a'a}.
\ee
We represent this graphically in figure \ref{isofig}, following the convention that operators are ordered from left to right, so that in the figure $T^\dagger$ is applied after $T$. We will call a tensor obeying \eqref{iso} an \textit{isometric tensor}.  

\bfig
\includegraphics[height=1.5cm]{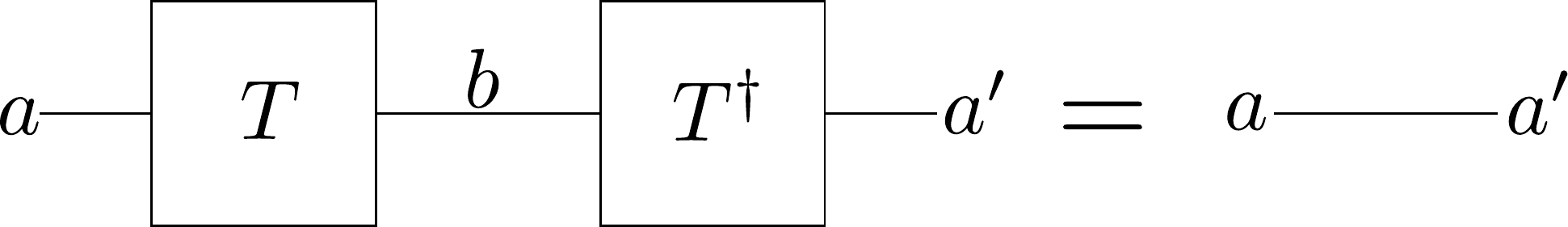}
\caption{Diagrammatic tensor notation, here showing that $T$ is an isometry.}\label{isofig}
\efig

Isometric tensors have the property that any operator $O$ acting on its ``incoming'' leg, can be replaced by an equal norm operator $O'$ acting on its ``outgoing'' leg, because
\be
TO = TOT^\dagger T = (TOT^\dagger)T \equiv O'T;
\ee
we illustrate this property in figure \ref{push}.
\bfig
\includegraphics[height=3cm]{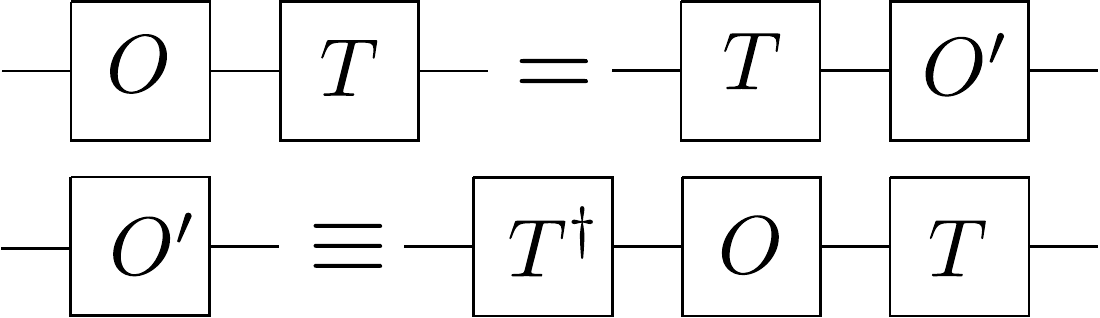}
\caption{Operator pushing through an isometric tensor.}\label{push}
\efig
This operation is essential for what follows, and we will often describe it as ``pushing an operator through a tensor''.  It is also easy to check a useful converse of operator pushing: If the two-index tensor $T$ has the property that any \textit{unitary} transformation $U$ contracted with its incoming index can be replaced by a corresponding \textit{unitary} transformation $U'$ contracted with its outgoing index (\textit{i.e.}, $TU = U'T$), then $T$ obeys \eqref{iso} up to a scalar factor, and therefore must be proportional to an isometric tensor.  

\bfig
\includegraphics[height=1.5cm]{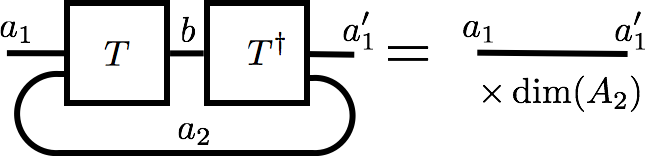}
\caption{If $\HA=\mathcal{H}_{A_2}\otimes \mathcal{H}_{A_1}$, then we can move one of the factors to the output while preserving the isometric structure.}\label{shrink}
\efig 
Another important property of isometric tensors is that if the input Hilbert space factorizes, we may reinterpret an input factor as an output factor while preserving \eqref{iso}, up to an overall rescaling.  
That is, if $T: \cH_{A_2}\otimes \cH_{A_1}\mapsto \cH_B $ is an isometric map, acting on a basis according to
\begin{align}
T:|a_2 a_1\ran \mapsto \sum_b |b\ran T_{ba_2a_1},
\end{align}
then $\tilde T:\mathcal{H}_{A_1}\mapsto\HB \otimes \mathcal{H}_{A_2}$ acting as
\begin{align}
\tilde T:|a_1\ran \mapsto \sum_{ba_2} |ba_2\ran T_{ba_2a_1}
\end{align}
obeys $\tilde T^\dagger \tilde T = \textrm{dim}(A_2) I_{A_1}$. We illustrate this property in figure \ref{shrink}.  

In this paper we will be interested in a special class of isometric tensors, which we will call perfect tensors. To formulate the concept of a perfect tensor, first note that we may divide the $m$ indices of a tensor $T_{a_1 a_2 \ldots a_m}$ into a set $A$ and a complementary set $A^c$. We use $|A|$ to denote the cardinality of the set $A$; hence $|A|+|A^c| = m$. Then $T$ may be regarded as a linear map from the span of the indices in $A$ to the span of the indices in $A^c$. We will usually assume that each index ranges over $v$ values, and we will use $A$ to denote both the set of $|A|$ indices and the corresponding vector space with dimension $v^{|A|}$; thus we say $T$ maps $A$ to $A^c$.
\begin{definition}
A $2n$-index tensor $T_{a_1a_2\ldots a_{2n}}$ is a \textbf{perfect tensor} if, for any bipartition of its indices into a set $A$ and complementary set $A^c$ with $|A|\leq |A^c|$, $T$ is proportional to an isometric tensor from $A$ to $A^c$.
\end{definition}
It is not obvious that nontrivial perfect tensors exist, but they do! Note that for $T$ to be perfect it suffices for $T$ to be a unitary transformation when $|A|=|A^c| = n$; in that case the property illustrated in figure \ref{shrink} ensures that $T$ is proportional to an isometric tensor for $|A|< n$. 
In Appendix \ref{App:PerfectTensorExamples} we describe perfect tensors explicitly for the case $n=3$, $v=2$ and for the case $n=2$, $v=3$; other cases with larger $n$ and $v$ are also discussed there. 
To keep our discussion concrete, we will focus especially on the six-index tensor for qubits ($v=2$)\footnote{This can be obtained from the encoding map of the 5-qubit code.}, but much of what we say applies to arbitrary $2n$-index perfect tensors. 

Perfect tensors are related to other notable ideas in quantum information theory.  In general, a tensor $T$ with $m$ indices, each ranging over $v$ values, describes a pure quantum state $|\psi\ran$ of $m$ $v$-dimensional spins, where, up to a normalization factor,
\be
|\psi\rangle = \sum_{a_1,a_2, \ldots, a_m} T_{a_1a_2 \ldots a_m} |a_1a_2 \ldots a_m\ran. 
\ee
A perfect tensor describes a pure state of $2n$ spins with a special property ---  any set of $n$ spins is maximally entangled with the complementary set of $n$ spins. Such states have been called \textit{absolutely maximally entangled} (AME) states \cite{Helwig2012,Helwig2013}.  
Conversely any $AME$ state defines a perfect tensor. Regarded as a linear map from one spin to $2n-1$ spins, a perfect tensor is the isometric encoding map of a quantum error-correcting code which encodes a single logical spin in a block of $2n-1$ physical spins, where the logical spin is protected against the erasure of any $n-1$ physical spins. 
Because $n$ is more than half of all the physical spins, this is the best possible protection against erasure errors compatible with the no-cloning principle.  In coding terminology this code has \textit{distance} $n$ and is denoted $[[m,k,d]]_v=[[2n-1,1,n]]_v$, where $m$ is the number of physical spins in the code block, $k$ is the number of protected logical spins, and $d$ is the code distance. 
This code is also the basis for a quantum-secret-sharing scheme called a $((n,2n-1))$ threshold scheme \cite{Cleve1999}; code states have the property that a party holding any $n-1$ spins has no information about the logical spin, while a party holding any $n$ spins has complete information about the logical spin (because erasure of the remaining $n-1$ spins is correctable).

\section{Construction of holographic quantum states and codes}\label{sec:model}

We have seen how tensors can be interpreted as quantum states or quantum codes. In this section we construct tensor networks in which the fundamental building blocks are perfect tensors. Our tensor networks describe states which we call \textit{holographic states}, and codes which we call \textit{holographic codes}.  

We shall focus on examples based on tilings of two-dimensional hyperbolic space, which are specific realizations of uniform hyperbolic tilings known as hyperbolic tessellations. These tilings have desirable symmetries for constructing a toy model of the AdS/CFT correspondence. In particular they are discretely scale-invariant, and there exist graph isomorphisms that bring any point in the graph to the center while preserving the local structure of the tiling.\footnote{Such transformations can be directly visualized using \emph{Kaleidotile} software~\cite{Kaleidotile}, which is freely available and has been of great aid in developing geometric intuition and producing figures of uniform hyperbolic tilings in this paper.} The machinery we develop may also be straightforwardly applied to non-uniform and higher-dimensional graphs. 

\begin{figure}
\centering
    \subfloat[Holographic hexagon state]{ 
  \includegraphics[width=0.4\linewidth]{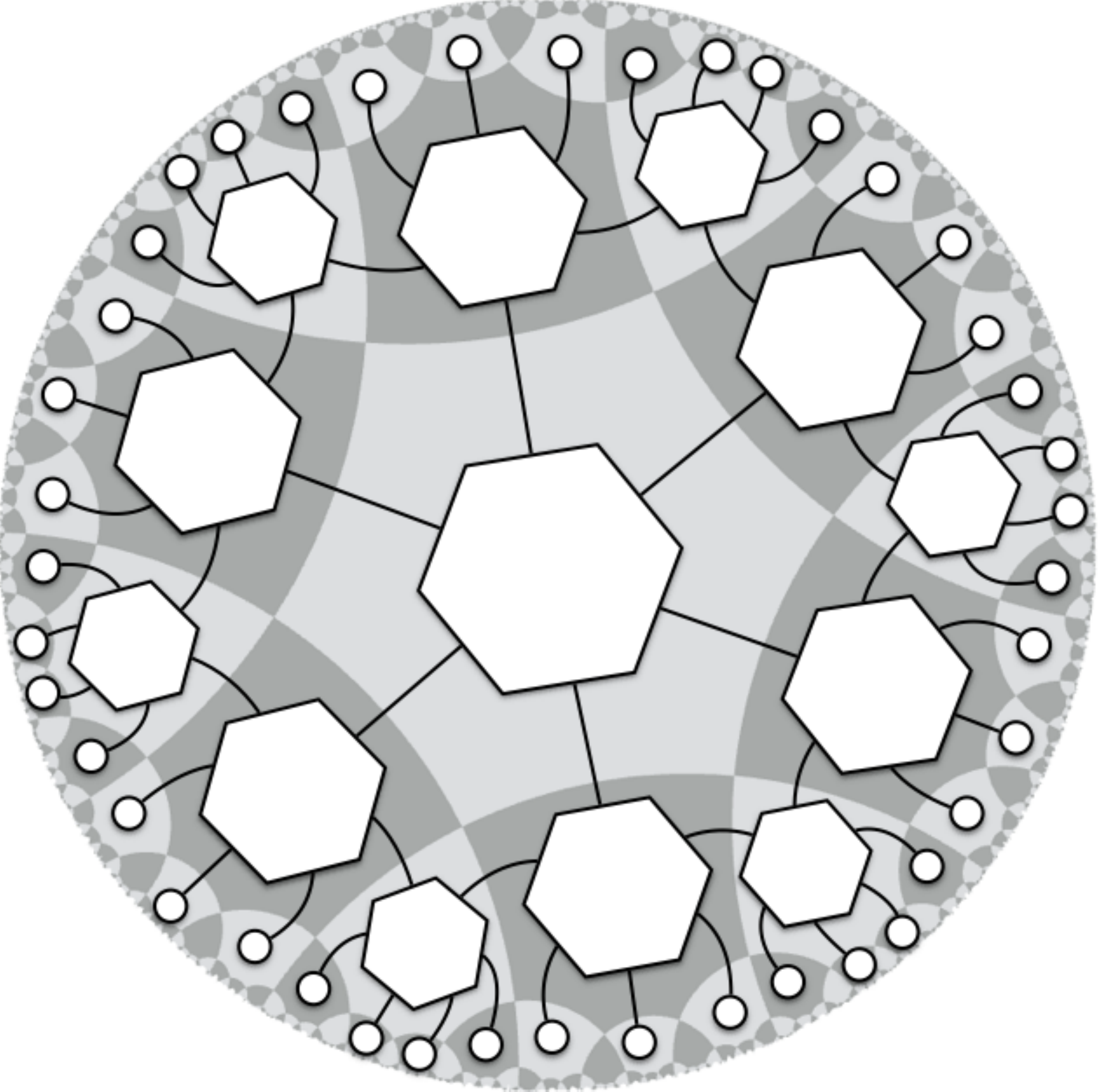}
   \label{fig:HolographicHexagonState}}
\hspace{1cm} 
    \subfloat[Holographic pentagon code]{ 
\includegraphics[width=0.4\linewidth]{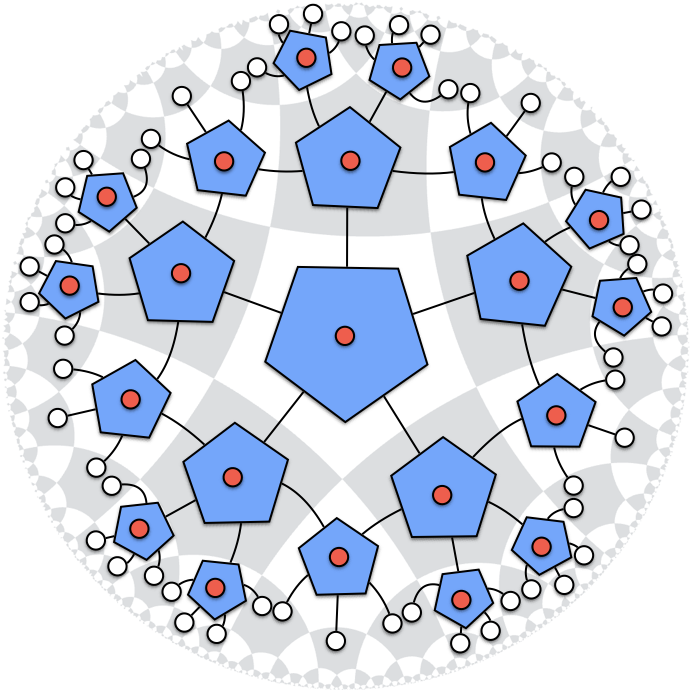}
   \label{fig:HolographicPentagonCode}}
\caption{White dots represent physical legs on the boundary. 
   Red dots represent logical input legs associated to each perfect tensor. 
}\label{fig:HolographicStateAndCode}
\end{figure}
Let's first consider a uniform tiling of a two-dimensional hyperbolic space by hexagons, with four hexagons adjacent at each vertex, as depicted in Fig~\ref{fig:HolographicHexagonState}. 
A perfect tensor with six legs is placed at each hexagon, and legs of perfect tensors are contracted with neighboring tensors at shared edges of the hexagons. 
We associate physical spins with the uncontracted open tensor legs on the boundary of the hyperbolic tiling; the tensor network corresponds to a pure state of these boundary spins, which we call a \emph{holographic state}. 
Note that perfect tensors are not necessarily symmetric under all the possible permutations of tensor legs, and thus we specify some particular ordering of tensor legs in the construction.  

We may similarly attach a state interpretation to more general networks constructed by contracting perfect tensors: 

\begin{definition}
Consider a tensor network composed of perfect tensors which cover some geometric manifold with boundary, where all the interior tensor legs are contracted. 
A \textbf{holographic state} is a state interpretation of such a tensor network, where physical degrees of freedom are associated with all uncontracted legs at the boundary of the manifold. 
\end{definition}

We now provide an example of a holographic quantum code.
 As in a holographic state, we consider a uniform tiling of the hyperbolic disc, this time by pentagons, with four pentagons adjacent at each vertex. 
A perfect tensor with six legs is placed at each pentagon, so that each tensor has one additional uncontracted open leg. 
This additional tensor leg is interpreted as a bulk index or logical input for the tensor network (see Fig.~\ref{fig:HolographicPentagonCode}). 
The entire system can be viewed as a big tensor with logical legs in the bulk and physical legs on the boundary. We then have the following theorem:
\begin{theorem}
The pentagon-tiling tensor network is an isometric tensor from the bulk to the boundary. 
We call it the holographic pentagon code.
\end{theorem}
We can prove this theorem by noting that if we order the tensors into layers labeled by increasing graph distance from the center, each tensor has at most two legs contracted with the tensors at the previous layer (this property is a consequence of the ``negative curvature'' of the graph).  Therefore, even if we regard the pentagon's bulk logical index as an input leg, the total number of input legs is at most three, and we may therefore regard each tensor as an isometry from input legs to output legs. 
Applying the perfect tensors layer by layer, and recalling that the product of isometries is an isometry, we obtain an isometry mapping all the logical indices in the bulk to the physical indices on the boundary.

We can view this isometry as the encoding transformation of a quantum error-correcting code, which we call a \emph{holographic code}. 
The number of logical $v$-dimensional spins is the number $N_{\rm bulk}$ of pentagons in the tiling, and the number of physical $v$-dimensional spins in the code block is the number $N_{\rm boundary}$ of uncontracted boundary indices in the tensor network. 
We show in Appendix \ref{App:CountingTensors} that the rate of the code, meaning the ratio of the number of logical spins to the number of physical spins, approaches
\begin{align}\label{eq:RatioCount}
\frac{N_{\rm bulk}}{N_{\rm boundary}}\to \frac{1}{\sqrt{5}}\approx .447
\end{align}
in the limit of a large number of layers.

This pentagon code was constructed by successively adding layers of tensors starting from the center and stopping after repeating this procedure a certain number of times (two layers in figure \ref{fig:HolographicPentagonCode}). 
 Alternatively, we may fill the bulk using a non-uniform cutoff, so that the graph distance between the ``center'' and the boundary varies from one portion of the boundary to another (as occurs in figure \ref{fig:HolographicHexagonState}). 
By exerting this freedom, we may change the corresponding value \ref{eq:RatioCount}  for the rate of the code and even slightly increase it.
 By varying the choice of perfect tensor and the shape of the cutoff, a large family of holographic codes can be constructed:

\begin{definition}
Consider a tensor network composed of perfect tensors which cover some geometric manifold with boundaries. The tensor network is called a \textbf{holographic code} if it gives rise to an isometric map from uncontracted bulk legs to uncontracted boundary legs.
\end{definition}

Tensor networks with open legs in the bulk were first proposed by Vidal~\cite{Vidal08}. 
More recently, Qi~\cite{Qi13} constructed a tensor-tree model with an exact unitary mapping between the bulk and the boundary. 
The most important difference between their models and ours is that their states are not protected against erasure of physical spins because the code rate is asymptotically unity. 
In addition our models are more symmetric; since perfect tensors can be interpreted as isometries along any direction, our models have no preferred direction in the bulk and all bulk sites are treated equally.
In particular, the pentagon code has the nice feature that, because the 6-leg perfect tensor we construct in appendix \ref{App:PerfectTensorExamples} is symmetric under cyclic permutations of five of the legs, which we take to be the contracted legs, the symmetry of the network is just the full symmetry of the graph.

\section{Entanglement structure of holographic states}\label{sec:state}
In this section we explore to what extent holographic states reproduce key properties of the AdS/CFT correspondence, such as the Ryu-Takayanagi formula for entropy of a boundary region \cite{Ryu06} and the negativity of tripartite information \cite{Hayden13b}.  

\subsection{Ryu-Takayanagi formula}
The Ryu-Takayanagi (RT) formula says that for a CFT whose gravitational dual is well-approximated by Einstein gravity at low energies, in any static state with a geometric bulk description the entropy $S_A$ of a boundary subregion $A$ at fixed time obeys
\be
S_A=\frac{\mathrm{Area}(\gamma_A)}{4G};
\ee
here $G$ is Newton's constant and $\gamma_A$ is the minimal-area codimension-two bulk surface whose boundary matches the boundary $\partial A$ of $A$.  
In our examples the bulk theory is $2+1$ dimensional, so $\gamma_A$ will be a spacelike bulk geodesic whose ``area'' is just defined as its length.  

In our discrete setting, we will define $\gamma_A$ as a certain \textit{cut} through the tensor network which partitions it into two disjoint sets of perfect tensors. 
Associated with a cut $c$ is a decomposition of the tensor network as a contraction of two tensors $P$ and $Q$, where the contracted legs lie along the cut; the number of contracted legs is called the \textit{length} of $c$, denoted $|c|$. If $A$ is a set of boundary legs and $A^c$ is the complementary set of boundary legs, then we say that the boundary of the cut $c$ matches the boundary of $A$ if the uncontracted legs of $P$ are the legs of $A$, and the uncontracted legs of $Q$ are the legs of $A^c$. 
The \textit{minimal bulk geodesic bounded by $A$}, $\gamma_A$, is then defined as the cut $c$ of shortest length whose boundary matches the boundary of $A$. 
We use $P$ to denote, not just the tensor associated with one side of the cut, but also the set of bulk lattice sites corresponding to the perfect tensors which are contracted to construct $P$; likewise for $Q$. 
We note that $P$ or $Q$ might have more than one connected component, and so might $\gamma_A$ when regarded as a path in the dual graph.

A standard argument for tensor network representations of quantum states shows that $|\gamma_A|$ provides an \textit{upper bound} on $S_A$.
If $P$ and $Q$ are the tensors associated with a cut $c$ whose boundary matches the boundary of $A$, then the holographic state $|\psi\ran$ may be expressed (up to normalization) as 
\be\label{eq:RQ-schmidt}
|\psi\ran=\sum_{a,b,i}|ab\ran P_{ai}Q_{bi} \equiv \sum_i |P_i\rangle_A\otimes |Q_i\rangle_{A^c}.
\ee
Here $a$ and $b$ run over complete bases for $A$ and $A^c$ respectively, and $i$ runs over all possible values of the indices contracted along $c$; the vectors $\{|P_i\rangle\}$ in $\HA$ and the vectors $\{|Q_i\rangle\}$ in $\mathcal{H}_{A^c}$ are not necessarily orthogonal or normalized. (See figure \ref{RQfig}.)
\bfig
\includegraphics[width=0.8\linewidth]{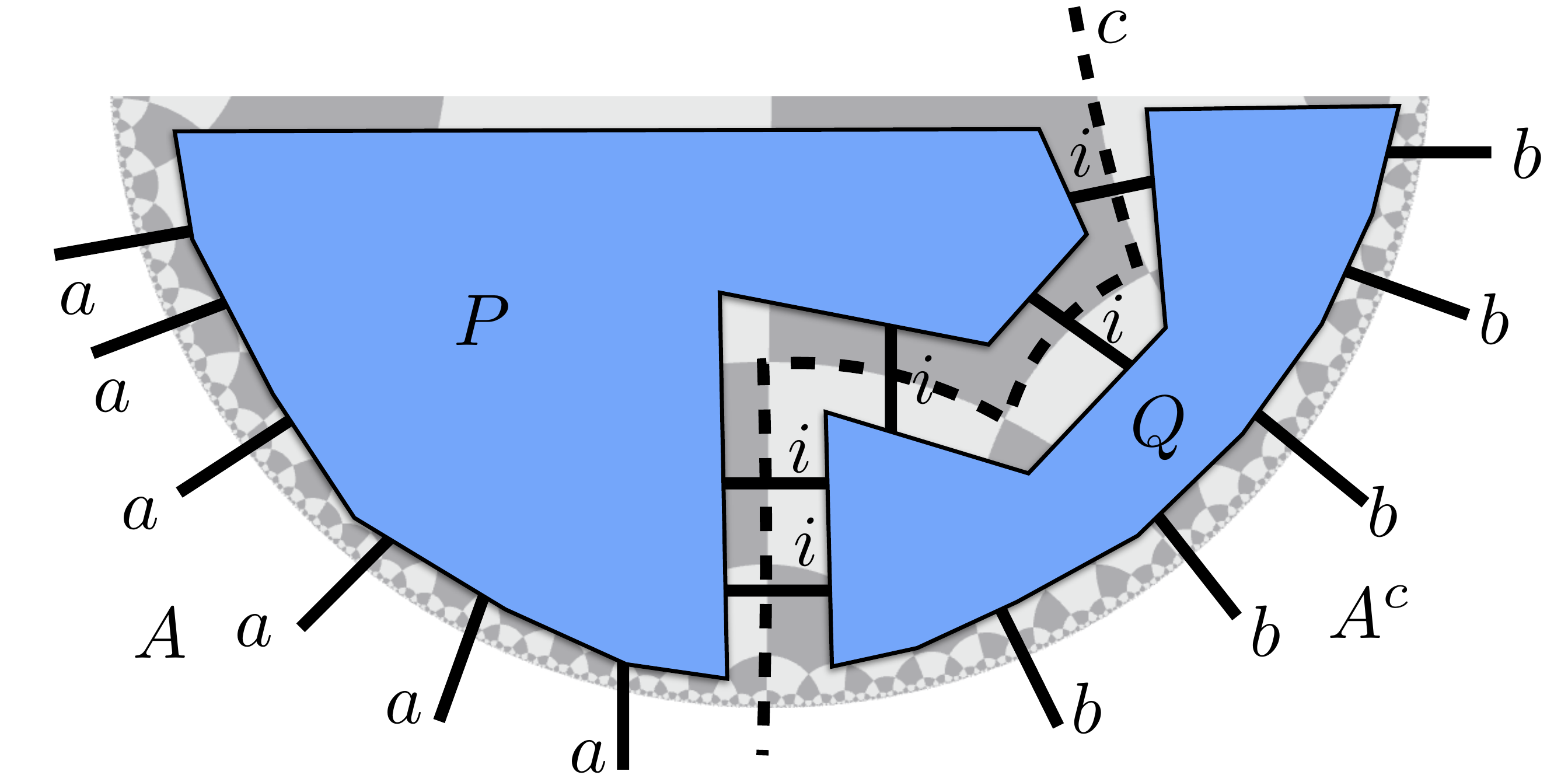}
\caption{A cut through a holographic tensor network by a curve $c$ bounded by $\partial A$. 
Boundary indices $a$ and $b$ are uncontracted in $A$ and its complement $A^c$ respectively; tensors $P$ and $Q$ are contracted by summing over the index $i$ which is cut by $c$. }\label{RQfig}
\efig
Tracing out $A^c$ we obtain (up to normalization) the density operator on $A$:
\be\label{rhoA}
\rho_A=\sum_{i,i'} \langle Q_{i'}|Q_i\rangle |P_i\rangle \langle P_{i'}|.
\ee
Evidently the rank of $\rho_A$ is at most the number of terms in the sum over $i$, namely $v^{|c|}$.  The density operator of a given rank with maximal Von Neumann entropy is proportional to the identity on its support, and has entropy equal to log of the rank. We obtain the best bound by choosing the cut $c=\gamma_A$ with the shortest length:
\be\label{upperbound}
S_A\leq |\gamma_A| \cdot \log v.
\ee
In most of what follows, we will define entropy by taking logs with base $v$, and so suppress the $\log v$ factor. 

If the tensors $P$ and $Q$ are actually isometries from $i$ to $a$ and $b$ respectively, then $\{|P_i\rangle\}$ and $\{|Q_i\rangle\}$ are sets of orthonormal vectors; in that case  \eqref{upperbound} is saturated and a discrete analogue of the RT formula holds exactly. Under what conditions will $P$ and $Q$ be isometries? We can prove the following theorem:
\begin{theorem}\label{RTth}
Suppose that we have a holographic state associated to a simply-connected planar tensor network of perfect tensors, whose graph has ``non-positive curvature''.\footnote{The scalar curvature of a graph is somewhat tricky to define in general; the condition we really need here is that the distance functional from one point on the dual network to another does not have interior local maxima. 
}   Then for any connected region $A$ on the boundary, we have $S_A=|\gamma_A|$; in other words, the lattice RT formula holds.
\end{theorem}
The strategy of the proof is to show that $P$ and $Q$ can in fact be interpreted as \textit{unitary} transformations, from the cut together with some subregion of $A$ or $A^c$ to the rest of $A$ or $A^c$ respectively.  
We can then use the identity depicted in figure \ref{shrink} to re-interpret these transformations as isometries from the cut to $A$ and from the cut to $A^c$ respectively; the RT formula follows.  
The key to the argument, explained in appendix  \ref{app:PlanarGraphProof}, is using a strengthened version of the max-flow min-cut theorem (which is standard in graph theory~\cite{Papadimitriou1998}) to establish that the tensor network representations of $P$ and $Q$ can be interpreted as unitary quantum circuits.  

\subsection{Bipartite entanglement of disconnected regions}\label{subsec:multiple-regions}
Unfortunately the proof of Theorem \ref{RTth} does not directly generalize to a disconnected region $A$, nor even to connected regions for states, such as our holographic code states, where not all perfect tensor indices are contracted in the bulk.  We do not consider this to be a serious problem for our models. 
However, we still find it worthwhile to introduce some machinery that allows us to quantify this presumption somewhat. 

The first technique we will introduce is an algorithmic procedure for constructing, given a boundary region $A$, a bulk curve $\gamma^\star_A$ bounded by $\partial A$ such that the corresponding tensor $P$ is guaranteed to be an isometry. For a holographic state the isometry $P$ maps $\gamma^\star_A$ to $A$, and for a holographic code $P$ maps $\gamma^\star_A$ and all incoming bulk indices of $P$ to $A$. Furthermore, $\gamma^\star_A$ is a \textit{local} minimum of the length, in the sense that no single tensor can be added to or removed from $P$ which reduces the length of the cut. 

The algorithm makes essential use of the properties of perfect tensors and is quite simple. 
We consider a sequence of cuts $\{c_\alpha\}$ each bounded by $\partial A$, and a corresponding sequence of isometries $\{P_\alpha\}$, such that each cut in the sequence is obtained from the previous one by a local move on the bulk lattice.  
The sequence begins with the trivial cut, $A$ itself; in each step we identify one perfect tensor which has at least half of its legs contracted with $P_\alpha$ and construct $P_{\alpha +1}$ by adding this perfect tensor to $P_\alpha$. 
Thus $P_{\alpha+1}$ is obtained by composing $P_\alpha$ with an isometry defined by a perfect tensor, and therefore $P_{\alpha+1}$ is an isometry if $P_\alpha$ is. 
The procedure halts when the cut reaches $\gamma^\star_A$ and no further local moves are possible. Though many different sequences of local moves are allowed, $\gamma^\star_A$ is well defined; tensors eligible for inclusion in $P_{\alpha+1}$ remain so as other tensors are included, so the output of the algorithm does not depend on the order of inclusion. Following standard computer science terminology, we call this procedure the \textit{greedy algorithm} and call $\gamma^\star_A$ the \textit{greedy geodesic}. A step of the greedy algorithm is illustrated in figure \ref{greedyfig}.
\bfig
\includegraphics[height=2cm]{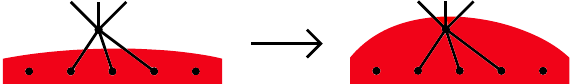}
\caption{A step in the greedy algorithm.  The upper node has at least three legs contracted with the region $P$, which we have shaded red, so we include it into $P$.}\label{greedyfig}
\efig

\bfig
\includegraphics[width=0.7\linewidth]{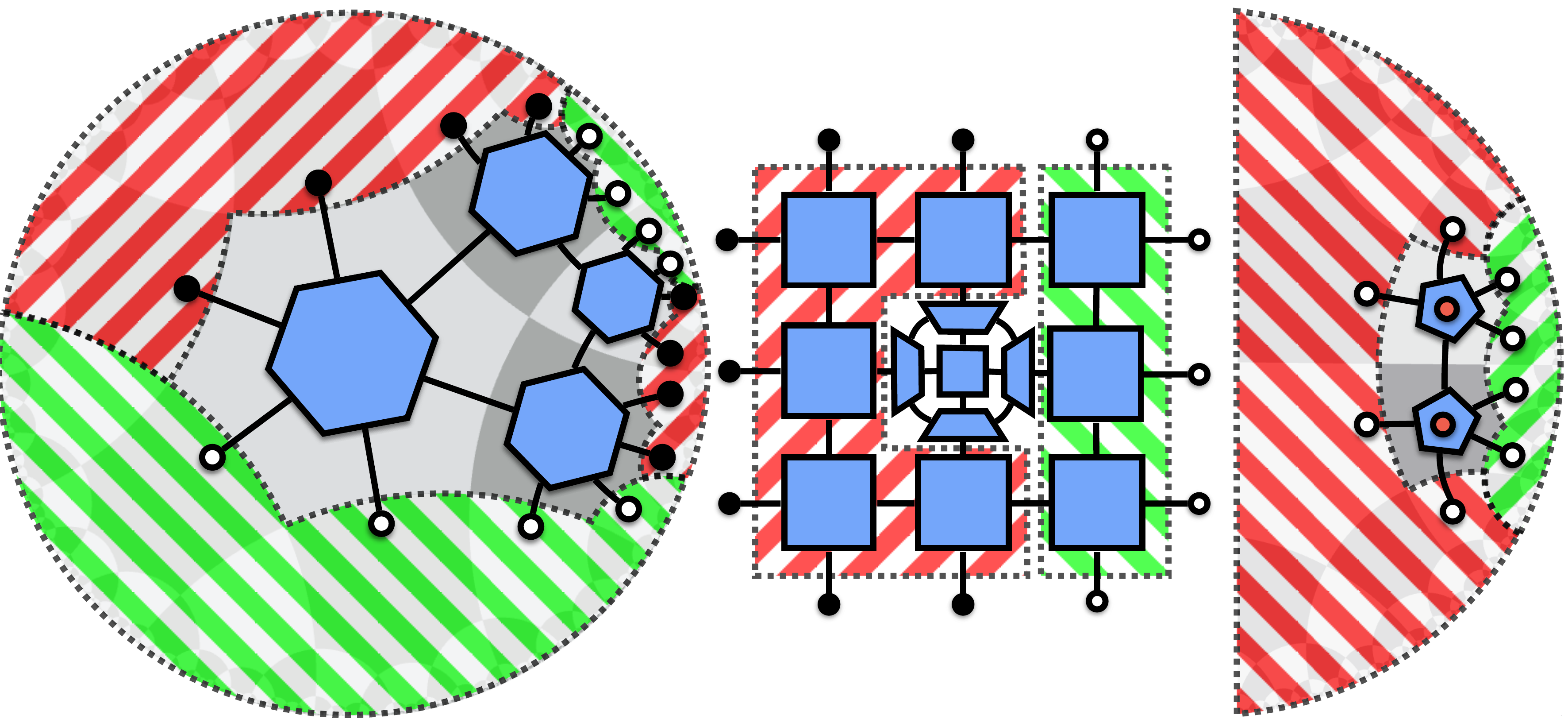}
\caption{Three examples where the greedy algorithm fails to find the matching minimal geodesics from complementary regions.  
The first example involves disconnected regions in the holographic state.
The second example involves a positive curvature obstruction at the center of the tiling which blocks the greedy geodesic from reaching the global minimal surface.
The third example involves a connected region for the holographic code.
In both the first and the third figure the greedy algorithm finds minimal geodesics from both sides but they do not match. 
In both cases, it is possible for the entropy to be slightly smaller than the length of the geodesic.
This depends on tensors which were not absorbed by either of the greedy geodesics which we call the bipartite residual regions.}
\label{fig:RTcounterexamplesd}
\efig

When the assumptions of Theorem \ref{RTth} are satisfied, the argument in appendix \ref{app:PlanarGraphProof} ensures that the greedy algorithm will find a true minimal geodesic $\gamma_A$. 
If there is more than one minimal geodesic, as is sometimes the case, then the greedy algorithm might continue past a minimal geodesic and proceed through minimal geodesics of equal length.  
In that case, the tensors in between the successive geodesics define a unitary transformation from one cut to the other. 
If $A$ has more than one connected component, if there is positive curvature, or if there are uncontracted bulk indices as for a holographic code, the greedy algorithm does not necessarily succeed in finding matching minimal geodesics, as we illustrate in figure \ref{fig:RTcounterexamplesd}.

In cases where the greedy algorithm fails to find a minimal geodesic, we can still use it to prove an interesting \textit{lower} bound on the entropy $S_A$.  Suppose that $\gamma^\star_A$ and $\gamma^\star_{A^c}$ are two greedy geodesics, produced by applying the greedy algorithm to $A$ and its complement $A^c$ respectively, where $P$ and $Q$ are the corresponding tensors.  Furthermore, suppose that $\gamma^\star_A\cap\gamma^\star_{A^c}$ is non-empty, in the sense that some links are cut by both geodesics.  We can represent that state as\footnote{For holographic codes with dangling bulk legs, we assume for now that a product state is fed into all bulk legs. If the input bulk state were entangled instead, there would be additional contributions to the boundary entanglement which we are not including.  This same proviso also applies to the discussion in the following subsection.}
\be
|\psi\ran=\sum_{a,b,i,j,k} |ab\ran P_{a,ij} Q_{b,ik}S_{jk}\equiv \sum_{i,j,k} S_{jk}|P_{ij}\rangle_A\otimes |Q_{ik}\rangle_{A^c}.
\ee
Here $i$ denotes the index shared between $\gamma^\star_A$ and $\gamma^\star_{A^c}$, $j$ is the index unique to $\gamma^\star_A$, $k$ is the index unique to $\gamma^\star_{A^c}$, and $S$ denotes the tensor that sits ``in between'' $\gamma^\star_A$ and $\gamma^\star_{A^c}$.  We call the set of lattice sites in $S$ the \textit{bipartite residual region} (where the modifier ``bipartite'' draws a distinction with the \textit{multipartite residual region} to be discussed in section \ref{subsec-multi-eng}.) Because $P$ and $Q$ are isometries, both $\{|P_{ij}\rangle\}$ and $\{|Q_{ik}\rangle\}$ are sets of orthonormal vectors. Therefore, the marginal density operator for $A$ is
\be
\rho_A=\sum_{i,j,j',k}S_{jk}S_{j'k}^*|P_{ij}\rangle\langle P_{ij'}|.
\ee
This density operator has support on the subspace of $A$ spanned by $\{|P_{ij}\rangle\}$, which has dimension $v^{|\gamma^\star_A|}$, and this subspace has a decomposition into subsystems $A_1\otimes A_2$ such that the basis element $|P_{ij}\rangle$ may be expressed as $|i\rangle_{A_1}\otimes |j\rangle_{A_2}$, where $\{|i\rangle\}$ and $\{|j\rangle\}$ are orthonormal bases for $A_1$ and $A_2$ respectively. We may then write 
\be\label{rhoA2}
\rho_A=\left(\sum_i|i\ran\lan i|_{A_1}\right)\otimes \left(\sum_{j,j',k}S_{jk}S^*_{j'k}|j\ran\lan j'|_{A_2}\right),
\ee
and from the additivity of the entropy, using $\textrm{dim}(A_1) = v^{|A_1|}=v^{|\gamma^\star_A\cap\gamma^\star_{A^c}|}$, we obtain the following theorem.

\begin{theorem}\label{thm-RT-inequality}
For a holographic state or code, if $A$ is a (not necessarily connected) boundary region and $A^c$ is its complement, then the entropy of $A$ satisfies
\be\label{residual-lowerbound}
S_A\geq |\gamma^\star_A\cap\gamma^\star_{A^c}|,
\ee
where $\gamma_A^\star$ is the greedy geodesic obtained by applying the greedy algorithm to $A$ and $\gamma_{A^c}^\star$ is the greedy geodesic obtained by applying the greedy algorithm to $A^c$.
\end{theorem}

We see from Theorem \ref{thm-RT-inequality} that violations of the Ryu-Takayanagi formula are closely related to the size of the bipartite residual region. In particular, if there is no bipartite residual region then $S_A=|\gamma^\star_A|$; the upper bound \eqref{upperbound} and the lower bound \eqref{residual-lowerbound} together imply that $\gamma^\star_A$ is in fact a minimal geodesic, and RT holds. We will argue in section \ref{subsec-multi-eng} that the bipartite residual region has size $O(1)$ when the regions $A$ and $A^c$ on the boundary have $O(1)$ connected components. In this sense, the corrections to the RT formula are typically small. 

\subsection{A map of multipartite entanglement}\label{subsec-multi-eng}

So far we have emphasized the bipartite entanglement between a boundary region $A$ and its complement $A^c$ in a holographic state or code. But we may also divide the boundary into three or more regions and investigate the structure of the entanglement among these regions. The entanglement structure can be elucidated via an entanglement ``distillation'' procedure which we will now describe.

To explain this procedure we begin by revisiting the case of bipartite entanglement. 
We have seen that if the conditions of Theorem \ref{RTth} are satisfied, then a holographic state can be expressed in the form \eqref{eq:RQ-schmidt}, where a subsystem of $A$ of dimension $v^{|\gamma_A^\star|}$ is maximally entangled with a corresponding subsystem of $A^c$. 
This entanglement shared between two systems is generally diluted, since each party may contain many more than $S_A$ spins. 
The entanglement would be more useful in a more concentrated form. 

The procedure for transforming dilute entanglement into concentrated entanglement, called \textit{entanglement distillation}, is particularly simple for a bipartite pure state like $|\psi\rangle$ in \eqref{eq:RQ-schmidt}. We choose $|\gamma_A^\star|$ specified spins in $A$ (the subsystem $A_1$ of $A$) and we choose $|\gamma_A^\star|$ spins in $A^c$ (the subsystem $A_1^c$ of $A^c$). Then we apply a unitary transformation $U_A$ acting on $A$ that transforms the basis states $\{|P_i\rangle_A\}$ to the standard basis states of $A_1$, and a unitary transformation $U_{A^c}$ acting on $A^c$ that transforms the basis states $\{|Q_i\rangle_{A^c}\}$ to the standard basis states of $A_1^c$, thus obtaining the state 
\be 
|\psi'\rangle = \left(|\Phi\rangle^{\otimes |\gamma_A^\star|} \right)_{A_1A^c_1}\otimes |\tilde \chi\rangle_{A_2}\otimes |\tilde \phi\rangle_{A_2^c},
\ee  
in which the entanglement of $A$ with $A^c$ now resides entirely in the system $A_1 A_1^c$. Here $A_2$ denotes the complement of $A_1$ in $A$, $A_2^c$ denotes the complement of $A_1^c$ in $A^c$, and 
\be
|\Phi\rangle =\frac{1}{\sqrt{v}} \sum_{\alpha=1}^v |\alpha\rangle \otimes |\alpha\rangle
\ee
is a maximally entangled EPR pair of two spins.

There is a method for constructing the unitary transformations $U_A$ and $U_{A^c}$ explicitly, which has a pleasing geometrical interpretation. The method uses the greedy algorithm for constructing $\gamma_A^\star$, but where now each local move, in which the cut through the tensor network advances into the bulk by moving past one additional tensor, is accompanied by a local unitary transformation that decouples spins from the network. This local unitary transformation is depicted in figure \ref{fig_distillation}, where entanglement distillation is performed on a pair of contracted six-leg tensors. 

\begin{figure}[htb!]
\centering
\includegraphics[width=0.8\linewidth]{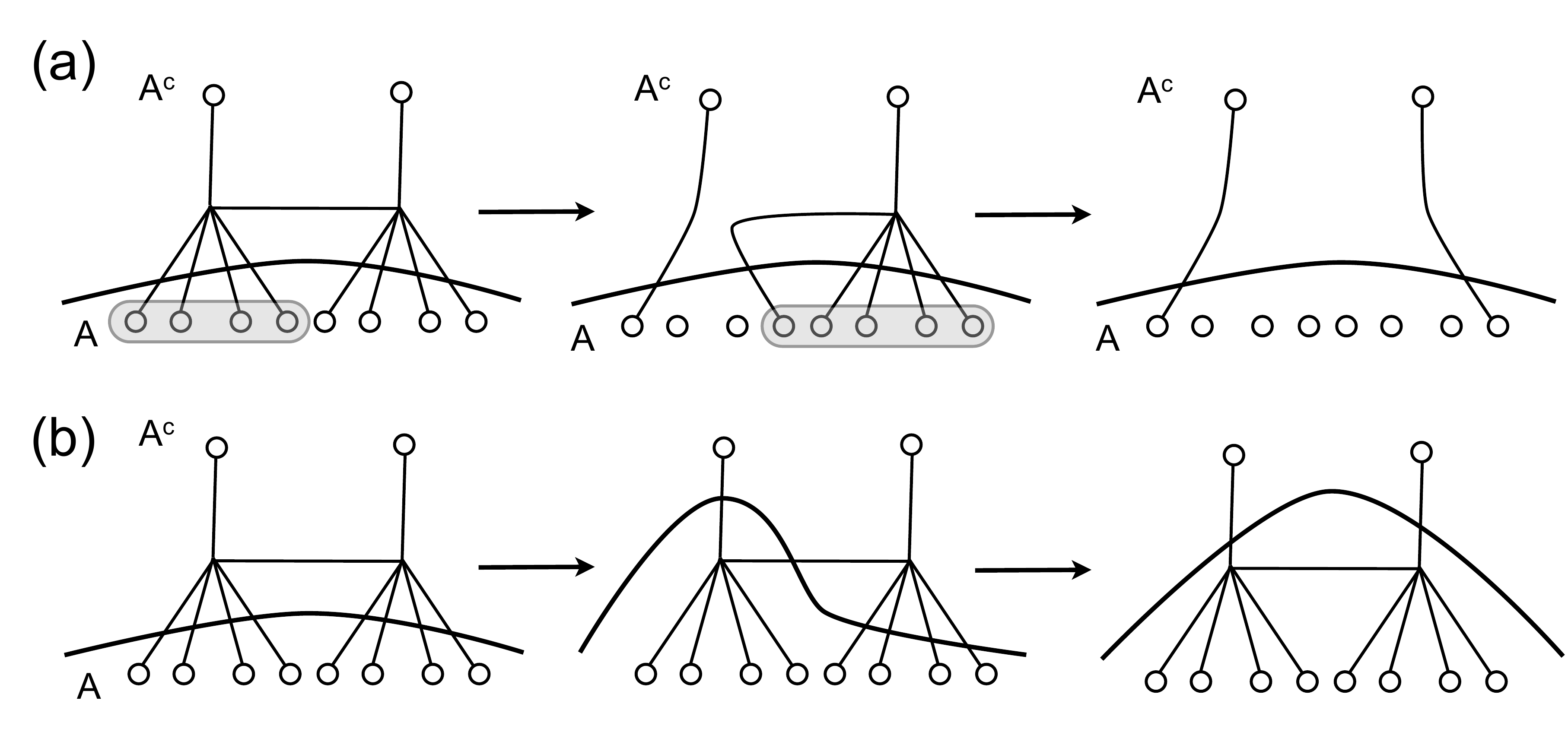}
\caption{ The correspondence between local moves and distillation of EPR pairs. (a) Distillation of two EPR pairs. (b) The corresponding local moves.
Before the first move, the tensor on the left has four legs crossed by the cut. Because the tensor is perfect, its remaining two legs are maximally entangled with a subsystem of these four. The first local unitary transformation acts on the four spins below the cut, transforming the basis to decouple the second and third spin, while the first and fourth spins remain contracted across the cut; in the corresponding local move, the cut advances upward past the tensor on the left. After the first move, the tensor on the right has five legs crossed by the cut. The second unitary transformation changes the basis of these five spins, decoupling the first four, while the fifth remains contracted across the cut; now the corresponding local move advances the cut upward past the tensor on the right. The product of the two local unitaries has distilled two EPR pairs which cross the cut, while decoupling six spins below the cut. 
} 
\label{fig_distillation}
\end{figure}

Since each local move of the greedy algorithm moves the cut past a tensor which initially has at least three legs crossed by the cut, the legs above the cut are always maximally entangled with the legs below, and the corresponding local unitary transformation exists. For purposes of visualization, we may imagine that the spins which remain contracted across the cut advance further into the bulk in each step, remaining adjacent to the cut, while the spins which decouple are left behind. When the greedy algorithm applied to $A$ terminates, then, all the decoupled spins of $A$ are distributed throughout the bulk region in between the greedy geodesic and the boundary, while $|\gamma_A^\star|$ spins of $A$, lined up along the greedy geodesic, are contracted with tensors on the other side of the greedy geodesic. If we also apply the greedy algorithm to $A^c$, then under the conditions of Theorem \ref{RTth}, the algorithm terminates at the same greedy geodesic. Acting together, then, the unitary transformations associated with the two greedy algorithms have decoupled all the boundary spins, except for $|\gamma_A^\star|$ EPR pairs, one for each of the legs crossed by the greedy geodesic, thus executing the entanglement distillation protocol.

Run backwards, the sequence of local unitary transformations associated with the greedy algorithm constitutes a \textit{holographic quantum circuit}, which prepares the boundary state. 
The input to this circuit is $|\gamma_A^\star|$ EPR pairs, plus a suitable number of additional spins in a product state, distributed throughout the bulk. The circuit builds the state step by step, gradually incorporating the bulk spins as the cut advances outward from the greedy geodesic toward the boundary. The input state, envisioned as a set of EPR pairs lined up along $\gamma_A^\star$, provides a \textit{map of entanglement}, a picture characterizing the structure of the entanglement between $A$ and $A^c$. (See figure \ref{fig_tiling2}.) The initial EPR pairs along the greedy geodesic which are deep inside the bulk encode long-range entanglement between $A$ and $A^c$, while the EPR pairs closer to the boundary encode shorter-range entanglement. 

\begin{figure}[htb!]
\centering
\includegraphics[width=0.60\linewidth]{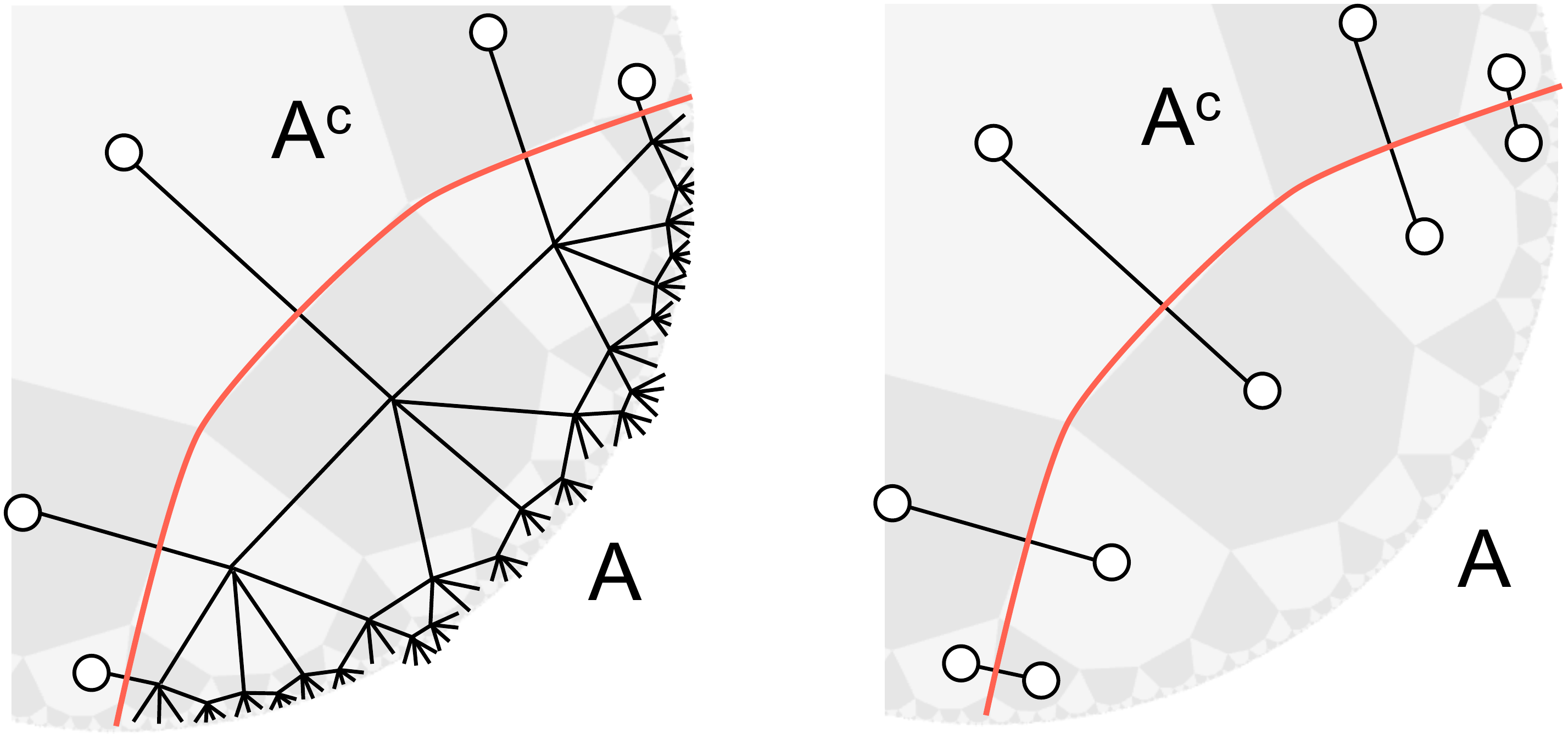}
\caption{A geometric map of bipartite entanglement. White dots represent physical spins distilled by applying local unitary transformations to $A$ and $A^c$. 
} 
\label{fig_tiling2}
\end{figure}

We can likewise use the greedy algorithm to create a map of multipartite entanglement, whether or not the conditions of Theorem 2 are satisfied. Suppose, for example, that we divide the boundary into four regions $A,B,C,D$, each of which is connected, as in figure \ref{fig_residual_tensor}. 
We may apply the greedy algorithm separately to each of the four regions, obtaining greedy geodesics $\gamma_A^\star, \gamma_B^\star,\gamma_C^\star, \gamma_D^\star$. The bulk region in between $A$ and its greedy geodesic $\gamma_A^\star$ is called the \textit{causal wedge} of $A$, denoted $\CA$. (The significance of the causal wedge in holographic codes will be discussed at length in section \ref{sec:code}.) 
As figure \ref{fig_residual_tensor} indicates, the union $\mathcal{C}[A]\cup \mathcal{C}[B]\cup \mathcal{C}[C]\cup \mathcal{C}[D]$ of the four causal wedges need not cover the entire bulk lattice --- there may be a \textit{multipartite residual region} in the bulk, which the greedy algorithm fails to reach when applied to the boundary regions one at a time. As we explain below, the size of the multipartite residual region is expected to be $O(1)$, independent of the total system size.

\begin{figure}[htb!]
\centering
\includegraphics[width=0.80\linewidth]{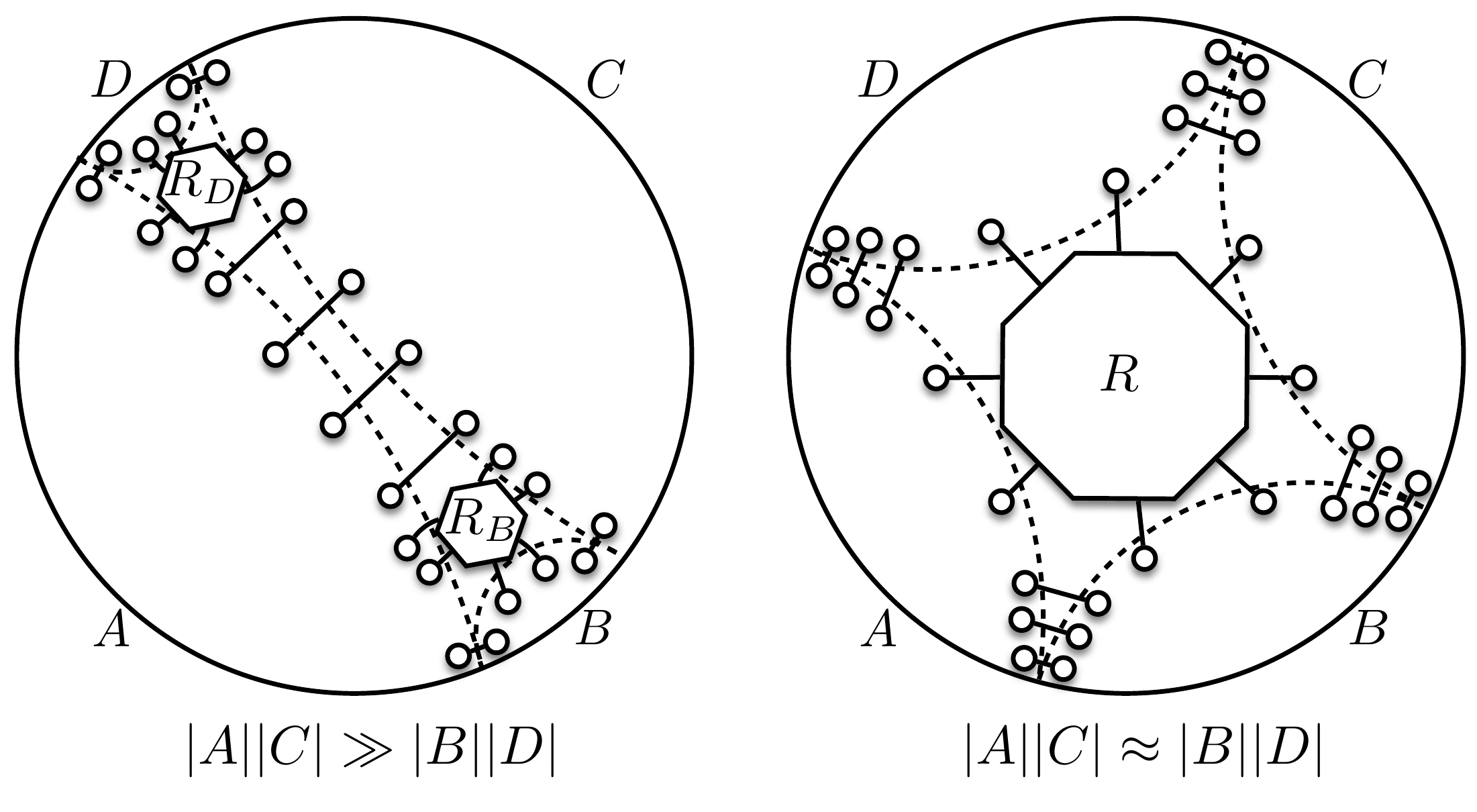}
\caption{The ``map of entanglement'' and multipartite residual regions in a holographic state. 
For $|A||C|  \gg |B||D|$ it is possible for the residual region to pinch off so much that EPR pairs can be directly distilled between $A$ and $C$.
In other words, due to the discretization of the lattice, the causal wedges $\cC[A]$ and $\cC[C]$ may be adjacent in the bulk.
In this case the residual region may be composed of two disconnected components, $R_B$ and $R_C$ which can contribute tripartite correlations.
A similar analysis holds for $|A||C|  \ll |B||D|$. 
For $|A||C| \approx |B||D|$, a single connected residual region $R$ contiguous to the four causal wedges is expected and may contribute four-party correlations.
} 
\label{fig_residual_tensor}
\end{figure}

Multipartite residual regions in the bulk can indicate multipartite entanglement among the four regions on the boundary. As discussed above for the case of bipartite entanglement, suppose we decouple spins in each of $A,B,C,D$ by performing suitable local unitary transformations associated with each step of the greedy algorithm. 
Where the greedy geodesics of adjacent regions meet, EPR pairs are distilled, in keeping with our observation in section \ref{subsec:multiple-regions} that the bipartite entanglement of two boundary regions is no less than the length of the greedy geodesic shared by the two regions. 
The tensors trapped inside a multipartite residual region however, do not necessarily have a decomposition into EPR pairs. 
Instead it describes a state with multipartite entanglement, which cannot be expressed as a product of states with only bipartite entanglement. 

Just as for a partition of the boundary into connected regions $A$ and $A^c$, we can reverse the order in which tensors are incorporated by the greedy algorithm to obtain a holographic quantum circuit of isometries which prepares the boundary state. 
When we partition the boundary into four connected regions, however, the input to the circuit includes more than just EPR pairs distributed along shared greedy geodesics and decoupled spins in the bulk; additional multipartite states associated with each connected component of the bulk multipartite residual region are also part of the input. The circuit factorizes into a product $U_A\otimes U_B\otimes U_C\otimes U_D$, with each of the four unitary transformations acting within its own causal wedge to build the corresponding connected component of the boundary. Again, the greedy geodesics encode a ``map'' of the entanglement among $A,B,C,D$, now including a description of multipartite entanglement among all the regions as well as bipartite entanglement among pairs of regions. Two such maps are shown in figure \ref{fig_residual_tensor}; in these cases a single six-leg tensor is trapped in each connected component of the bulk multipartite residual region, though in general a more complex tensor network could be trapped inside as indicated in figure \ref{fig:RTcounterexamplesd}.

We may also argue that if the bulk has constant negative curvature, then for any partition of the boundary into $O(1)$ connected components, the multipartite residual region is always $O(1)$ in size. This statement is true for the Riemannian geometry of the hyperbolic plane, but is merely heuristic because it disregards subtleties arising from the discrete lattice structure of the bulk. For a two-dimensional Riemannian manifold, the Gauss-Bonnet theorem applied to the residual region $R$ states that
\begin{align}\label{eq:gauss-bonnet}
\int_R K\;dA+\int_{\partial R}k_g\;ds=2\pi\chi(R).
\end{align}
here $K$ is the Gaussian curvature, $k_g$ is the geodesic curvature, and $\chi(R)$ is the Euler characteristic of the residual region, which is $\chi=1$ when $R$ has the topology of a disk. If $R$ is the interior of an $m$-gon whose sides are geodesics, \eqref{eq:gauss-bonnet} says that the integral of $K$ over $R$ is the deviation of the sum of interior angles of the $m$-gon from the corresponding sum for an $m$-gon in flat space; the latter sum is $(m-2)\pi$ because the $m$-gon can be covered by $m-2$ triangles. For the AdS space, the interior angles approach zero as the space becomes large compared to its curvature radius; therefore assuming uniform negative curvature $K=-1/\alpha^2$ (where $\alpha$ is the AdS radius), we conclude that the volume of the residual region is 
\begin{align}
V(R) = \pi(m-2)\alpha^2.
\end{align}
In our tensor networks $\alpha$ is of order the length of a link; therefore $V(R)$ is $O(1)$ in lattice units if $m$ is $O(1)$, which establishes our claim.

Likewise, the bipartite residual region arising from a partition of the boundary into two regions $A$ and $A^c$, discussed in section \ref{subsec:multiple-regions}, has size $O(1)$ if $A$ and $A^c$ both have $O(1)$ connected components. 
Indeed, the bipartite residual region is contained in the multipartite residual region found by applying the greedy algorithm separately to each connected component of $A$ and of $A^c$. 

\subsection{Negative tripartite information}

A useful characterization of multipartite entanglement is the tripartite information, defined as
\be
I_3(A,B,C)\equiv S_A+S_B+S_C-S_{AB}-S_{AC}-S_{BC}+S_{ABC}.
\ee
For a general (mixed) tripartite quantum state, $I_3$ can take any real value. 
It is zero, though, for any tripartite pure state of $ABC$, since in that case $S_{ABC}=0$ and \textit{e.g.} regions $A$ and $BC$, being complementary, have the same entropy and therefore make cancelling contributions to $I_3$. 
Nor is there a contribution to $I_3$ from EPR pairs shared between a pair of the three regions (because \textit{e.g.} a pair shared by $AB$ yields positive contributions to $S_A$ and $S_B$ which are cancelled by negative contributions from $-S_{AC}$ and $-S_{BC}$) or from entanglement shared between one of the three regions and a fourth disjoint region. 
Thus, for a holographic state and for any partition of the boundary into four regions $A,B,C,D,$ nonzero contributions to $I_3(A,B,C)$ can arise only from the distilled multipartite states trapped in residual regions. 

In the holographic setting, it has been shown that $I_3 \leq 0$ follows from the RT formula \cite{Hayden13b}. 
For holographic states and codes, the non-positivity of $I_3$ is not ensured in general, because of the potential (small) violations of the RT formula. In some special cases, though, RT holds exactly, and the non-positivity of $I_3$ then follows. For example, suppose that we partition the boundary into four connected regions $A,B,C,D$, and that each connected component of the multipartite residual region traps just one perfect tensor. In that case there is no bipartite residual region, so Theorem \ref{thm-RT-inequality} implies that RT is exact and therefore $I_3 \leq 0$. 
To see that there is no bipartite residual region in this case, consider the bipartite partition of the boundary into the two disconnected regions $AC$ and $BD$, and consider an isolated $2n$-index perfect tensor surrounded by three or all four of the greedy geodesics $\gamma_A^\star, \gamma_B^\star,\gamma_C^\star,\gamma_D^\star$. 
This tensor must have at least $n$ legs crossing either $\gamma_A^\star \cup \gamma_C^\star$ or $\gamma_B^\star \cup \gamma_D^\star$. 
Therefore, when we apply the greedy algorithm to the boundary regions $AC$ and $BD$, one cut or the other will advance past this isolated tensor, excluding it from the bipartite residual region.

Under suitable conditions we can actually prove a stronger result  --- that $I_3$ is \textit{strictly negative}. Let us say that a connected component of the multipartite residual region is \textit{three sided} if surrounded by three of the four greedy geodesics, and \textit{four sided} if surrounded by all four greedy geodesics. Three-sided components make no contribution to $I_3$; if the three surrounding greedy geodesics are those of $X,Y,Z$, the symmetry of $I_3$ implies $I_3(A,B,C)= I_3(X,Y,Z)$, which vanishes for any pure state of $XYZ$. But an isolated $2n$-index perfect tensor which crosses all four greedy geodesics makes a negative contribution to $I_3$:

\begin{theorem}\label{lem:negativeTripartiteInformation}
Suppose the $2n$ indices of a perfect tensor state are partitioned into four disjoint nonempty sets $A,B,C,D$ such that $0 < |A|,|B|,|C|,|D| < n$. Then the tripartite information $I_3$ is strictly negative: $I_3(A,B,C) < 0$.
\end{theorem}
\begin{proof}
First we notice that for a four-part \textit{pure} state $ABCD$, the tripartite information $I(A,B,C)$ is actually completely symmetric under permutations of the four subsystems, which we can see by using the property that complementary regions have the same entropy in a pure state:
\begin{align}
I_3(A,B,C) & = S_A + S_B + S_C - S_{AB} - S_{BC} -S_{AC} + S_{ABC}\\
            & = S_A + S_B + S_C + S_D-\frac{1}{2}( S_{AB} + S_{CD} + S_{BC} + S_{AD} + S_{AC} +S_{BD} ).
\end{align}
We may therefore assume without loss of generality that $|A| \leq |B| \leq |C| \leq |D|$ which implies $|AB| \leq |CD|$ and $|AC| \leq |BD|$. Now we use the defining property of $2n$-index perfect tensors, that a set of $n$ or fewer indices is maximally entangled with its complement, which implies $S_X = \min(|X|, 2n-|X|)$, with entropy expressed in units of $\log v$. Therefore $S_A=|A|, S_B=|B|, S_C=|C|,S_D=|D|$, and furthermore $S_{AB} = |AB|$ and $S_{AC} = |AC|$. Now we distinguish two cases.
If $|AD| \leq |BC|$, then $S_{BC} = S_{AD}= |AD|$ and we have
\begin{align}
I_3(A,B,C) = |A| + |B| + |C| +|D| - |AB| - |AC| - |AD| = -2|A| < 0.
\end{align}
If on the other hand $|BC| \leq |AD|$, then $S_{AD} = S_{BC} = |BC|$ and we have
\begin{align}
I_3(A,B,C) = |A| + |B| + |C| +|D| - |AB| - |AC| - |BC| = 2|D| - 2n < 0,
\end{align}
where to obtain the second equality we use $|AB| + |AC| + |BC| = 2(|A| + |B| + |C|) = 2(2n-|D|)$. This completes the proof.
\end{proof}

For a holographic state with boundary partitioned into sets $A,B,C,D$, the conditions of Theorem \ref{lem:negativeTripartiteInformation} are satisfied by an isolated perfect tensor trapped inside a four-sided component of the multipartite residual region; fewer than $n$ of the tensor's legs cross any greedy geodesic, because otherwise the greedy algorithm would have moved the cut forward past this perfect tensor, which therefore would not be in the multipartite residual region. Furthermore, since entropy is additive for a product state, $I_3$ is also strictly negative for any product of perfect tensor states shared by $A,B,C,D$, provided that at least one factor has support on all four sets. Since only the four-sided regions contribute to $I_3$, we conclude that $I_3$ is strictly negative if the multipartite residual region contains at least one four-sided connected component, and if  each four-sided connected component contains only one perfect tensor.

\section{Quantum error correction in holographic codes}\label{sec:code}
In this section we study the error correction properties of our holographic codes in more detail.  The idea that a CFT with a gravity dual must have error correcting properties was recently proposed in \cite{Almheiri14}, and in this section we will see that our holographic codes illustrate many aspects of the proposal of \cite{Almheiri14} quite explicitly.  

\subsection{AdS-Rindler reconstruction as error correction}
We begin by briefly recalling the main point emphasized in \cite{Almheiri14}, which is that in AdS/CFT a bulk local observable can be realized by many different operators in the CFT. In fact, if $x$ is any point in the bulk, and $Y$ is any point on the boundary, the AdS/CFT dictionary can be chosen so that it maps the bulk local field $\phi(x)$ to a CFT operator $\mathcal{O}[\phi(x)]$ which has no support in an open set containing $Y$, and therefore commutes with any local field of the CFT supported near $Y$. Since $Y$ is an arbitrary boundary point, if the CFT operator corresponding to $\phi(x)$ were actually unique, we would conclude that $\mathcal{O}$ commutes with all local fields in the CFT, and therefore is a multiple of the identity because the local field algebra is irreducible. This paradox is evaded once we recognize that the correspondence is not unique. If $Y,Z$ are two distinct boundary points, the CFT operator corresponding to $\phi(x)$ can be chosen to be either $\mathcal{O}$, which commutes with CFT local fields supported near $Y$, or $\mathcal{O}'$, which commutes with CFT local fields supported near $Z$, where $\mathcal{O}$ and $\mathcal{O}'$ are inequivalent CFT operators even though they can be used interchangeably for describing bulk physics. 

\bfig
\includegraphics[height=7cm]{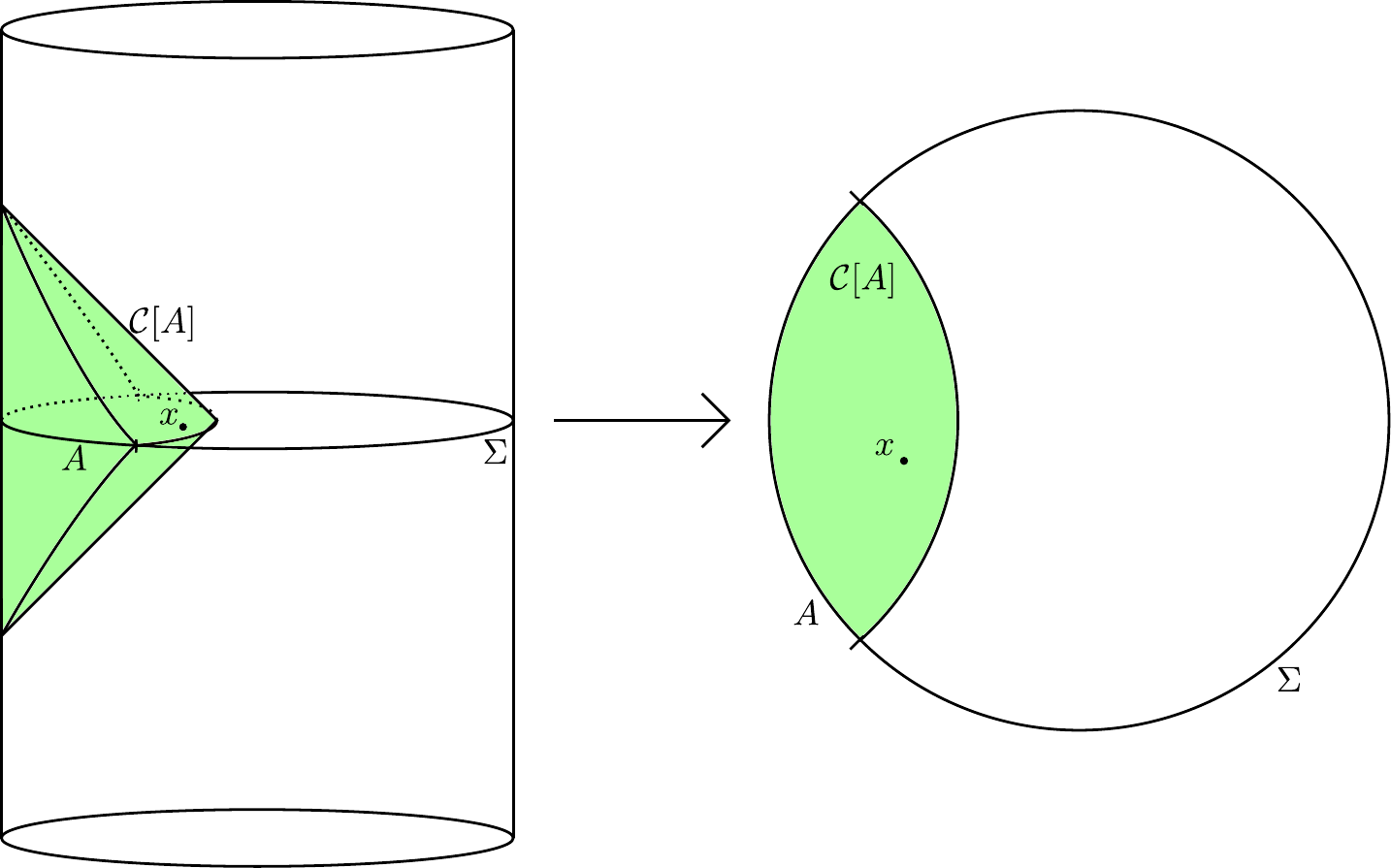}
\caption{Bulk field reconstruction in the causal wedge.  On the left is a spacetime diagram, showing the full spacetime extent of the causal wedge $\CA$ associated with a boundary subregion $A$ that lies within a boundary time slice $\Sigma$. The point $x$ lies within $\CA$ and thus any operator at $x$ can be reconstructed on $A$.  On the right is a bulk time slice containing $x$ and $\Sigma$, which has a geometry similar to that of our tensor networks. The point $x$ can simultaneously lie in distinct causal wedges, so $\phi(x)$ has multiple representations in the CFT.}\label{rindlerfig}
\efig
This novel feature of AdS/CFT, that a bulk local observable can be represented by boundary CFT operators in multiple ways, is illustrated in figure \ref{rindlerfig}.
The idea is that any fixed-time CFT subregion $A$ defines a subregion in the bulk, the causal wedge $\CA$.  For any point $x\in \CA$, bulk quantum field theory ensures that any bulk local operator $\phi(x)$ can be represented in the CFT as some nonlocal operator on $A$.  This representation is called the AdS-Rindler reconstruction of the operator \cite{Hamilton2006, Morrison2014}. 
Because a given bulk point $x$ can lie within distinct causal wedges associated with different boundary regions, the bulk operator $\phi(x)$ can have distinct representations in the CFT with different spatial support.  

In \cite{Almheiri14} the non-uniqueness of the CFT operator corresponding to the bulk operator $\phi(x)$ was interpreted as indicating that $\phi(x)$ is a logical operator preserving a code subspace of the Hilbert space of the CFT. This code subspace is protected against ``errors'' in which parts of the boundary are ``erased.'' If the boundary operator corresponding to $\phi(x)$ acts on a subsystem of the CFT which is protected against erasure of the boundary region $A^c$, then this operator can be represented in the CFT as an operator supported on $A$, the complement of the erased region. Thus we may interpret the AdS-Rindler reconstruction of $\phi(x)$ on boundary region $A$ as correcting for the erasure of $A^c$; choosing the erased portion of the boundary in different ways leads to different reconstructions of $\phi(x)$. Moreover, operators near the center of the bulk are ``well protected'' in the sense that a large region needs to be erased to prevent their reconstruction, while operators near the boundary can be erased more easily by removing a smaller part of the boundary \cite{Almheiri14}.

We may think of this code subspace as the low-energy sector of the CFT corresponding to a relatively smooth dual classical geometry. All CFT operators are physical, and thus have some bulk interpretation, but the ``logical'' operators are special ones which map low-energy states to other low-energy states. The same logical action can be realized by distinct CFT operators, as these distinct operators act on high-energy CFT states (those outside the code subspace) differently even though they act on low-energy states in the same way. 

\subsection{The physical interpretation of holographic codes}

The error-correcting properties of the AdS/CFT correspondence were motivated in \cite{Almheiri14} by bulk calculations, together with plausibility arguments regarding the CFT.  Our central observation in this paper is that analogous statements are provably true in holographic codes.

We emphasize that in holographic codes the uncontracted bulk legs hanging from each tensor should \textit{not} be thought of as tensor factors in addition to the boundary legs.  Rather the entire physical Hilbert space is spanned by states of the boundary legs only. The bulk legs just provide a way of conveniently describing states in a certain code subspace of this boundary Hilbert space, obtained by feeding states of the bulk legs through the isometry defined by the entire tensor network; this code subspace can be regarded as a simplified model of the low-energy states in a CFT.  

Likewise, operators acting on the dangling bulk indices correspond to nonlocal operators in the boundary theory whose algebra and action on the code subspace resembles what we would expect for the CFT description of how bulk local operators act on low-energy CFT states. When we speak of a ``bulk local operator'' we really mean the nonlocal boundary operator obtained by pushing an operator acting on a dangling bulk index out to the boundary using the isometry defined by the network.

\subsection{Bulk reconstruction from tensor pushing}
We now explain how holographic codes realize the AdS-Rindler reconstruction of figure \ref{rindlerfig}.  
The basic idea is that, instead of using the full isometry of the entire network to push a local bulk operator to the boundary, we can instead successively push it through individual perfect tensors in a manner of our choosing by using the operation of figure \ref{push}. 
We illustrate the reconstruction for two different bulk points of the pentagon code in figure \ref{pentagonpushfig}. Here we use the defining property of perfect tensors --- that the tensor provides a unitary transformation which maps any three legs of the tensor to the complementary set of three legs, and therefore also an isometry mapping any set of three or fewer ``incoming'' legs to any disjoint set of three ``outgoing'' legs. In figure \ref{pentagonpushfig}, each bulk vertex with arrows showing incoming and outgoing directions indicates such an isometry, and the complete set of blue  legs is a product of such isometries, and hence also an isometry. The blue operator on the boundary is obtained by conjugating the blue bulk operator by the blue isometry, and the same applies to the green bulk and boundary operators. In the construction of the isometry, we regard the dangling bulk index on each tensor as an incoming index, and therefore require that no more than two contracted indices are incoming for each blue (or green) tensor. The same blue isometry, then, can be used to push not just the central blue bulk index to the boundary, but also any of the other incoming bulk indices (which are not shown in the figure) on blue tensors. 

\bfig
\includegraphics[height=8cm]{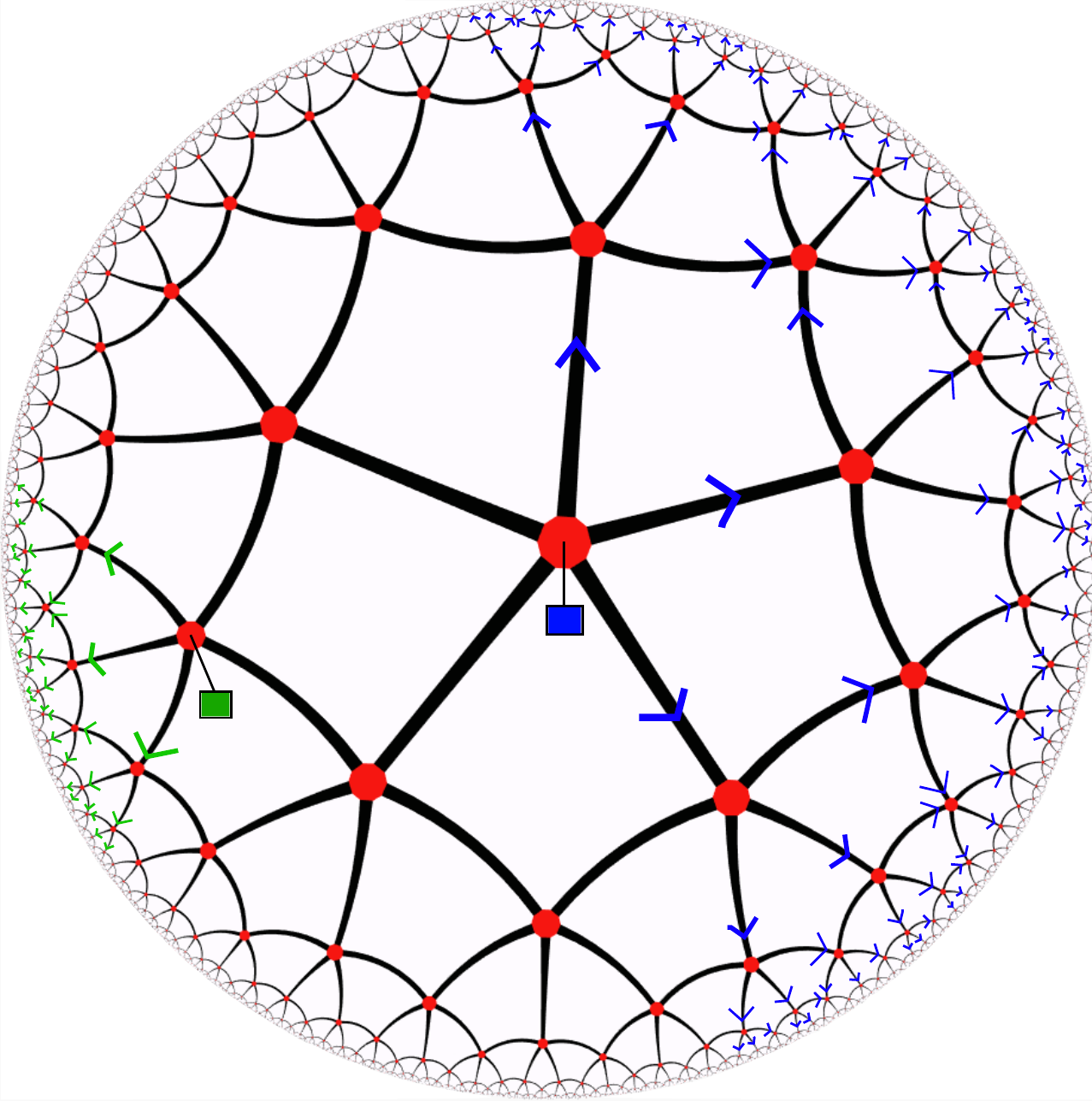}
\caption{Boundary reconstruction of bulk operators.  The blue operator on the central bulk leg is pushed to an operator supported on a fairly large boundary region, while the green bulk operator further from the center is pushed to an operator supported on a smaller boundary region.  Bulk legs for the other tensors are not shown.}\label{pentagonpushfig}
\efig

The boundary operation corresponding to a given bulk local operator manifestly has the non-uniqueness we described in our discussion of the AdS-Rindler reconstruction.  For example, we could move one of the three blue arrows directed outward from the central blue vertex to a different edge, thus reconstructing the central bulk operator on a different boundary region, or we could have sent the green arrows in the opposite direction and reconstructed the green bulk operator on a considerably larger boundary region on the opposite side. No matter which reconstruction we use, the boundary operator is obtained from the isometric embedding of the bulk indices into the code subspace of the boundary Hilbert space, and therefore each reconstructed operator corresponding to a given bulk operator acts on the code subspace in the same way.

In the theory of quantum error-correcting codes, we say that an error is an \textit{erasure} (or equivalently a \textit{located error}) if the set of spins damaged by the error is known, so this information can be used in recovering from the error. Holographic codes also provide protection against errors which act at unknown locations on the boundary, but for the purpose of developing the analogy with the AdS/CFT correspondence we will focus on protecting against erasure. A logical system can be protected against erasure of a set of spins in the code block if the full algebra of logical operations has a realization supported on the complementary set of unerased spins. In AdS/CFT we might only require reconstruction of a subalgebra of the full logical algebra; for example, the pentagon code provides better protection for the degrees of freedom deep within the bulk than for those closer to the boundary. The framework in which a quantum code protects only a subalgebra of the code's full logical algebra has been called \textit{operator algebra quantum error correction} \cite{Kribs2005, Kribs2006, Beny2007, Beny2007a}.
  
\subsection{Connected reconstruction and the causal wedge}

Given a subregion $A$ of the boundary, which bulk local operators can be reconstructed on $A$? This is not an easy question to answer in general, but at least we can give a simple description of a large logical subsystem reconstructable on $A$, namely those logical operators acting on bulk sites which are reachable using the greedy algorithm explained in section \ref{sec:state}.

Recall that the greedy algorithm associates with any boundary region $A$ a greedy geodesic $\gamma_A^\star$ whose boundary matches the boundary of $A$, such that $A$ and $\gamma_A^\star$ enclose a tensor $P_A$ which defines an isometry mapping free bulk legs in $P_A$ together with the legs cut by $\gamma^\star_A$ to $A$. Using this isometry applied to any operator acting on a bulk leg in $P_A$(tensored with the identity acting on all the rest of the isometry's input indices), we may push that logical operator through the isometry to obtain its reconstruction on $A$. Let's call the position of a perfect tensor in the network a \textit{bulk point} and say that the greedy algorithm \textit{reaches} a bulk point if it moves the cut past that tensor, hence using it in the construction of $P_A$

This operator reconstruction procedure can be applied to any boundary region $A$. In the special case where $A$ is connected, it provides a precise analog of the AdS-Rindler reconstruction in holographic codes, which we can formalize with a definition and theorem:
\begin{definition}
Suppose that $A$ is a {\bf connected} boundary region.  The \textbf{causal wedge of $A$}, denoted $\CA$, is the set of bulk points reached by applying the greedy algorithm to $A$.  
\end{definition}
\noindent We then have:
\begin{theorem}\label{theorem:causal-wedge}
Suppose $A$ is a connected boundary region.  Then any bulk local operator in the causal wedge $\CA$ can be reconstructed as a boundary operator supported on $A$.  
\end{theorem}

\noindent We could have formulated a geometric notion of the causal wedge, defining it as the set of bulk points enclosed between $A$ and the actual minimal geodesic $\gamma_A$, rather than the greedy geodesic. This geometrical definition is closer in spirit to how the term ``causal wedge'' has been used in the context of AdS/CFT. But we prefer this greedy notion of causal wedge instead, so that Theorem \ref{theorem:causal-wedge} is correct as stated. 

As figure \ref{pentagonpushfig} illustrates, bulk operators near the boundary can be reconstructed on smaller connected regions than bulk operators near the center, just as for the AdS-Rindler reconstruction in AdS/CFT. 
It is natural to wonder how large the connected region $A$ should be for the operator at the center of the bulk to be reconstructable on $A$. 
This question is studied for the pentagon code in appendix \ref{App:CountingTensors} by investigating whether the greedy algorithm applied to $A$ reaches the central tensor in the network. We find that a connected region of $N_A$ boundary spins necessarily allows reconstruction of all operators acting on the center provided that $A$ covers a sufficiently large fraction of the boundary, namely
\be\label{cthres}
f_A \equiv \frac{N_A}{N_{\rm boundary}}>\frac{5+\sqrt{5}}{10}\equiv f_c\approx .724.
\ee
The analogous result for the AdS-Rindler reconstruction is $f_c=1/2$, but the discreteness of our lattice introduces some additional overhead.
It turns out, though, that because the tensor network is not invariant under translations of the boundary, whether the connected region $A$ allows reconstruction of the center depends not just on the size of $A$ but also on its location. In appendix \ref{App:CountingTensors} we show that, while the condition \eqref{cthres} is needed to guarantee reconstruction of the central operator on an arbitrary connected region, there are some connected regions with $f_A =\frac{N_A}{N_{boundary}}=\frac{3+\sqrt{5}}{10}\approx .524$ that suffice for the reconstruction.

\subsection{Disconnected reconstruction and the entanglement wedge}
Now let's consider what bulk operators can be constructed on boundary regions with more than one connected component. First we extend the definition of the causal wedge to disconnected regions:
\begin{definition}
Suppose that $A$ is a boundary region, which is a union of connected components $A_1,A_2,\ldots$.  The \textbf{causal wedge of $A$}, denoted $\CA$, is defined as the union of the causal wedges of the components of $A$, $\CA=\bigcup_i \mathcal{C}[A_i]$.
\end{definition}
\noindent Since we have already established that any bulk operator in $\mathcal{C}[A_i]$ is reconstructable on $A_i$ if $A_i$ is connected, it follows immediately from this definition that even for disconnected regions any bulk operator in $\CA$ is reconstructable on $A$.

The causal wedge contains bulk operators which can be reconstructed when we apply the greedy algorithm to the connected components of $A$ one at a time. But the greedy algorithm might advance further into the bulk, beyond the causal wedge, when applied to $A$ instead. Specifically, there could be a $2n$-index tensor just beyond the causal wedge of $A$ with $n$ or more legs crossing the union of greedy geodesics $\gamma_{A_i}^\star \cup \gamma_{A_j}^\star$, even though fewer than $n$ legs cross $\gamma_{A_i}^\star$ or $\gamma_{A_j}^\star$ individually. Then applying the greedy algorithm to $A_i \cup A_j$ moves the cut past this tensor. This step may then render further tensors eligible for inclusion, and in fact we will see that sometimes the greedy algorithm can move far beyond the causal wedge $\CA$

\bfig
\includegraphics[height=7cm]{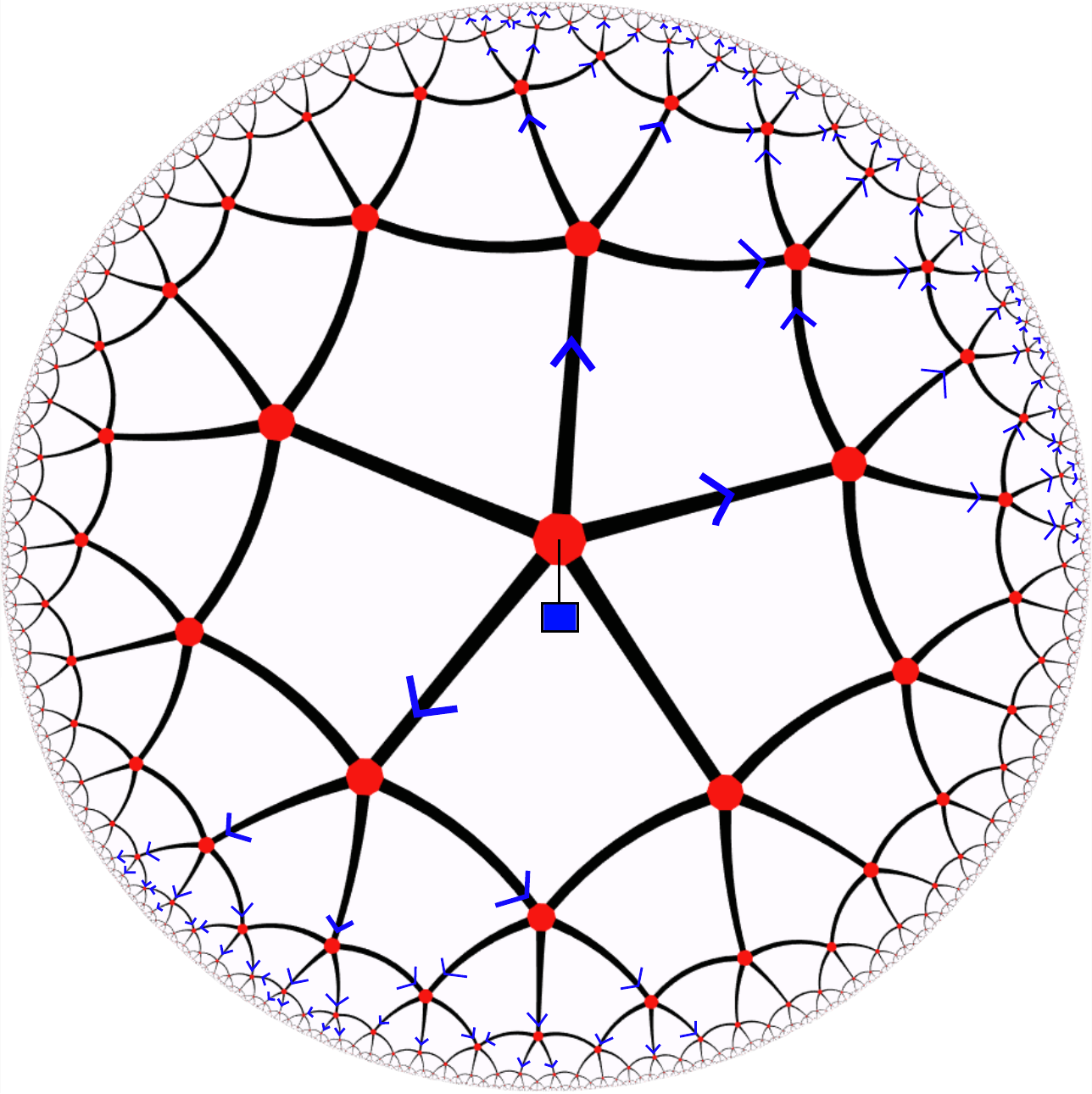}
\caption{Disconnected reconstruction of a central operator beyond the causal wedge.  Each of two separate connected boundary regions is too small for reconstruction of the central operator, yet the reconstruction is possible on the union of the two regions. In this example the greedy algorithm reaches the central tensor when applied to both connected components at once, but not when applied to either component by itself.}\label{pentagondcfig}
\efig
 
A concrete first example illustrating reconstruction of a bulk operator outside the causal wedge is shown in figure  \ref{pentagondcfig}. In this example, $A$ is the union of two connected components $A_1$ and $A_2$, and the full operator algebra of the central tensor can be pushed to either  $A_1^c$ or $A_2^c$. This implies that no nontrivial operator acting on the central tensor can be pushed to either $A_1$ or $A_2$. For every nontrivial operator $\phi$ in the algebra there is another operator $\phi'$ which does not commute with $\phi$. If $\phi'$ can be pushed to $A_1^c$, then surely $\phi$ cannot be pushed to $A_1$, because operators supported on complementary regions must commute. The same argument applies to $A_2$. Yet the greedy algorithm applied to $A$ reaches the central tensor, showing that its full operator algebra can be pushed to the union of $A_1$ and $A_2$.

That operators beyond the causal wedge of $A$ can be reconstructed on $A$ has deep potential implications for AdS/CFT. Perturbative gravity techniques like the AdS-Rindler reconstruction can be used to construct bulk operators in the causal wedge but not beyond. 
Yet there has been speculation in the literature that reconstruction should be possible in a larger region, the \textit{entanglement wedge} \cite{Headrick2014}, see also \cite{Wall2012, Czech2012, Jafferis2014}.  
In AdS/CFT, the entanglement wedge $\EA$ is defined by first finding the minimal area surface $\gamma_A$ used in the RT formula, and then drawing a codimension one (\textit{i.e.}, two-dimensional for ${\rm AdS}_3$) spatial slice in the bulk whose only boundaries are $\gamma_A$ and $A$.  
The bulk domain of dependence of this slice is then defined as the entanglement wedge $\EA$.  The entanglement wedge contains the causal wedge, but can be much larger in some cases. 
Figure \ref{transitionfig} illustrates a simple example highlighting the distinction between the causal and entanglement wedges.\footnote{In excited states where the geometry deviates from pure AdS, there are differences between the entanglement wedge and the causal wedge even for connected boundary regions.  
We will not try to capture this in our toy models, since without a theory of dynamics we cannot capture the full spacetime definitions of these regions.  Our discussion is limited to the case where we stick with states near the vacuum, in which case $A$ needs to be disconnected for its causal wedge and entanglement wedge to differ.} 

\bfig
\includegraphics[height=5cm]{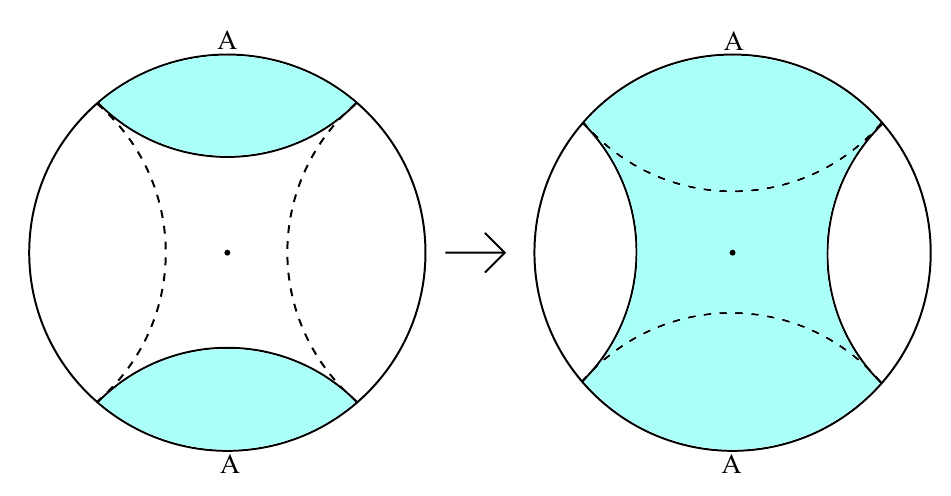}
\caption{The intersection of the entanglement wedge $\EA$ with a bulk time-slice, in the case where $A$ has two connected components. Minimal geodesics in the bulk are solid lines.  
When $A$ is smaller than $A^c$, we have the situation on the left and the causal wedge agrees with the entanglement wedge.  When $A$ is bigger, however, the minimal geodesics switch and the entanglement wedge becomes larger. 
In particular the point in the center lies in the $\EA$ but not $\CA$.}\label{transitionfig}
\efig

We would like to investigate whether bulk operators in the entanglement wedge are reconstructable for holographic codes, but how should the entanglement wedge be defined? A definition of $\EA$ close to that used in AdS/CFT is:
\begin{definition}
Suppose $A$ is a (not necessarily connected) boundary region.  The \textbf{geometric entanglement wedge} of $A$ is the set of bulk points in the bulk region bounded by $A$ and $\gamma_A$, where $\gamma_A$ is the minimal bulk geodesic whose boundary matches the boundary of $A$. If there is more than one minimal bulk geodesic, $\gamma_A$ is chosen to make the geometric entanglement wedge as large as possible.
\end{definition}

The main motivation for the conjecture that operators in the entanglement wedge are reconstructable in AdS/CFT comes from the validity of the RT formula for disconnected regions. (Additional evidence was given in \cite{Almheiri14} based on a typicality argument.)  But we have already seen above that the RT formula does not hold exactly in holographic codes, so we should not necessarily expect the entanglement wedge conjecture to hold in detail for the geometric entanglement wedge. Instead, as in defining the causal wedge, we prefer a definition that makes  the reconstructability manifest:

\begin{definition}
Suppose $A$ is a (not necessarily connected) boundary region.  The \textbf{greedy entanglement wedge} of  $A$, denoted $\EA$, is the set of bulk points reached by applying the greedy algorithm to all connected components of $A$ simultaneously.
\end{definition}
\noindent
With this definition, bulk local operators in $\EA$ are automatically reconstructable in $A$, using the isometry defined by $P_A$ to push these operators to the boundary. The greedy algorithm also ensures that the interior boundary of $\EA$ is the greedy geodesic $\gamma_A^\star$, though not necessarily the minimal geodesic $\gamma_A$. 

A drawback of this definition is that $\EA$ includes only the bulk local operators which can be reconstructed on $A$ using the greedy algorithm; it might miss additional bulk operators which can be reconstructed by other methods. 
In fact we can find examples of codes such that some bulk local operators lying outside $\EA$ \textit{can} be reconstructed on $A$, as discussed in appendix \ref{App:BeyondGreedy}. 
These codes typically have special properties, such as symmetries, which make the reconstruction possible. 
If we know nothing more about the perfect tensors used to construct the code, aside from their perfection, we have no general reason to expect that bulk operators far outside the greedy entanglement wedge will be reconstructable.
That said, we confess that we lack a complete understanding of when reconstruction is possible, and hope that further progress on this issue can be achieved in future work. 

\subsection{Erasure threshold}

If the entanglement wedge conjecture is true for AdS/CFT, if holographic codes faithfully model the entanglement structure of boundary theories with classical gravitational duals, and if the greedy entanglement wedge is a reasonable stand-in for the entanglement wedge, then we should be able to find holographic codes and boundary regions such that the greedy entanglement wedge reaches far outside the causal wedge. 
In this section we provide examples which confirm this expectation. 
One way to formalize this is to choose $A$ to be a randomly chosen  set of boundary spins, whose size is a specified fraction of the total boundary. 
The geometry of the hyperbolic plane suggests that, if $A$ is large enough, the causal wedge $\CA$ will stick close to the boundary, yet the entanglement wedge $\EA$ reaches the center of the bulk with high probability; we illustrate this in figure \ref{fig:EntanglementvsCausal}.  
We will see that not all holographic codes have this property, but we are able to provide concrete examples that do.

\bfig
\subfloat[Shallow causal wedge]{
	\includegraphics[width=0.35\linewidth]{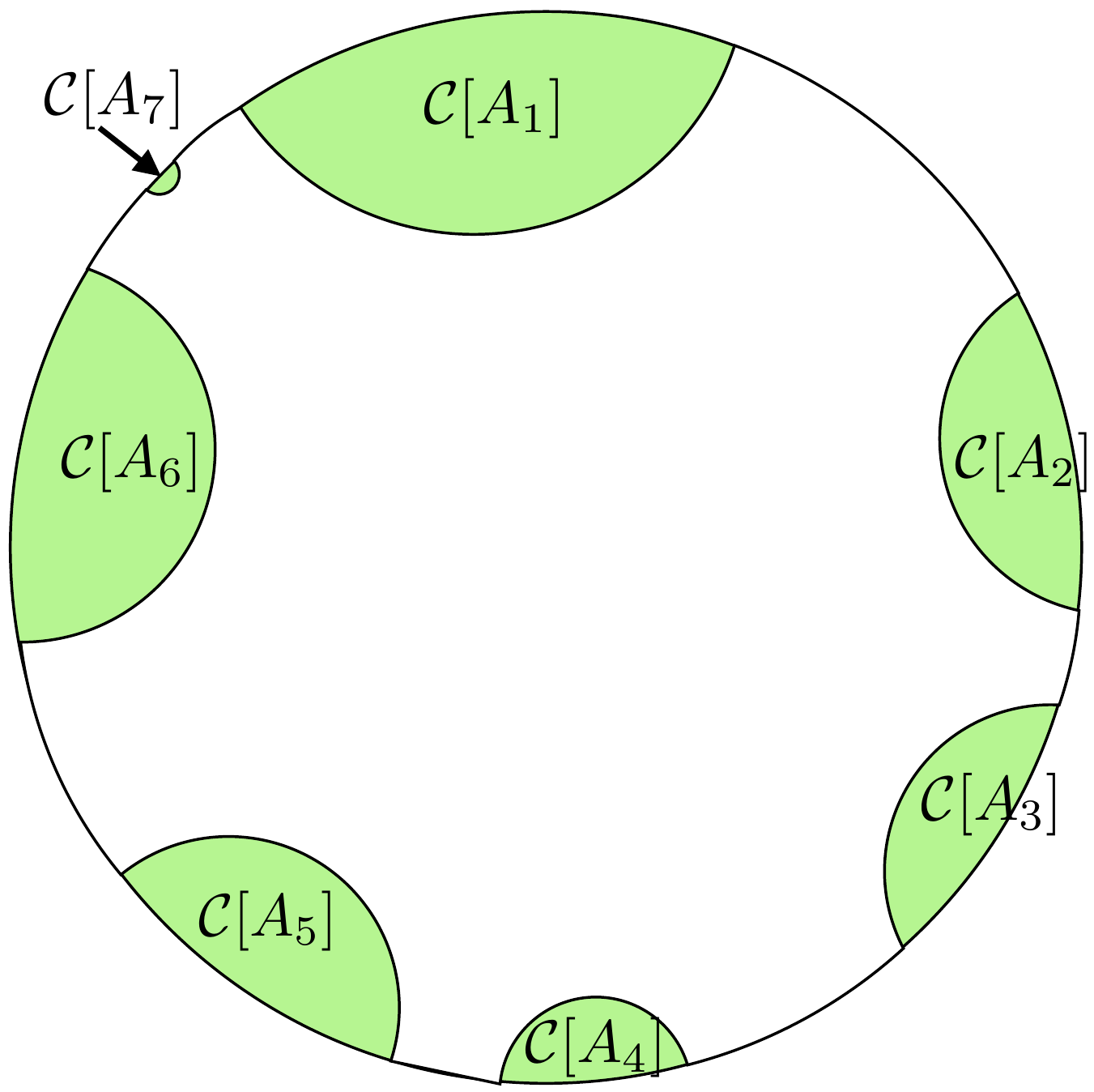}\hspace{1cm}
	\label{fig:ShallowCausalWedge}}
\subfloat[Deep entanglement wedge]{ 
 	\includegraphics[width=0.35\linewidth]{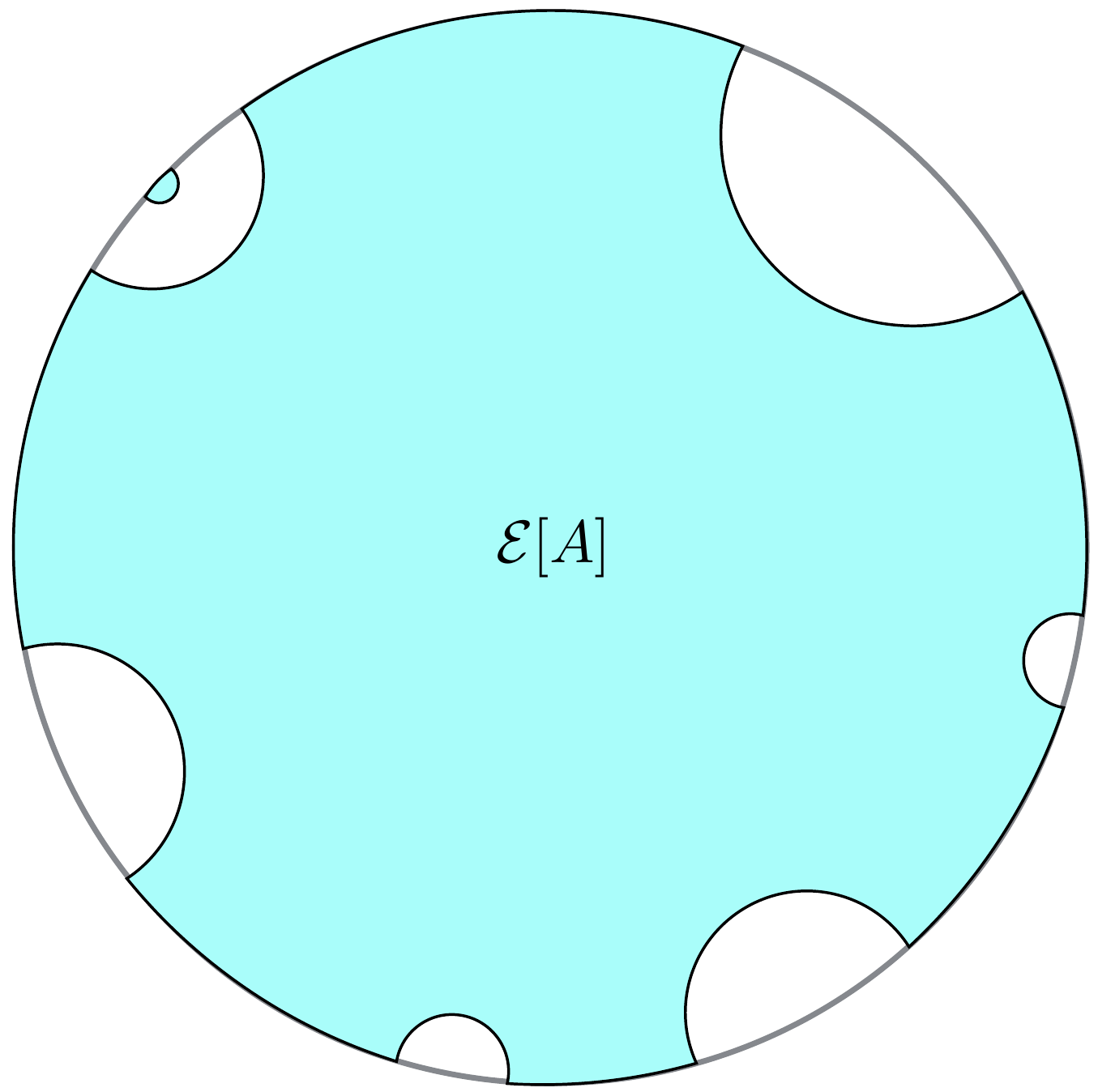}
 	\label{fig:DeepEntanglementWedge}}
\caption{ (a)  When a boundary region $A$ is partitioned into many connected components it may have a very shallow causal wedge $\cC[A]$ if each connected component is small.
(b) In contrast, if $A$ comprises a sufficiently large fraction of the boundary, its entanglement wedge $\cE[A]$ will extend deep into the bulk.
}
\label{fig:EntanglementvsCausal}
\efig

Another, perhaps better, way to formulate this case is to imagine a probabilistic noise model which acts independently (without any noise correlations) on each of the physical boundary spins, where each spin is either erased with probability $p$ or left untouched with probability $1-p$. If $p$ is small, the set $A$ of unerased boundary spins breaks into many connected islands, where a typical island contains $O(1/p)$ spins and has a causal wedge which reaches into the bulk by only a constant distance. We can show, though, that if the holographic code is properly chosen and the erasure probability $p$ is less than a \textit{threshold value} $p_c$, then $\EA$ contains the central bulk spin with a success probability deviating from one by an amount which becomes \textit{doubly exponentially small} as the radius of the bulk increases. 

Which codes have an erasure threshold? One necessary requirement is that the code must have a distance that increases with the system size. For the purpose of reconstructing the central tensor in the bulk, this means that there should not be any logical operator supported on a constant number of boundary spins which acts non-trivially on the central bulk index. That's because erasure of any constant number of spins occurs with a nonzero constant probability, and recovery from the erasure error is not possible if a nontrivial logical operator has support on the erased qubits. 

The pentagon code fails to fulfill this necessary condition. To illustrate the problem, it is helpful to consider first a simpler code, the ``triangle code'' constructed by contracting four-index perfect tensors, where each leg is a 3-level spin, a \textit{qutrit}. Each triangle in the bulk has a dangling bulk index, and the code is constructed as a tensor network forming a tree, the Bethe lattice; each triangle is contracted with one triangle closer to the center and two triangles further from the center, as shown on the left side of figure \ref{fig:BetheLattice}. (Qi's model~\cite{Qi13} is based on a tensor network with a similar structure.) One way to describe the greedy algorithm is to say that it propagates erasures from the boundary toward the center of the bulk --- the inward directed leg of a triangle is erased if either of its outward directed legs is erased, and the central triangle can be reconstructed only if at least two of that triangle's legs are unerased.

\bfig
\subfloat[Triangle code]{
	\includegraphics[height=6cm]{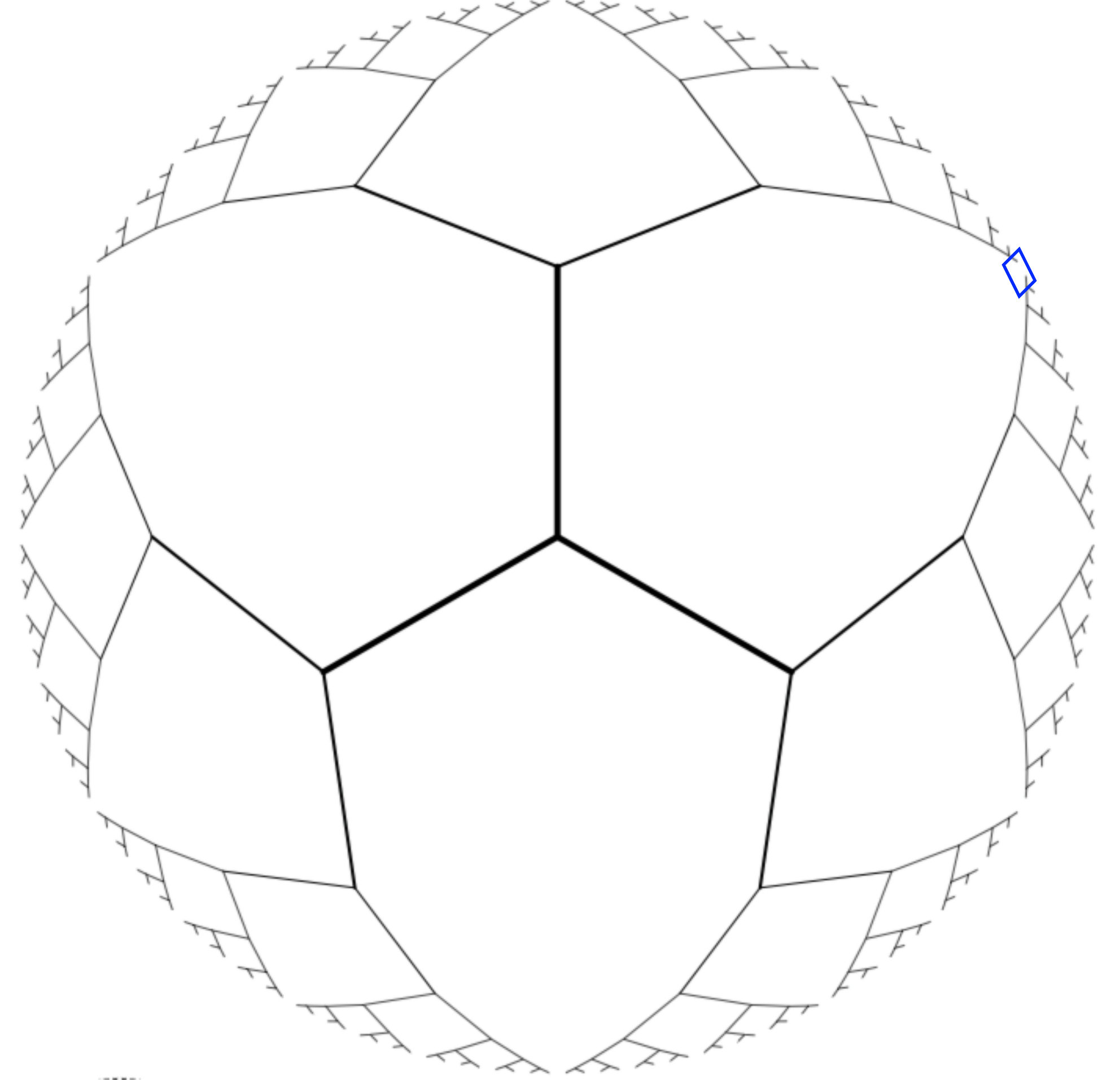}\hspace{1cm}
	\label{fig:BetheLattice}}
\subfloat[Pentagon code]{ 
 	\includegraphics[height=6cm]{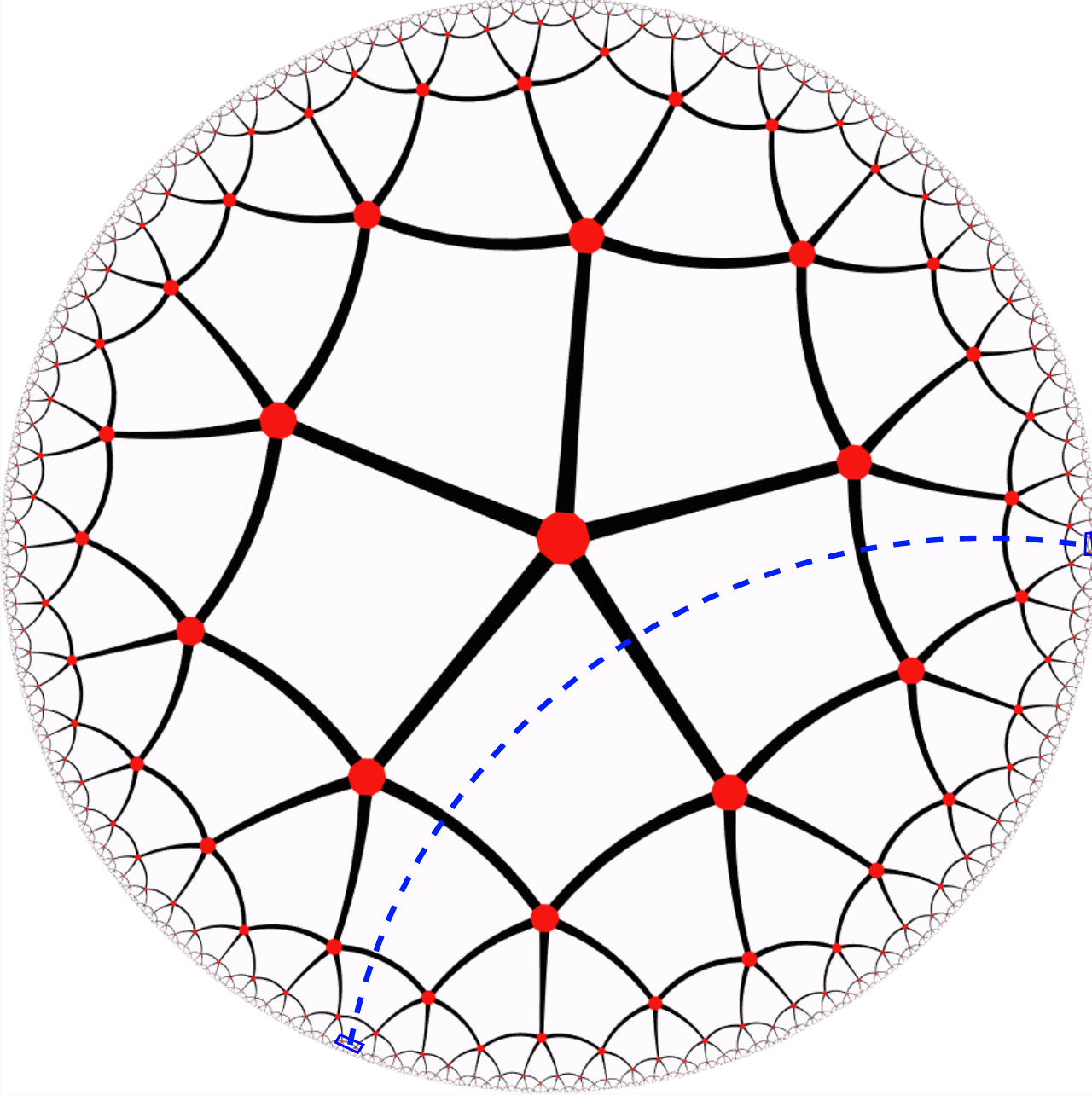}
 	\label{fig:benipoints}}
\caption{Dangerous small erasures for the triangle (a) and pentagon codes (b).  
In the triangle code erasing two boundary spins, boxed in blue, can prevent reconstruction of the central tensor.  
In the pentagon code erasing four spins can prevent the reconstruction.}
\label{fig:lowWeightTrouble}
\efig

It is easy, then, to prevent the greedy algorithm from reaching the center --- only two spins need to be erased. A single erasure on the boundary propagates all the way up to the center of the network, erasing one of the central triangle's legs. A single erasure on a different branch of the tree propagates up to another of the central triangle's legs, blocking the reconstruction of the central tensor on the remaining unerased spins. 

The greedy algorithm fails for a good reason. 
As described in appendix \ref{sec:qutritcode}, The logical algebra for the three-qutrit code represented by a single triangle is generated by logical operators of the form $\bar X =X\otimes X^{-1} \otimes I$, where $X$ is a generalized Pauli operator; 
in fact the code is symmetric under permutation of the three qutrits, so we can choose $X$ and $X^{-1}$ to act on any two of the three qutrits without changing the operator's action on the code space. 
Now choose a path through the Bethe lattice which begins on one leaf, travels to the center, exits the center on a different branch, and finally reaches another leaf on that branch. 
Apply the operator $\bar X$ to each of the logical bulk indices visited by this path. 
Then for each leg along the path the $X$ from the triangle on one side cancels the $X^{-1}$ coming from the triangle on the other side, except for one uncanceled $X$ on one leaf and one $X^{-1}$ on the other. 
We conclude that the code admits a logical operator acting nontrivially on the central triangle which has support on only two boundary spins. 
That is why the central bulk spin can be damaged by erasing only two boundary spins. 

For the pentagon code the situation is only slightly better. 
If we pick just four spins at the positions shown on the right side of figure \ref{fig:benipoints}, then the greedy algorithm applied to the complement of these four spins never absorbs any of the tensors adjacent to the dashed line. 
This failure is just a property of the graph defining the holographic code, but once again we can understand the failure by noting that there is a logical operator acting on the central pentagon supported on these four boundary spins, so erasing these four spins prevents central bulk operators from being reconstructed on their complement. Now we may consider a product of bulk logical operators acting on the pentagons just above and just below the dashed line. 
We use the logical operator of the five-qubit code $\bar X = - Z\otimes X\otimes Z\otimes I\otimes I$ described in appendix \ref{sec:5qubit}, where $X$ and $Z$ are Pauli operators (which square to one), and the operator's action is unchanged by cyclic permutations of the five qubits.
Now $X$'s applied from either side of the cut cancel on the legs crossed by the cut, and $Z$'s applied from either side cancel for the legs just above and below the cut, leaving only four uncanceled $Z$'s acting on the boundary qubits. 

Of course, uncorrectable damage deep inside the bulk caused by erasing just a few boundary spins is not at all what we expect in AdS/CFT, where according to the entanglement wedge conjecture we should always be able to reconstruct the center of the bulk from a sufficiently large fraction of the boundary, whatever its shape or location. To obtain a better model for AdS/CFT we should modify the holographic code, thinning out the algebra of bulk logical operators, and hence reducing the rate of the code. 

A code that works better can be obtained by a simple modification of the pentagon code --- the modified tensor network is constructed by starting with a pentagon at the center and adding alternating layers of hexagons (with no dangling bulk indices) and pentagons (each with one bulk index) as the network grows radially outward.
The associated network is depicted in figure \ref{fig:PentagonHexagon}.
This change suffices to remove all the constant-weight logical operators acting nontrivally on the center and in fact we can prove that this pentagon/hexagon code has an erasure threshold. Numerical studies show that erasure can be corrected by the greedy algorithm with high success probability for $p \le p_{c}^{\rm greedy} \approx 0.26$;  the erasure threshold $p_c$ achieved by the optimal recovery method might be higher than $p_{c}^{\rm greedy}$ if the tensors have further special properties aside from just being perfect.

\bfig
\subfloat[Pentagon/Hexagon code]{
\includegraphics[height=6cm]{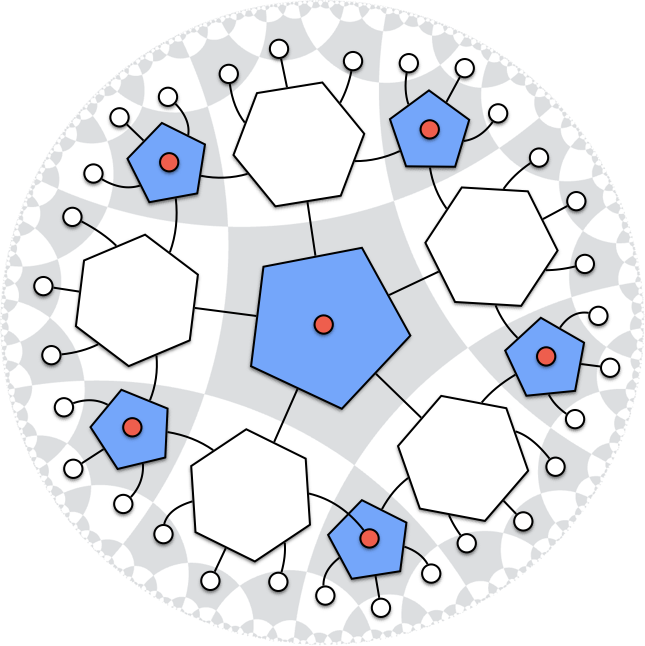}\hspace{1cm} 
\label{fig:PentagonHexagonCode}}
\subfloat[One qubit code]{
\includegraphics[height=6cm]{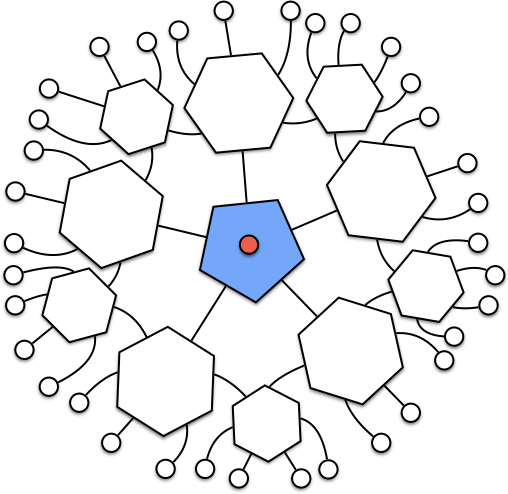}
\label{fig:1qubitHexagon}}
\caption{Tensor networks for holographic pentagon/hexagon codes with erasure thresholds, where neighboring polygons share contracted indices. In the network shown on the left, pentagons and hexagons alternate on the lattice; each pentagon carries one dangling bulk index, and hexagons carry no bulk degrees of freedom. The logical qubit residing on the central pentagon is well protected against erasure if the erasure probability on the boundary is below the threshold value $p_c$. 
In the network on the right, there is just a single bulk qubit located at the center; the rest of the network is similar to the holographic state constructed from hexagons only. 
}\label{fig:PentagonHexagon}
\efig

Since our main interest is in the reconstruction of the center of the bulk, in appendix \ref{App:ErasureThreshold} we study a code for which the only logical index resides at the center, also shown in figure \ref{fig:PentagonHexagon}. 
This code is almost the same as the holographic state obtained by contracting six-leg perfect tensors (hexagons), except that the tensor network contains one pentagon at the center; we therefore call it the \textit{single-qubit hexagon code}. We prove the existence of an erasure threshold for this code, and also derive an analytic lower bound on the threshold erasure rate $p_c \ge 1/12$. 
Numerical evidence indicates that the threshold is actually quite close to $p_c = 1/2$.

The lower bound on the threshold is derived using a simplified and less powerful version of the greedy algorithm, the hierarchical recovery method, which begins at the boundary and proceeds inward toward the center of the bulk. 
A tensor at level $j+1$ of this hierarchy is connected to at least four tensors at level $j$, and the level-$(j+1)$ tensor is erased if two or more of its level-$j$ neighbors are erased. 
The proof proceeds by recursively deriving an upper bound on the erasure probability $p_j$ at level $j$, finding
\begin{align}
p_j \leq p_c \left(\frac{p}{p_c}\right)^{\lambda^j},
\end{align}
where $p_c= 1/12$ and $\lambda = \frac{1+\sqrt{5}}{2}$. 
Thus the erasure probability for the central tensor drops doubly exponentially with the radius of the bulk if $p< p_c$, which means that the central tensor can be reconstructed on the set of unerased boundary qubits with very high probability.

A tricky aspect of the proof is that, because a single level-$j$ tensor couples to two level-$(j+1)$ tensors, there are noise correlations which propagate from level to level. Fortunately, the hyperbolic geometry controls the spread of correlations, making the analysis manageable. In fact, correlations beyond nearest neighbors never arise. This is one advantage of using the hierarchical recovery method rather than the greedy algorithm. A similar proof strategy may also be applied to other holographic codes. 

\subsection{Holographic stabilizer codes}\label{subsec:stabilizer}

Stabilizer codes have been extensively studied in quantum coding theory, and are often used in applications to fault-tolerant quantum computing \cite{Gottesman2009}. Here we describe how to construct a family of holographic codes which are also stabilizer codes. We introduce the stabilizer formalism to pave the way for section \ref{subsec:enough}, where we study some geometrical properties of holographic stabilizer codes. 

Stabilizer codes can be defined for higher-dimensional spins as well, but here we will assume the spins are qubits for simplicity. A \textit{Pauli operator} acting on $n$ qubits is a tensor product of Pauli matrices, that is, one of the $4^n$ operators contained in the set
\be
\{I,X,Y,Z\}^{\otimes n}
\ee
where $I$ is the $2\times 2$ identity matrix and $X, Y, Z$ are the $2\times 2$ Pauli matrices (often denoted $\sigma_x,\sigma_y,\sigma_z$). We use $[[n,k]]$ to denote a quantum code with $k$ logical qubits embedded in a block of $n$ physical qubits. We say that an $[[n,k]]$ code is a \textit{stabilizer code} (also called an \textit{additive} quantum code), if the code space can be completely characterized as the simultaneous eigenspace of $n-k$ commuting Pauli operators. These commuting Pauli operators are called the code's \textit{stabilizer generators} because they generate an abelian group called the code's \textit{stabilizer group}. The special case of a $k=0$ stabilizer code is called a \textit{stabilizer state}. We say that an $n$-index tensor is a \textit{stabilizer tensor} if the corresponding $n$-qubit state is a stabilizer state. 

For example, the six-index perfect tensor is a perfect stabilizer tensor, and holographic codes defined by tiling a hyperbolic geometry with pentagons are stabilizer codes. More generally, we may formulate the following theorem: 

\begin{theorem} \label{thm:stabilizer}
Consider a holographic code defined by a contracted network of perfect stabilizer tensors, and suppose that the greedy algorithm starting at the boundary reaches the entire network. Then the code is a stabilizer code. 
\end{theorem}

\noindent To understand why Theorem \ref{thm:stabilizer} is true we need to see how to construct the code's stabilizer generators. To be concrete, consider holographic codes constructed from tilings by hexagons and pentagons. The six-index perfect tensor defines a [[6,0]] stabilizer code, whose stabilizers are enumerated in appendix \ref{sec:5qubit}. 
As we have already noted, it also defines isometries from any set of 1, 2, or 3 indices to the complementary set of indices; these isometries may be regarded as the encoding maps for $[[5,1]]$, $[[4,2]]$, and $[[3,3]]$ stabilizer codes respectively. 

To be specific, consider the $[[5,1]]$ code, and let $M$ denote its isometric encoding map taking a one-qubit input to the corresponding encoded state in the code block of five qubits. We can characterize $M$ by specifying how it acts on Pauli operators, which (together with the identity) span the space of operators acting on a single qubit. Since the Pauli group is generated by $X$ and $Z$ it suffices to specify
\be
M: X\mapsto \bar X,\quad M:Z \mapsto \bar Z,
\ee
where $\bar X$ and $\bar Z$ are the code's logical Pauli operators, given explicitly in appendix \ref{App:PerfectTensorExamples}. 
Similarly, the action on Pauli operators defines isometric encoders for the $[[4,2]]$ and $[[3,3]]$ stabilizer codes, except that for \textit{e.g.} the $[[4,2]]$ code we specify the action on the four independent Pauli operators $X_1,X_2,Z_1,Z_2$, where the subscript $1,2$ labels the code's two logical qubits. For stabilizer codes the encoding isometry is always a \textit{Clifford isometry}, meaning its action by conjugation maps $k$-qubit Pauli operators to $n$-qubit Pauli operators. 

We already explained in section \ref{sec:model} that when the condition of Theorem \ref{thm:stabilizer} is satisfied then the encoding isometry for the holographic code can be obtained by composing the isometries associated with each perfect tensor in the network. 
A given tensor may have 0, 1, 2, or 3 incoming legs, including the dangling bulk leg (if the tensor is a pentagon) and all the incoming contracted legs, which are output legs from previously applied isometries. 
To prove Theorem \ref{thm:stabilizer} then, it is enough to know that composing the encoding isometries of two stabilizer codes yields the encoding isometry of a stabilizer code. 

To see how this works, it is helpful to think about the simple special case of a \textit{concatenated quantum code}, for which the tensor network is a tree.
 Consider in particular a code with just one logical qubit --- the central pentagon has one incoming logical leg and five outgoing legs, while every other tensor is a hexagon with one incoming leg and five outgoing legs. 
 If the $[[5,1]]$ code is concatenated just once, the tensor network has five hexagons and describes a $[[25,1]]$ stabilizer code. 
To obtain this code's isometric map, we first apply the encoding isometry $M$ of the $[[5,1]]$ to the logical qubit, and then apply $M$ again to \textit{each one} of the five outgoing qubits. 
If $S$ denotes the stabilizer group of the $[[5,1]]$ code, then the stabilizer of the $[[25,1]]$ code will include $S$ acting on each one of the five subblocks corresponding to the five hexagons in the tensor network. 
But it also includes elements which act collectively on four of the five hexagons. For example, as described in appendix \ref{sec:5qubit}, one of the stabilizer generators for the $[[5,1]]$ code is the Pauli operator $X\otimes Z\otimes Z\otimes X\otimes I$.
 The isometries associated with the five hexagons map this operator to $\bar X\otimes \bar Z\otimes \bar Z\otimes \bar X\otimes I$, where now $\bar X,\bar Z$ are the logical Pauli operators acting on the five outgoing qubits emanating from a single hexagon.

The same idea applies to more general compositions of code isometries. 
Suppose that $\mathcal{S}_1, M_1$ are the stabilizer group and encoding isometry for an $[[n_1,k_1]]$ stabilizer code and that $\mathcal{S}_2, M_2$ are the stabilizer group and encoding isometry for an $[[n_2,k_2]]$ stabilizer code. 
We may apply $M_2$ to $m$ of the $n_1$ output qubits from $M_1$ along with $k_2-m$ additional input qubits (where $m\le n_1$ and $m \le k_2$), thus obtaining an $[[ n_1- m + n_2, k_1 + k_2-m]]$ code. 
In fact this code is a stabilizer code, whose stabilizer group is generated by $\mathcal{S}_2$ and $M_2(\mathcal{S}_1)$; here we use a streamlined notation, in which it is understood that operators and maps are extended by identity operators where necessary, and we note that the elements of $M_2(\mathcal{S}_1)$ are Pauli operators because $M_2$ is a Clifford isometry. 
Thus we have proven Theorem \ref{thm:stabilizer}. It is also worthwhile to note that the stabilizer group and encoding isometry for the holographic code can be efficiently computed by composing the isometries arising from the perfect tensors in the network. 

\subsection{Are local gauge constraints enough?}\label{subsec:enough}

It has recently been argued that in AdS/CFT gauge constraints in the boundary CFT may pick out a small enough subspace of states to explain the error correcting properties of AdS/CFT \cite{Mintun2015}.  
The idea is that any gauge-invariant state already possess some non-local entanglement via the imposition of the gauge constraints, and that this might be enough to resolve the various paradoxes of \cite{Almheiri14}.\footnote{The word ``gauge'' is sometimes used in quantum information theory in a way that is non-standard from the point of view of quantum field theorists.  In quantum field theory, states that are not gauge-invariant have no physical interpretation, and are not really part of the Hilbert space of the theory; they appear only as a mathematical convenience.   
This is what the authors of \cite{Mintun2015} meant by gauge constraints, and it is what we mean here.}  

We can try to test this idea for the holographic stabilizer codes discussed in section \ref{subsec:stabilizer}. 
Since gauge constraints are spatially local, the argument of Ref. \cite{Mintun2015} suggests that the code's stabilizer group should be \textit{locally generated}, in the sense that it has a complete set of generators, each with support on a constant number of neighboring boundary qubits. In fact, though, holographic stabilizer codes do not have this property in cases where the greedy entanglement wedge reaches outside the causal wedge. 
This property poses no problem for the proposal of  \cite{Almheiri14} however, as those authors argued that energetic constraints should also be included in defining the code subspace.

Consider for example the disconnected boundary region $A=A_1\cup A_2$ in the pentagon code, depicted in figure \ref{pentagondcfig}. We have already seen that the full logical algebra of the central pentagon can be reconstructed on the disconnected region $A_1\cup A_2$, but that no nontrivial logical operator acting on the central pentagon is supported on either one of the connected components $A_1$, $A_2$. In a stabilizer code, a logical Pauli operator supported on $A_1\cup A_2$ is a tensor product $\mathcal{O}=\mathcal{O}_{A_1}\otimes \mathcal{O}_{A_2}$ of Pauli operators supported on $A_1$ and $A_2$ separately. In order to preserve the code space, this logical Pauli operator must commute with all of the code's stabilizer generators. But if the two components $A_1$ and $A_2$ are distantly separated and the stabilizer generators are geometrically local, then no stabilizer generator has nontrivial support on both $A_1$ and $A_2$. Any stabilizer generator with no support on $A_2$ trivially commutes with $\mathcal{O}_{A_2}$, and if it commutes with $\mathcal{O}$ then it must commute with $\mathcal{O}_{A_1}$ as well. Likewise, a stabilizer generator with no support on $A_1$ must commute with $\mathcal{O}_{A_2}$ if it commutes with $\mathcal{O}$. 
Therefore $\mathcal{O}_{A_1}$ and  $\mathcal{O}_{A_2}$ are logical operators, and at least one is nontrivial if their product is, contradicting the hypothesis that no nontrivial logical operator is supported on either connected component of $A$. The conclusion is that the stabilizer generators cannot be geometrically local.

The above argument applies even to higher-dimensional holographic stabilizer codes. In the case were the boundary is one dimensional, we may simply appeal to a known result in quantum coding theory, that a stabilizer code in one dimension with geometrically local generators has constant distance~\cite{Bravyi09, Pastawski15}. 
Therefore, a one-dimensional code with a local stabilizer cannot have a positive erasure threshold.

\section{Black holes and holography}\label{sec:black}
In holographic codes, bulk operators are reconstructed only on a subspace of the boundary Hilbert space.  This may seem troubling, since the holographic correspondence is supposed to assign a bulk interpretation to all possible states on the boundary.  A resolution of this confusion was proposed in \cite{Almheiri14} --- a particular bulk operator might not always be reconstructable because it lies deep inside a black hole for most boundary states.\footnote{We are currently agnostic about the reconstruction of bulk operators just inside the horizon, which must be needed in some form to describe the experience of an infalling observer.  This is a topic of much recent controversy \cite{Almheiri13, Harlow2014}, but we will not take sides here.}  In fact we can see this directly in our models if we incorporate black holes in a manner that we now describe.

To illustrate the idea, consider the pentagon code, but with the central tensor removed.  The central tensor's one free bulk index has been replaced by five bulk indices, those which had previously been contracted with legs of the missing pentagon; the tensor network now provides an isometry mapping these five indices, together with the bulk legs on the remaining pentagons, to the boundary.  Thus the code subspace of the boundary Hilbert space is larger than for the pure pentagon code.  We interpret this enlarged code space as describing configurations of the bulk with a black hole in the center, whose microstate is determined by the input to the new bulk legs.  The entropy of the black hole is the logarithm of the dimension of the Hilbert space of black hole microstates, or
\be
S_{BH}=\log_2\left(2^5 - 2\right)\approx 4.9,
\ee
since only four of the bulk spins are new and we shouldn't count states that were part of the original pentagon code subspace.  We depict this construction in figure \ref{fig_black_hole}a.  
\begin{figure}
\centering
\subfloat[Black hole]{
\includegraphics[width=0.35\linewidth]{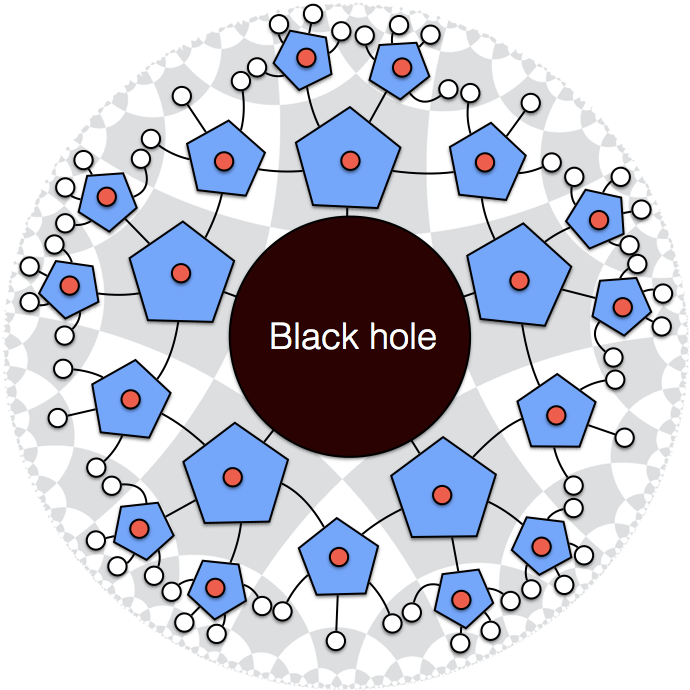}
\label{fig:BlackHole}
}\hspace{1cm}
\subfloat[Wormhole]{
\includegraphics[width=0.40\linewidth]{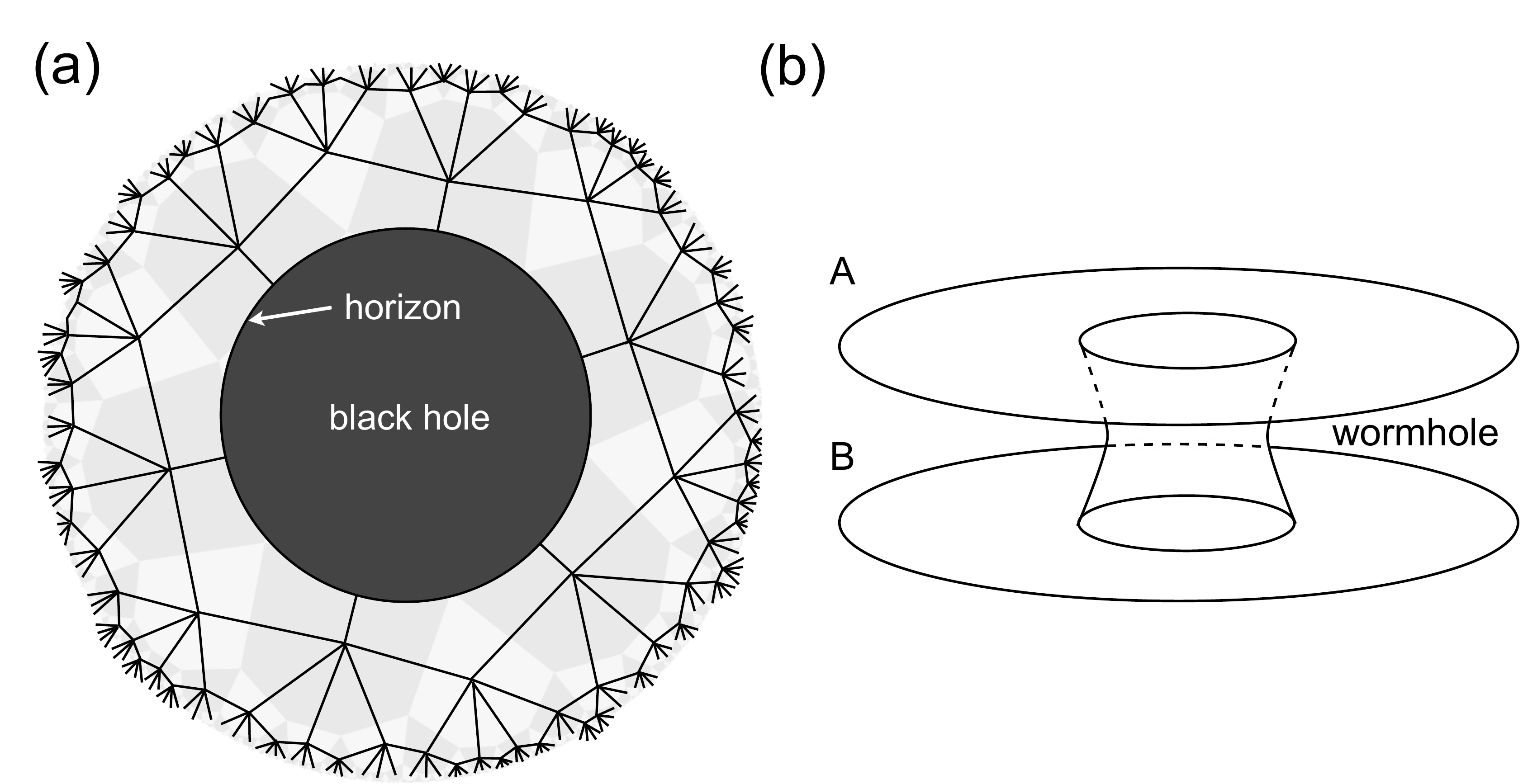}
\label{fig:Wormhole}
}
\caption{A black hole in a holographic code, and the corresponding wormhole geometry.
} 
\label{fig_black_hole}
\end{figure}

We can construct larger black holes by removing more central layers of the network; it is clear that their entropy scales with their horizon area, as predicted by Bekenstein and Hawking \cite{Bekenstein73, Hawking75}.  As the black hole grows, the number of bulk legs outside the black hole decreases, so we can reconstruct fewer and fewer bulk local operators.  Eventually the black hole eats up the entire network, and our isometry becomes trivial (and unitary).  Thus our model really does assign a bulk interpretations to all boundary states, as demanded by AdS/CFT --- most boundary states correspond to large black holes in the bulk. 

It is amusing to note that we can also describe configurations corresponding to the two-sided wormhole of \cite{Maldacena03}; we just prepare two networks with central black holes of equal size, and maximally entangle the bulk legs at their horizons, as shown in figure \ref{fig_black_hole}b.  
It would be interesting to make contact with recent speculations about how the length of the wormhole relates to the complexity of the tensor network describing the state \cite{Hartman2013, Susskind2014, Roberts2014}, although for that purpose we would probably need to incorporate dynamics into our model.  

\section{Open problems and outlook}\label{sec:conclude}

A remarkable convergence of quantum information science and quantum gravity has accelerated recently, propelled in particular by a vision of quantum entanglement as the foundation of emergent geometry. We expect this interface area to continue to grow in importance, as practitioners in both communities struggle to develop a common language and toolset. This paper was spurred by the connection between AdS/CFT and quantum error correction proposed in \cite{Almheiri14}. We have strived to make this connection more concrete and accessible  by formulating toy models which capture the key ideas, and we hope our account will equip a broader community of scientists to contribute to further progress. Indeed, much remains to be done. 

First of all, the entanglement structure of holographic codes is not yet completely understood. We would like a more precise characterization of the violations of the Ryu-Takayanagi formula which can occur, and of the relationship between bulk residual regions and the multipartite entanglement of the boundary state. How is the greedy entanglement wedge different from the geometric entanglement wedge, and to what extent does the greedy entanglement wedge reach beyond the causal wedge?

We have not yet discussed the correlation functions of boundary observables in holographic codes because we do not have much to say. 
In a stabilizer state $|\psi\rangle$, where $P$ and $Q$ are Pauli operators, the expectation value $\langle \psi| PQ|\psi\rangle$ is either zero (if $PQ$ anticommutes with an element of the stabilizer) or a phase (if $PQ$ commutes with the stabilizer); the same conclusion applies to a stabilizer code unless $PQ$ is a nontrivial logical operator preserving the code subspace.
In contrast, two-point correlations in a CFT decay algebraically with distance; how might we recover this behavior in holographic codes? 
Perhaps algebraic decay is recovered for non-stabilizer holographic states, by defining suitable coarse-grained observables, or by injecting an encoded state such that bulk correlation functions decay exponentially as in Ref. \cite{Qi13}. 
Or we might replace perfect tensors by tensors which are nearly perfect.

The behavior of two-point correlators highlights one way our toy models differ from full-blown AdS/CFT, but there are other ways as well; for one, there is no obvious analog of diffeomorphism invariance in a lattice model. What features in our lattice model correspond to the $1/N$ corrections in the continuum theory? In AdS/CFT the AdS radius is large compared to the Planck scale when the bulk theory is weakly coupled, yet in the pentagon model for example the curvature scale is comparable to the lattice cutoff. To approximate flatter bulk geometries we should study more general tessellations, including higher dimensional ones. A particularly serious drawback of our toy models so far is that we have not introduced any bulk or boundary dynamics. Can holographic codes illuminate dynamical processes like the formation and evaporation of a black hole?

Finally, we have emphasized that holographic states and codes provide a concrete realization of some aspects of AdS/CFT, but they may also be interesting for other reasons, for example as models of topological matter. 
Furthermore, holographic codes generalize the concatenated quantum codes that have been extensively used in discussions of fault-tolerant quantum computing \cite{Gottesman2009}, and might likewise be applied for the purpose of protecting quantum computers against noise. 
For this application it would be valuable to develop the theory of holographic codes in a variety of directions, such as studying tradeoffs between rate and distance, formulating efficient schemes for correcting more general errors than erasure errors, and finding ways to realize a universal set of logical operations acting on the code space.

\section*{Acknowledgment}

We thank Ning Bao, Oliver Buerschaper, Glen Evenbly, Daniel Gottesman, Aram Harrow, Isaac Kim, Seth Lloyd, Nima Lashkari, Hirosi Ooguri, Grant Salton, Kristan Temme, Guifre Vidal and Xiaoliang Qi for useful comments and discussions.
We also have enjoyed discussions with Ahmed Almheiri, Xi Dong, and Brian Swingle, and with Matthew Headrick, about their independent and upcoming related work.
FP, BY, and JP acknowledge funding provided by the Institute for Quantum Information and Matter, a NSF Physics Frontiers Center with support of the Gordon and Betty Moore Foundation (Grants No. PHY-0803371 and PHY-1125565). BY is supported by the David and Ellen Lee Postdoctoral fellowship. DH is supported by the Princeton Center for Theoretical Science.  

\appendix

\section{Perfect tensor examples}\label{App:PerfectTensorExamples}

In this section we will present the 5-qubit code, the 3-qutrit code and discuss possibilities of constructing  perfect tensors with a larger number of legs.
The 5-qubit code is a qubit stabilizer code of the form introduced in section \ref{subsec:stabilizer} whereas
the 3-qutrit code can be described through a natural generalization of the stabilizer formalism to higher spin dimensions.

\subsection{$5$-qubit code and $6$-qubit state}\label{sec:5qubit}

The five-qubit code is a $[[5,1,3]]_2$ perfect code with distance 3 encoding one logical qubit in five physical qubits.
It is a stabilizer code with a stabilizer subgroup given by $\cS =\langle S_1,S_2,S_3,S_4 \rangle$, where
\begin{equation}\label{eq:5qubitStabilizers}
\begin{split}
S_{1}=X\otimes Z\otimes Z \otimes X \otimes  I \\
S_{2}=I\otimes X\otimes Z \otimes Z \otimes  X \\
S_{3}=X\otimes I\otimes X \otimes Z \otimes  Z \\
S_{4}=Z\otimes X\otimes I \otimes X \otimes  Z .
\end{split}
\end{equation}
Note that $S_1S_2S_3S_4 = Z\otimes Z \otimes X \otimes I\otimes X$ and hence the group is manifestly invariant under cyclic permutations.
As is the case in the stabilizer formalism, codes are characterized by an abelian stabilizer subgroup ($[S_{i},S_{j}]=0$) and codespace is the joint $+1$ eigenspace for this group, code states satisfy
\begin{align}
S_{j}|\psi\rangle = |\psi\rangle\qquad j=1,\ldots 4.
\end{align}
In this case, there are two orthogonal codeword states.

Logical operators are unitary operators which preserve the codeword space, but may act non-trivially on it. They are given by
\begin{align}\label{eq:5qubitLogicals}
\overline{X}=X\otimes X\otimes X \otimes X \otimes  X \qquad
\overline{Z}=Z\otimes Z\otimes Z \otimes Z \otimes  Z.
\end{align}
One may see that both $\overline{X}$ and $\overline{Z}$ commute with all the stabilizer generators, so they indeed preserves the codeword space. Yet, they anti-commute with each other, so they characterize one logical qubit, and $\overline{X}$ and $\overline{Z}$ behave as logical Pauli-$X$ and -$Z$ operators for a logical qubit. Namely, one can denote two codeword states by $|\tilde{0}\rangle$ and $|\tilde{1}\rangle$ such that $\overline{Z}|\tilde{0}\rangle=|\tilde{0}\rangle$, $\overline{Z}|\tilde{1}\rangle=-|\tilde{1}\rangle$, $\overline{X}|\tilde{0}\rangle=|\tilde{1}\rangle$, $\overline{X}|\tilde{1}\rangle=|\tilde{0}\rangle$. Applications of stabilizer generators to logical operators do not change the action on the codeword space, so representations of logical operators are not unique. 
Then one can introduce the following equivalence relations among logical operators
\begin{align}
\overline{Z}\sim \overline{Z}U \qquad \overline{X}\sim \overline{X}U \qquad \mbox{where} \quad U \in \mathcal{S}
\end{align}
as equivalent logical operators act in the same way on the codeword space.
In particular, one can conclude that on the codespace, $\bar{X} \sim - Z \otimes X \otimes Z \otimes I \otimes I$ or any cyclic permutation thereof by multiplying eq. \ref{eq:5qubitLogicals} by stabilizer generators in eq. \ref{eq:5qubitStabilizers}.
In the five-qubit code, one can show that logical operators must act non-trivially on at least three physical qubits (weight 3) and the reduced density matrices on any two physical qubits is always maximally mixed. 

One can convert the five-qubit code into a six-qubit perfect state. 
Imagine that we add one extra qubit to the five-qubit code such that the new qubit is entangled with a logical state of the five-qubit code. To be specific, we consider a six-qubit state whose stabilized by $\cS' = \langle S'_1, S'_2, S'_3, S'_4, S'_5 ,S'_6 \rangle$, with generators are given by
\begin{equation}
\begin{split}
S_{1}'&=X\otimes Z\otimes Z \otimes X \otimes  I \otimes  I\\
S_{2}'&=I\otimes X\otimes Z \otimes Z \otimes  X \otimes  I\\
S_{3}'&=X\otimes I\otimes X \otimes Z \otimes  Z \otimes  I\\
S_{4}'&=Z\otimes X\otimes I \otimes X \otimes  Z \otimes  I\\
S_{5}'&=X\otimes X\otimes X \otimes X \otimes  X \otimes  X = \overline{X}\otimes X \\
S_{6}'&=Z\otimes Z\otimes Z \otimes Z \otimes  Z \otimes  Z = \overline{Z}\otimes Z .
\end{split}
\end{equation}
Here we have ``recycled'' stabilizer generators $S_{1},\ldots,S_{4}$ from the five-qubit code: 
\begin{align}
S_{j}' = S_{j} \otimes I \qquad  \mbox{for}\quad j=1,\ldots,4.
\end{align}
We then constructed new stabilizer generators $S_{5}'$ and $S_{6}'$ from logical operators of the five-qubit code as follows:
\begin{align}
S_{5}' = \overline{X} \otimes X \qquad  S_{6}' = \overline{Z} \otimes Z
\end{align}
where $\overline{X}$ and $\overline{Z}$ act on five qubits. One may easily check that stabilizer generators commute with each other. The wavefunction is specified by 
\begin{align}
S_{j}'|\psi\rangle = |\psi\rangle\qquad j=1,\ldots, 6,
\end{align}
and the six-qubit state is given by $|\psi\rangle = |\tilde{0}\rangle \otimes |0\rangle + |\tilde{1}\rangle \otimes |1\rangle$. From the construction, one can see that $\rho_{A}\propto I_{A}$ if $|A|\leq 3$. It turns out that this conversion is generic. That is, one can always convert a perfect code with $2n-1$ spins into a perfect state with $2n$ spins.

\subsection{3 -qutrit code and 4-qutrit state}\label{sec:qutritcode}

One of the simplest examples of perfect tensors is given by the three qutrit code.
This stabilizer code allows encoding one logical qutrit onto three physical qutrits in a way that it may be recovered even after erasure of any single physical qutrit.
The reason for providing this additional example at this point is that we believe that this family of perfect tensors may naturally be suited to generalizations leading to a continuum type limit. 

The code states for this code can be given as follows.
\begin{align*}
 \sqrt{3}| \tilde{0} \rangle  & = | 000 \rangle + |111\rangle + | 222\rangle \\
  \sqrt{3}| \tilde{1} \rangle  & = | 012 \rangle + |120\rangle + | 201\rangle \\
   \sqrt{3}| \tilde{2} \rangle  & = | 021 \rangle + |102\rangle + | 210\rangle.
\end{align*}
Correspondingly, the perfect state can be given explicitly as
\begin{align*}
3| [[4,0,3]]_3 \rangle = & | 0000 \rangle + |1110\rangle + | 2220\rangle \\
                            + & | 0121 \rangle + |1201\rangle + | 2011\rangle \\
                            + &  | 0212 \rangle + |1022\rangle + | 2102\rangle.
\end{align*}

The $[[4,0,3]]_3$ state is determined by the following stabilizer group
\begin{align}
\cS = \langle ZZZI, ZZ^{-1}IZ, XXXI,  X X^{-1} IX \rangle,
\end{align}
where we have omitted the three tensor product operator $\otimes$ between the qutrit Pauli operators 
\begin{align}
I = \begin{pmatrix}
1 & 0 & 0\\
0 & 1 & 0 \\
0 & 0 & 1
\end{pmatrix}
\quad
X = \begin{pmatrix}
0 & 1 & 0\\
0 & 0 & 1 \\
1 & 0 & 0
\end{pmatrix}
\quad
Z = \begin{pmatrix}
1 & 0 & 0\\
0 & \omega & 0 \\
0 & 0 & \omega^2 
\end{pmatrix}
\quad
\omega=e^{2i\pi/3} .
\end{align}
From this presentation of the stabilizer group, it becomes clear that this perfect state corresponds to a self dual CSS code.
\footnote{For CSS codes, the stabilizer group can be decomposed into $X$ part and $Z$ part. 
Self-dual means that the stabilizer subgroups for the $X$ and the $Z$ part have exactly the same form.}
This separation of $X$ and $Z$ type operators, and treating them on similar footing may potentially allow generalization to continuum variable where position and momentum conjugate variables play a similar role.

Like any other maximally entangled stabilizer state, given a bipartition, the $[[4,0,3]]_3$ perfect state may be interpreted as a unitary gate belonging to the generalized Clifford group.
A second reason for presenting this code is that we can provide a simple and explicit presentation of the corresponding Clifford circuit.
In the case of the $[[4,0,3]]_3$, the corresponding Clifford is composed to two controlled adder  gates which are run in opposite directions one after the other.
\begin{align}
  U_{[[4,0,3]]_3} = |x_1,x_2 \rangle \rightarrow | 2x_1+x_2, x_1+x_2 \rangle = \begin{picture}(8,8)(2,16)
\put(0,0){\includegraphics[scale=.18]{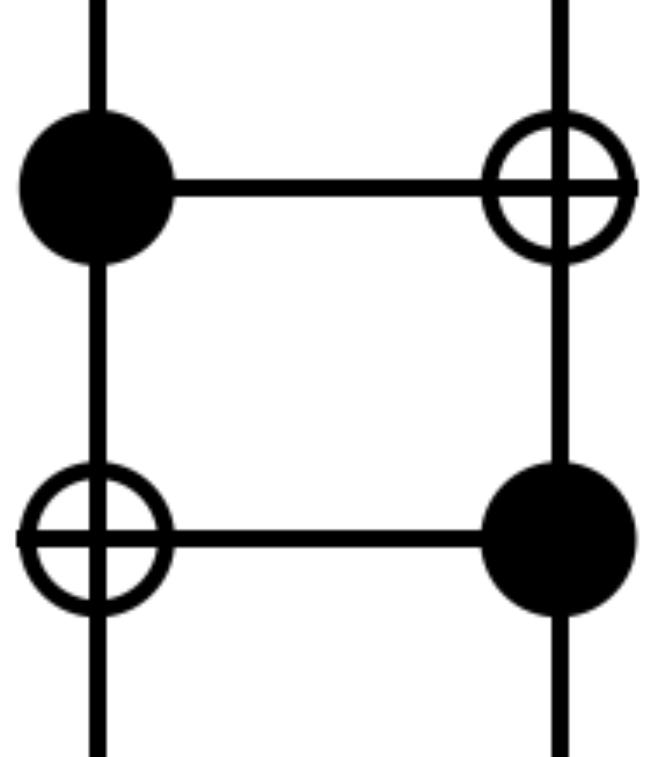}}
\end{picture}
\end{align}

\subsection{Large $n$}

In order to obtain perfect tensors for large $n$, one needs to increase $v$ as well, with the first known construction \cite{Aharonov97} having $v \propto O(n)$.
A construction with $v \propto O(\sqrt{n})$ was later proposed~\cite{Grassl04}. 

While perfect tensors are essential in guaranteeing the isometric properties used in the analysis of holographic codes, they do require a large degree of fine tuning for their construction.
An interesting observation is that according to canonical typicality, most pure states are almost maximally entangled along any balanced bi-partition~\cite{Page93, Goldstein06}.
Hence Haar random states may in some sense provide a good approximation to perfect states.
This is not meant  in the sense of trace distance.
The observation is that random states will except for a measure zero subset have full rank along any bipartition.
This allows operators to be pushed from lower dimensional side to higher dimensional side in a way similar to figure \ref{push}.
The only difference being that in this case normalization is not preserved (if it was there to begin with).
Furthermore, the average bipartite entanglement of random states is very close to maximal.
For this reason, we might expect that they typically do not change normalization too drastically.

\section{Proof of RT for negatively curved planar graphs}\label{app:PlanarGraphProof}

As we have presented in eq. \ref{upperbound}, there is an upper bound on the amount of entanglement a tensor network state can have based on the minimal cut $\gamma_A$ dividing the network into two tensors $P$ and $Q$.
In this section, we will explain how to guarantee that this upper bound is saturated for connected regions in a certain class of holographic states.
Namely, we shall focus on holographic states associated to planar graphs with non-positive curvature.

We argued that proving the RT formula amounts to showing that the tensors $P$ and $Q$ can be interpreted as unitary transformations, from the minimal geodesic cut $\gamma_A$ together with some subregion of $A$ or $A^c$ to the rest of $A$ or $A^c$ respectively.  
We will show that this is indeed the case by giving unitary circuit interpretations to the tensor networks for each of them.

Indeed, the following are necessary and sufficient conditions for a circuit interpretation of a network of perfect tensors:
\begin{itemize}
\item {\bf Covering:} Each edge (contracted or uncontracted index) is assigned a directionality.
\item{\bf Flow:} Each tensor has an equal number of incoming and outgoing indices.
\item{\bf Acyclic:} The resulting directionality has no closed cycles (no loops).
\end{itemize}
The covering condition is necessary to interpret the direction in which each tensor in the network processes information by having well defined inputs and outputs.
The flow condition is necessary for the interpretation of every tensor to be that of a unitary gate.
The acyclic condition is non-local and guarantees that the order of application of the operations in the network is consistent, where an inconsistency may be thought of as the presence of closed time-like curves in the circuit picture.
These conditions are enough to show that the interpretation is that of a unitary quantum circuit.

One additional condition is required in order to prove the saturation of the RT entanglement entropy,  for a simply connected boundary region $A$. 
\begin{itemize}
\item {\bf Equal time interpretation:} The minimal cut $\gamma_A$ is crossed in the same direction by the directed interpretation of each link that it cuts.
\end{itemize}
This condition allows viewing the geodesic as an ``equal time'' curve in the unitary circuit interpretation of the tensor network.

Let us first describe the construction for the circuit interpretation.
The steps may be readily visualized in figure \ref{fig_HolographicStateOrdering}.  

In a planar graph we can associate a minimal cut $\gamma_A$ to a path through the dual graph, and since we are taking $A$ to be connected, this path will also be connected.  We may associate two nodes in the dual lattice at the end points of a simply connected boundary region $A$, which will also be the endpoints of $\gamma_A$ in the dual lattice.  We will take one of these nodes to be the starting node and label it $0$.  We may then label all other nodes in the dual lattice according to the distance (number of steps/cuts) from the starting node.
Since $\gamma_A$ is a minimal geodesic, this labeling monotonically increases along the nodes it traverses.

We may now assign an orientation to edges (contracted indices) in the tensor network.
The orientation is chosen such that, from the two adjacent nodes of the dual lattice, the node with higher label is always found to the right.
This orientation may be interpreted as the direction of ``flow of information'' through the circuit.
 
We now argue that this orientation gives a unitary circuit interpretation for $P$ (and thus for $Q$ by exchanging $A$ and $A^c$).  We emphasize that the argument rests on the following assumptions about the graph:
\begin{itemize}
\item {\bf Planar embedding:} The tensor network may be laid out in a planar fashion with the boundary of the network corresponding to a simple boundary on the embedding.
\item {\bf Perfect tensors:} Tensors in the network have an even number of legs and are unitary along any balanced distribution of the legs.
\item {\bf Curvature: } The network is expected to represent an AdS bulk and thus is expected to have the discrete analogue of negative curvature.  We have not tried to define this idea in general, but the aspect of it we need here is that the distance function between two nodes of the network has no local maxima away from the boundary.  
\end{itemize}

  \begin{figure}[h!b]
  \centering
    \subfloat[Circuit interpretation construction]{ 
    \includegraphics[width=0.44\textwidth]{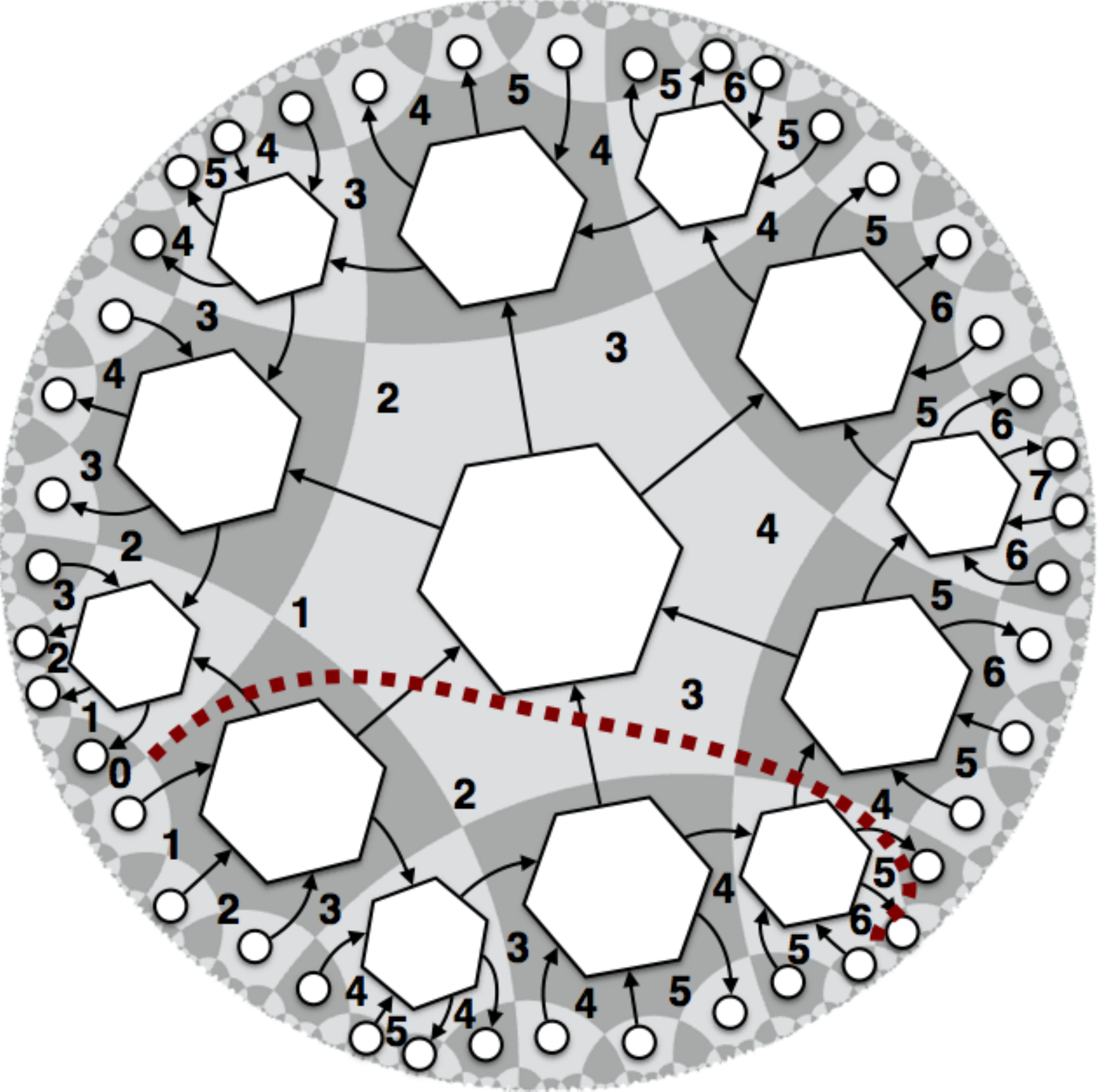} 
    \label{fig_HolographicStateOrdering}}\
   \subfloat[Bulk operator reconstruction circuit]{
    \includegraphics[width=0.44\textwidth]{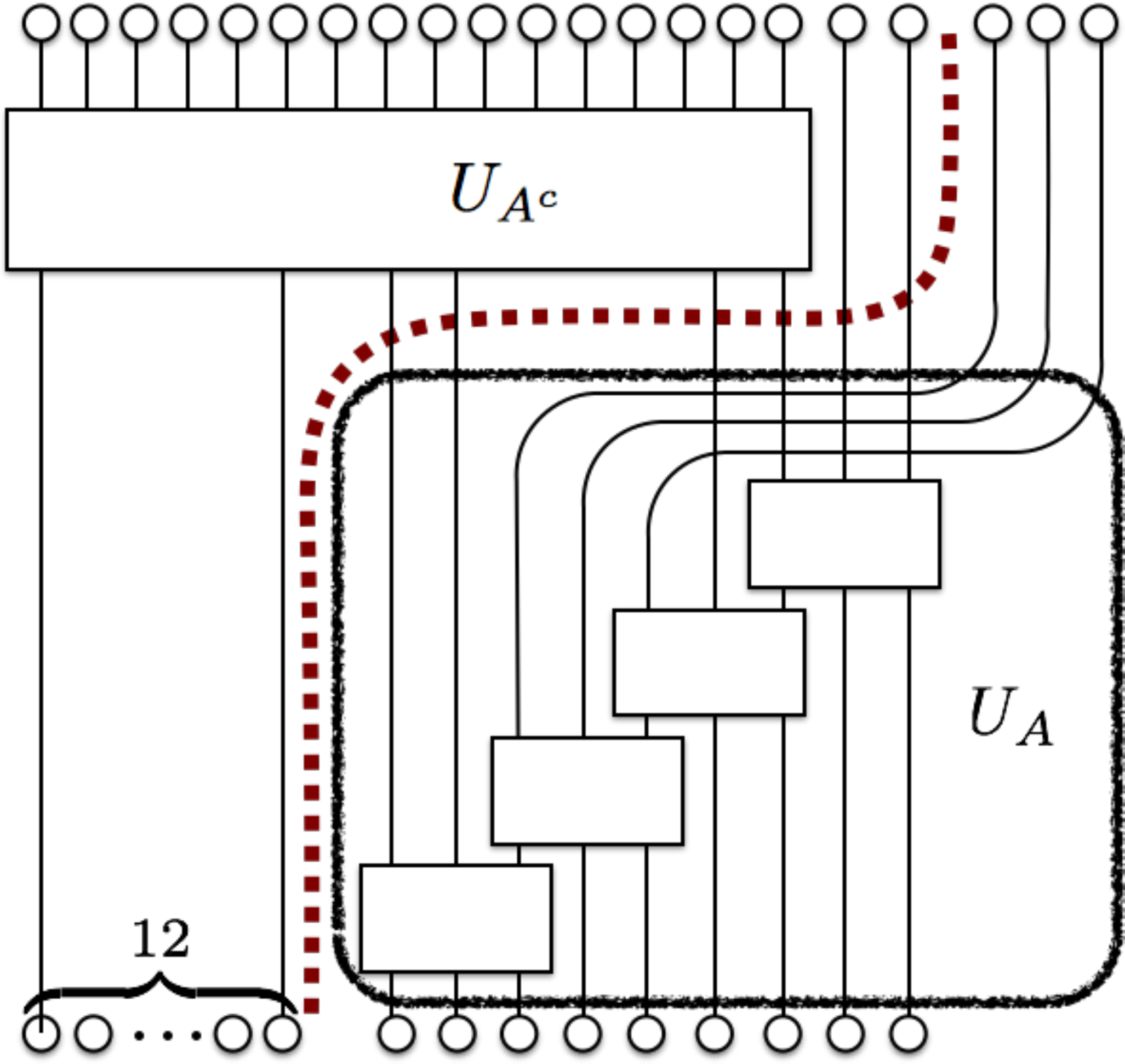}
     \label{fig:CosmeticCircuitRedraw}}\
\caption{Here we illustrate how to construct a unitary circuit interpretation of a holographic state which witnesses the RT entanglement for a simply connected boundary region.
(a) The following steps for the construction are illustrated:
i) Label the node located at one end of the boundary as $0$.
ii) Label all other nodes according to the distance from this node.
iii) Direct all tensor indices such that the larger label lies to the right.
(b) In the example, the circuit interpretation for the network has depth $12$.
For this reason we provide a full sequential presentation of the circuit interpretation along one side of the geodesic, condensing the remaining $7$ gates into $U_{A^c}$.
Note that there are outputs that are produced directly by $U_A$, without going through $U_{A^c}$ as well as inputs that are fed directly to $U_{A^c}$ without going through $U_A$.
} 
\label{HolographicStateOrderingandRedraw}
  \end{figure}

We have already used the \emph{planar embedding} assumption implicitly by constructing the dual lattice and referring to right and left.
It is however less obvious that we have also made a restricted use of the \emph{perfect tensor} assumption.
Namely, we have used the fact that each tensor has an even number of legs.
Because of this, the dual lattice is bipartite.
In other words, nodes may be labeled with two ``colors'' which we shall conveniently call `even' and `odd' such that two nodes of the same color are never adjacent.
In particular, the parity of the distance labeling coincides with the ``color''.
Hence, two neighboring nodes can not have the same value label making the directionality of the tensor index between them always be well defined (i.e. satisfies \emph{covering} condition).

To ensure the \emph{flow} condition, it suffices to count the number of incoming indices minus the number of outgoing indices and verify that this value is zero.
Due to the triangle inequality, and the bipartite nature, labels for neighboring nodes can only differ by one.
Hence, the difference between the number of outgoing and incoming indices for any tensor in the network is given by $\sum_{j=1}^{2n}  \left( f_{j+1}-f_j \right) =0$, where $f_j$ are the labels associated to the $2n$ nodes immediately surrounding the tensor taken in cyclic order. 

The \emph{acyclic} condition is a bit more subtle since it is a non-local property.
We will prove that the presence of a cycle in this context implies the existence of an interior local maximum for the labeling.
Let us assume that our construction produces some cycle $C$ in the  tensor network.
Depending on the orientation of $C$ (clockwise or counterclockwise), the the node label values immediately to the interior of the loop will be larger or smaller than those immediately to the exteriors.
In the counterclockwise case, we may chose a node in the interior of $C$ with lowest possible label.
The label for this node is smaller than those of all its neighbors, including those in the exterior of $C$, which contradicts the assumption that it is defined based on a graph distance function.
In the clockwise case, we may choose a node in the interior of $C$ with the largest possible label.
In this case the label for this node is larger than that of all its neighbors, including those outside of $C$.
In other words, we have found an interior maximum for the distance function.
This is in contradiction with our stand-in assumption associated to negatively curved surface homeomorphic to the disc, which leads us to the conclusion that our construction produces no loops.

Finally, it is straightforward to show that the geodesic $\gamma_A$ can be provided an \emph{equal time interpretation} in the circuit.
Firstly,  it completely splits the circuits in two parts.
Secondly, the direction associated to all contracted indices crossed by $\gamma_A$ is uniform since the nodes it transverses are by definition labeled in strict ascending order.  

\section{Counting tensors in the pentagon code}\label{App:CountingTensors}
\subsection{Counting tensors}
\begin{figure}
\begin{center}
\includegraphics[height=1.5cm]{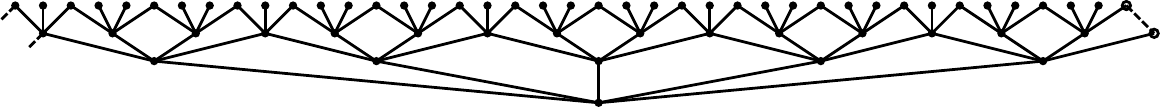}
\caption{The first few levels of the unwrapped pentagon code.  The dangling logical legs are not shown, and the ``hollow'' lines and dots are identified with the solid ones on the other side.  The tensors in the top row have all remaining non-logical legs extended outwards and treated as physical qubits.}\label{pentfig}
\end{center}
\end{figure}
In this appendix we compute some basic properties of the pentagon code that we quoted in the text.  In these computations it is useful to ``unwrap'' the code, as in figure \ref{pentfig}.  The central tensor is placed at the bottom, which we will refer to as the zeroth layer, the first five tensors as the first layer, and so on, with the layer number $n$ being equivalent to the graph distance to the central tensor.  

At each layer, it is clear that there are two kinds of tensors: those with one leg connected to the previous layer and those with two.  We will denote the numbers of these at layer $n$ as $f_n$ and $g_n$, so for example $f_1=5$ and $g_1=0$.  Moreover if we group these together as a two component vector, we have the recursive equation
\be
\begin{pmatrix}
f_{n+1}\\
g_{n+1}
\end{pmatrix}
=
\begin{pmatrix}
2 & 1\\
1 & 1
\end{pmatrix}
\begin{pmatrix}
f_{n}\\
g_{n}
\end{pmatrix}.
\ee
Applying this equation repeatedly we can compute the number of tensors of either type at any $n$ via
\be
\begin{pmatrix}
f_{n}\\
g_{n}
\end{pmatrix}
=\begin{pmatrix}
2 & 1\\
1 & 1
\end{pmatrix}^{n-1}
\begin{pmatrix}
5\\
0
\end{pmatrix}.
\ee 
This is easily computed by diagonalizing the matrix $M\equiv \begin{pmatrix}2 & 1\\1 & 1\end{pmatrix}$, at large $n$ we have
\begin{align}\nonumber
f_n&=\frac{5-\sqrt{5}}{2}\left(\frac{3+\sqrt{5}}{2}\right)^n\left[1+O\left(\left(\frac{3-\sqrt{5}}{3+\sqrt{5}}\right)^n\right)\right]\\
g_n&=\frac{3\sqrt{5}-5}{2}\left(\frac{3+\sqrt{5}}{2}\right)^n\left[1+O\left(\left(\frac{3-\sqrt{5}}{3+\sqrt{5}}\right)^n\right)\right].
\end{align}
If we truncate at layer $n$, the total number of boundary qubits is
\be
N_{boundary}=4f(n)+3f(n)
\ee
and the total number of bulk tensors is
\be
N_{bulk}=1+\sum_{k=1}^n(f_k+g_k).
\ee
Asymptotically we have
\be\label{pentrate}
\frac{N_{bulk}}{N_{boundary}}\to \frac{1}{\sqrt{5}}\approx .447,
\ee
which reproduces equation \eqref{eq:RatioCount}.

\subsection{Connected reconstruction}
We'll now compute the size of connected boundary region that we need to reconstruct operators on the logical leg of the central ($n=0$) tensor.  This is complicated by the fact that our network is not translationally invariant; whether or not we can reconstruct the center depends not only on the size of the boundary region we have access to but also where it is.  The ``best case'' situation is illustrated in figure \ref{bestpentfig}.
\begin{figure}
\begin{center}
\includegraphics[height=1.5cm]{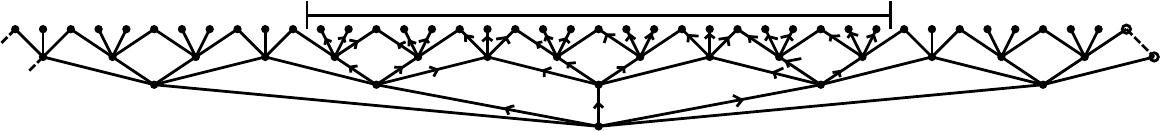}
\caption{The best-case reconstruction of the center; we need to use only three ``branches'' of the tree, and in the outer branches we can choose to push ``inwards'' every time.}\label{bestpentfig}
\end{center}
\end{figure}

We can compute the number of physical qubits in the best-case region by growing the tree it bounds: we start with $f_1=3$ and $g_1=0$, and then proceed as before by applying $M$ repeatedly, being careful to remove two tensors at the ends at each layer.  This tree thus obeys the modifed recursion relation
\be
\begin{pmatrix}
f_{n+1}\\
g_{n+1}
\end{pmatrix}
=
\begin{pmatrix}
2 & 1\\
1 & 1
\end{pmatrix}
\begin{pmatrix}
f_{n}\\
g_{n}
\end{pmatrix}-\begin{pmatrix}
0\\
1
\end{pmatrix},
\ee
which has solution
\be
\begin{pmatrix}
f_{n}\\
g_{n}
\end{pmatrix}
=\begin{pmatrix}
2 & 1\\
1 & 1
\end{pmatrix}^{n-1}
\begin{pmatrix}
3\\
0
\end{pmatrix}-\sum_{k=0}^{n-2}\begin{pmatrix}
2 & 1\\
1 & 1
\end{pmatrix}^{k}\begin{pmatrix}
0\\
1
\end{pmatrix}.
\ee
After cutting off the tree at level $n$ the number of physical qubits in the best-case region will be $N_{best}=4f_n+3g_n-3$, and by diagonalizing $M$ we see that asymptotically
\be
\frac{N_{best}}{N_{boundary}}\to \frac{3+\sqrt{5}}{10}\approx .524.
\ee

\begin{figure}
\begin{center}
\includegraphics[height=1.5cm]{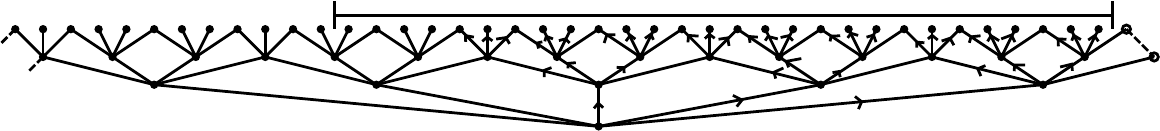}
\caption{The worst-case reconstruction of the center; we are just barely not able to use the second branch, so we need to use the third, fourth, and fifth.}\label{worstpentfig}
\end{center}
\end{figure}
We can also study the ``worst case'' location of the region, which is shown in figure \ref{worstpentfig}.  To compute its size in the large $n$ limit we only need to replace the initial condition $f_1=3, g_1=0$ by $f_1=4, g_1=0$, and we find
\be\label{pentthres}
\frac{N_{worst}}{N_{boundary}}\to \frac{5+\sqrt{5}}{10} \approx .724,
\ee
which reproduces \eqref{cthres}.

We can connect this worst-case result to the ``bad'' sets of points from figure \ref{fig:benipoints} that prevent a general threshold for this code: these points are shown in figure \ref{badpointsfig}.
\begin{figure}
\begin{center}
\includegraphics[height=1.5cm]{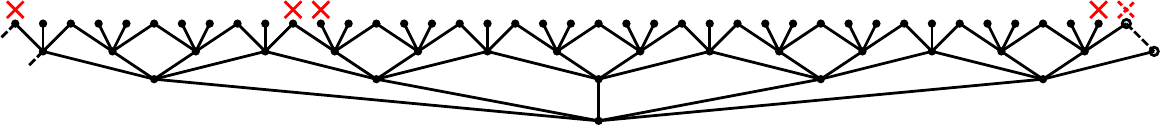}
\caption{Locations of the ``bad'' points from figure \ref{fig:benipoints}.  If the worst-case region were just a little smaller on the right it wouldn't contain any of them so, it wouldn't be able to reconstruct the center.}\label{badpointsfig}
\end{center}
\end{figure}

Similar calculations are possible in the pentagon-hexagon code we introduced to restore the threshold, we find
\be
\frac{N_{bulk}}{N_{boundary}}\to \begin{cases} \frac{3\sqrt{6}-4}{38}\approx .088 & n ~ \mathrm{odd}\\ \frac{3\sqrt{6}+4}{38}\approx .299 & n ~ \mathrm{even} \end{cases},
\ee
with the two cases being whether the last layer is taken to be pentagons ($n$ even) or hexagons ($n$ odd).  The rate is thus relatively small compared to \eqref{pentrate}, which suggests that this code should be better protected against erasures, as indeed we find.  The rate when $n$ is odd is smaller since it throws in an extra level of hexagons without any new logical legs.  We can also compute the sizes of the best and worst case connected reconstructions of the center, we find
\begin{align}\nonumber
\frac{N_{best}}{N_{tot}}&\to \frac{6+\sqrt{6}}{20}\approx .422\\
\frac{N_{worst}}{N_{tot}}&\to\frac{10+\sqrt{6}}{20}\approx .622.
\end{align}
These are smaller than the pentagon results, as expected since the code is denser, and are closer to the AdS/CFT value of $1/2$. 

\section{Estimating greedy erasure thresholds}\label{App:ErasureThreshold}

Noise models which are independent and identically distributed (i.i.d.) are usually analytically tractable while providing reasonable predictive capability based on the few parameters defining the individual noise model.
In the case of erasure noise, this is particularly simple since the only parameter is $\epsilon$, the erasure probability per qubit.
In addition to randomized bench-marking it is sometimes instructive to provide analytic bounds on how small the error/erasure probability needs to be in order to guarantee a recovery probability for the encoded data which approaches unit.
This can give us information about the scaling of the logical error probability with other parameters of the code.
A complementary approach consists of providing precise numerical estimates of the threshold value $p_c$ based on numerical simulations.

\subsection{Analytic bounds}\label{sec:AnalitycBounds}

In this section, we derive such an analytic upper bound on the probability of logical error.
The goal of the section is to provide an example of how such a bound is derived illustrating a proof technique and obtaining a functional form for the logical error probability.
We do not strive to derive a tight bound or address a particularly relevant holographic code scenario.
In fact, the recovery procedure we model is strictly weaker than the one provided by the greedy algorithm which is itself weaker than an optimal erasure recovery algorithm.
The model we analyze is essentially identical to the holographic hexagon state except that instead of starting from a single $[[6,0,4]]$ tensor at the center, we start with a single $[[5,1,3]]$ and build up $n$ layers of $[[6,0,4]]$ tensors from there. 
We shall call this the \emph{single qubit hexagon code} and its analysis is essentially identical to that of a holographic hexagon state.

For this code we obtain the following conclusion.
\begin{theorem} Consider a single qubit hexagon code with $n$ layers and an i.i.d. erasure probability for physical qubit given by $\epsilon \leq \epsilon^\star$.
Then it is possible to recover the central logical qubit with probability $p$ greater than 
\begin{align}
p \geq 1 - \epsilon^\star \left(\frac{\epsilon}{\epsilon^\star}\right)^{\lambda^n}.
\end{align}
Here, $\epsilon^\star = 1/12$ and $\lambda = \frac{1+\sqrt{5}}{2}$.
\end{theorem}

This is the same functional form associated to concatenated error correcting codes.
Namely, the loss probability for the logical data decays doubly exponentially with the ``depth'' $n$ of the code or exponentially with the number of physical spins.
We would expect to get a result of the same form for any such code with an erasure threshold.
The only expected difference being the value of the threshold $\epsilon^\star$ and the scaling dimension $\lambda$.
We will now prove the theorem as an illustration to obtaining these values.
Since we will use a simplified hierarchical recovery procedure, the proof technique will be essentially equivalent to that of concatenated codes.
An interesting open problem is to provide an analytic threshold analysis fully respecting the greedy algorithm which corresponds to the problem of bootstrap percolation \cite{Adler1991}.
We have numerically analyzed this problem or the qubit hexagon code and found $\epsilon^\star \approx 0.48(2)$ whereas independent analytic arguments particular to this model predict an erasure threshold of $\epsilon^\star =1/2$ for an optimal recovery protocol.

\begin{proof}
We will consider a hierarchical recovery model which is even simpler and can not perform better than the greedy algorithm.
Namely, we may provide a coupling \cite{Levin} between the probability distributions over recovered tensors such that the set of recovered tensors by the greedy algorithm always includes the set of  tensors recovered by the hierarchical recovery.
The reason for this is that the hierarchical algorithm can be interpreted as $n$ iterations of the greedy algorithm where tensors a distance $j$ from the boundary may be incorporated only during iteration $j$.
The difference with the greedy algorithm, is that once the greedy algorithm recovers a tensor at distance $j+1$ from the boundary, it allows itself to reinspect its neighboring tensors at distance $j$ and incorporate those.
This may in general lead to highly non-trivial sequences for incorporating tensors in the bulk.
The sequential nature of the hierarchical recovery model allows establishing a clear dependence between the tensors.

In the hierarchical recovery model, each level consists of a ring of tensors, which are only connected with the next level and the previous.
In order to adequately model the errors at each level, we will need to inductively provide bounds for different error configurations.
Assuming we are dealing with the hexagon lattice with four hexagons adjacent per vertex, it will be sufficient to deal with two types of bounds, one for single errors (single missing tensor) and the second for pairs of neighboring missing tensors.
We will call these bounds $s$ and $d$ for single and double and we will use a subindex $j$ to label the layer to which these bounds apply.

Initially, we have $s_0 = \epsilon$ and $d_0 = \epsilon^2$ which corresponds to assuming an i.i.d. erasure model with each physical index being erased with probability $\epsilon$.
The core of the proof is simply to recursively bound $s_{j+1}$ and $d_{j+1}$ in terms of $s_{j}$ and $d_{j}$.

A non-trivial observation, is that we do not need to consider erasure correlations beyond nearest neighbors.
Due to the hierarchical structure which is contracting, correlated erasures beyond nearest neighbors of a chain can not exist (see figure \ref{fig:BoundedRangeCorrelations}).
This is an artifact of having chosen a privileged ``re-normalization'' direction and is an effect analogous to having all  scaling operators be three body in MERA.
In fact, long range correlations between tensors do arise in the recovery model dictated by the greedy algorithm.
\begin{figure}
\centering
\includegraphics[width=0.80\linewidth]{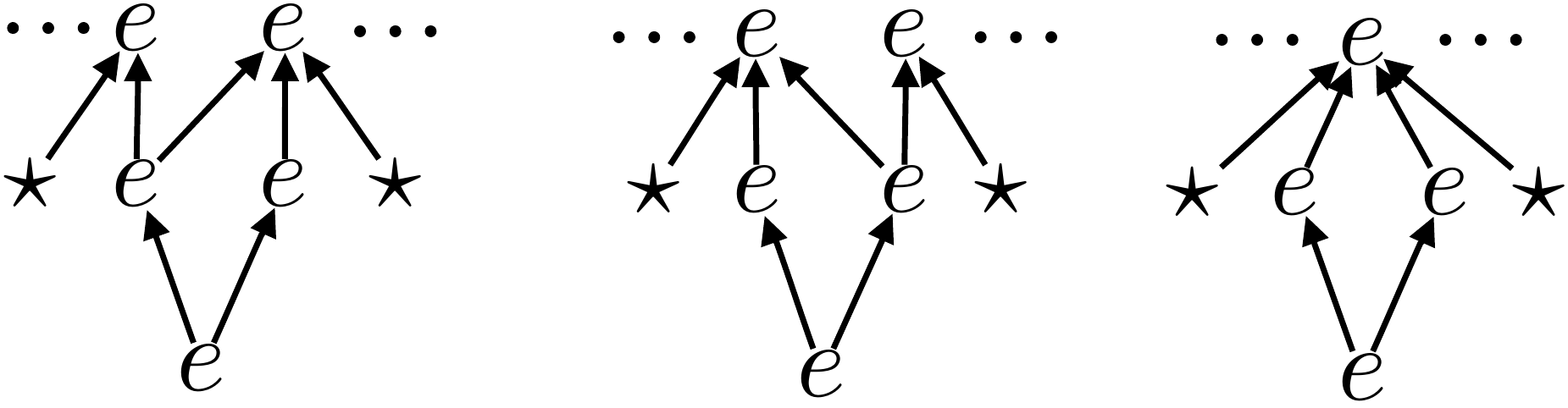}
\caption{We illustrate how correlated erasures can not grow beyond nearest neighbor in the hierarchical recovery model.
} 
\label{fig:BoundedRangeCorrelations}
\end{figure}
Here $e$ indicates where a reconstruction index is missing or erased whereas $\star$ indicates that the index could be missing or available.

Considering all error configuration at layer $j$ which could lead to errors at layer $j+1$,
we may bound
\begin{align}
   s_{j+1} &\leq 3 d_j + 3 s_j^2 \label{eq:singleBound}\\
   d_{j+1} &\leq 3s_j (3 d_n + 3 s_j^2). \label{eq:doubleBound}
\end{align}
Here, we have aimed for simplicity instead of tightness of the bound.
Let us give a brief explanation for the RHS of equations \ref{eq:singleBound} and \ref{eq:doubleBound}.

Since each hexagon has at least 4 legs connected to lower layers, two of its neighboring tensors (within four legs) need to be missing such that tensor fails to be recovered from the lower layer.
For this to happen, there must either be two neighboring indices missing from the lower chain or two non-neighboring indices missing.
There are three ways for this to happen illustrated by the following minimal error strings
\begin{align}
   \{  ee\star\star,  \star ee\star, \star \star ee, \\
   e\star e \star,  e \star \star e, \star e\star e \}.
\end{align}
These scenarios cover all possible situations leading to the failure of hierarchical reconstruction (some of the, more unlikely ones such as $eeee$ are being covered multiple times).

Similarly, we may account for all possible scenarios that lead to a double erasure $ee$ at level $j+1$
Two consecutive tensors at chain $j+1$ always share exactly one descendant which may provide a source of correlated errors.
Furthermore, at least one of the two tensors at layer $j+1$ will have a total of five descendants.
As an overestimate of $d_j$ we may disregard the value of the joint descendant.
Regardless, we know that there should be two erasures in the remaining four descendants of the tensor with five descendants and at least one erasure among the outermost three descendants of the other tensor (see figure \ref{fig:CorrelatedErrors}).
\begin{figure}
\centering
\includegraphics[width=0.40\linewidth]{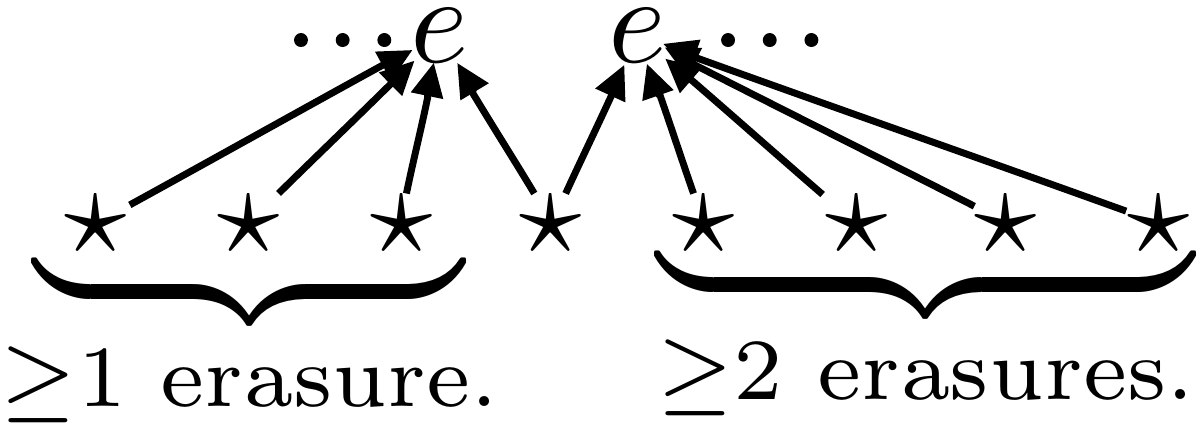}
\caption{We illustrate the prerequisite for two erasures to be propagated one layer higher.
The most fragile case corresponds to one of the sites having only four descendants since at least one of them needs to have five.
Irrespective of the availability  of a shared descendant there is a minimum number of erasures that need to occur in order to propagate two contiguous erasures.
} 
\label{fig:CorrelatedErrors}
\end{figure}

Assuming $s_j \geq d_j \geq s_j^2$ which must hold for such a model, we may extract the dominant (smallest) exponent for $\epsilon$ associated to the  $s_j$ and $d_j$ bounds given by the recursion relation
\begin{align}
\begin{pmatrix} \deg(s_{n+1},\epsilon) \\ \deg(d_{n+1},\epsilon) \end{pmatrix} = 
\begin{pmatrix} 0 & 1 \\ 1 & 1 \end{pmatrix} \begin{pmatrix} \deg(s_{n},\epsilon) \\ \deg(d_{n},\epsilon) \end{pmatrix}
\end{align}
Here, by $\deg(p(x),x)$ denotes the lowest  degree exponent of $x$ present in the polynomial $p(x)$.
The recursion matrix has the eigenvalues $\phi_{\pm} = \frac{1\pm\sqrt{5}}{2}$ and the $\phi_+$ eigenvector being $(1,\phi_+)$.
 This means that the $\epsilon$ exponent for $s_n$ increases exponentially as $\phi_+^n$.
 
We may now calculate the fix-point solution of inequalities \ref{eq:singleBound} and \ref{eq:doubleBound}, taking them as equations and find $(s^\star, d^\star)=(1/12,1/48)$.
Assume that $s_0 \leq r s^\star$  for some $r\leq 1$ and consequently, $d_0 = s_0^2 \leq r^{\phi_+} d^\star$.
We may prove inductively that  $s_j \leq r^{\phi_+^j}s^\star$ and $d_j \leq r^{\phi_+^{j+1}}d^\star$.
In order to do so, one need only verify
\begin{align}
       s_{j+1} &\leq 3 d^\star r^{\phi^{j+1}} + 3 s^\star r^{2\phi^j} &\leq r^{\phi_+^{j+1}} s^\star\\
      d_{j+1} &\leq  3(s^\star r^{\phi^{j}}) 3 ( d^\star r^{\phi^{j+1}} + 3 s^\star r^{2\phi^j} ) &\leq r^{\phi_+^{j+2}} d^\star.
\end{align}
Where we may divide by an appropriate power of $r$ such as $r^{\phi_+^{j+1}}$ and obtain
\begin{align}
 3 d^\star + 3 s^\star r^{\phi^{n-2}} &\leq s^\star\\
 3(s^\star ) 3 ( d^\star + 3 s^\star r^{\frac{3-\sqrt{5}}{2}\phi^n} ) &\leq d^\star.
\end{align}
This allows us to reach the conclusion of the theorem with $\epsilon^\star = s^\star$ and $\lambda = \phi^+$.
\end{proof}

A bounding procedure similar to this one may also be applied to the code involving alternating layers of pentagons and hexagons as for many other holographic codes with certain regularity structure.
We have only restricted to consider the hexagon lattice to exemplify the kind of reasoning involved.
In the case of the holographic pentagon code, there is no way to make such  an argument work and there is a good reason for this.
Namely there are constant weight $4$ logical operators affecting the central qubit in the holographic pentagon code.
The way this becomes manifest when attempting a similar proof approach is by obtaining a scaling dimension $\lambda=1$.

\subsection{Numerical evaluation}

We may numerically evaluate the probability for the greedy algorithm to absorb the central tensor given a boundary region constructed by erasing a random set of physical indices according to an i.i.d. distribution.
This gives us a conservative estimate of how well protected the central qubit is from i.i.d. erasure since an optimal recovery method can only do better.
We perform such estimates for different values of the lattice radius in order to identify the value $p_c$ associated to a correctability phase transition.
For the regular pentagon lattice tensor network we find that there is no indication of the central qubit being well protected.

Given that the pentagon code does not have a threshold, we introduce the pentagon/hexagon code as a similar example that does have an erasure threshold in this context.
It is the regular lattice composed of pentagons and hexagons, with two pentagons and two hexagons adjacent at each vertex.
Such a lattice might employ $[[6,0,4]]_2$ and $[[5,1,3]]_2$ tensors.
Intuitively, we expect that by diluting the number of logical legs, we may obtain better protection for the encoded logical qubits.
Indeed, for such a lattice, we find that there is a threshold value $p^*$ such that if the i.i.d. probability of boundary erasure is smaller than $p^*$, then the central logical qubit may be reconstructed with a probability which approaches one as the cutoff radius of the lattice is increased.
Such a statement can be proven using techniques very similar to the threshold proof of section \ref{sec:AnalitycBounds}.
We have numerically tested the performance of the greedy recovery algorithm for recovering the central qubit in three possible lattices (see figure \ref{fig:Numerics}) supporting our claim that a lower density of logical legs leads to a higher tolerable erasure threshold.

  \begin{figure}\label{fig:Numerics}
  \centering
    \subfloat[Holographic pentagon code]{ 
  \includegraphics[width=0.31\textwidth]{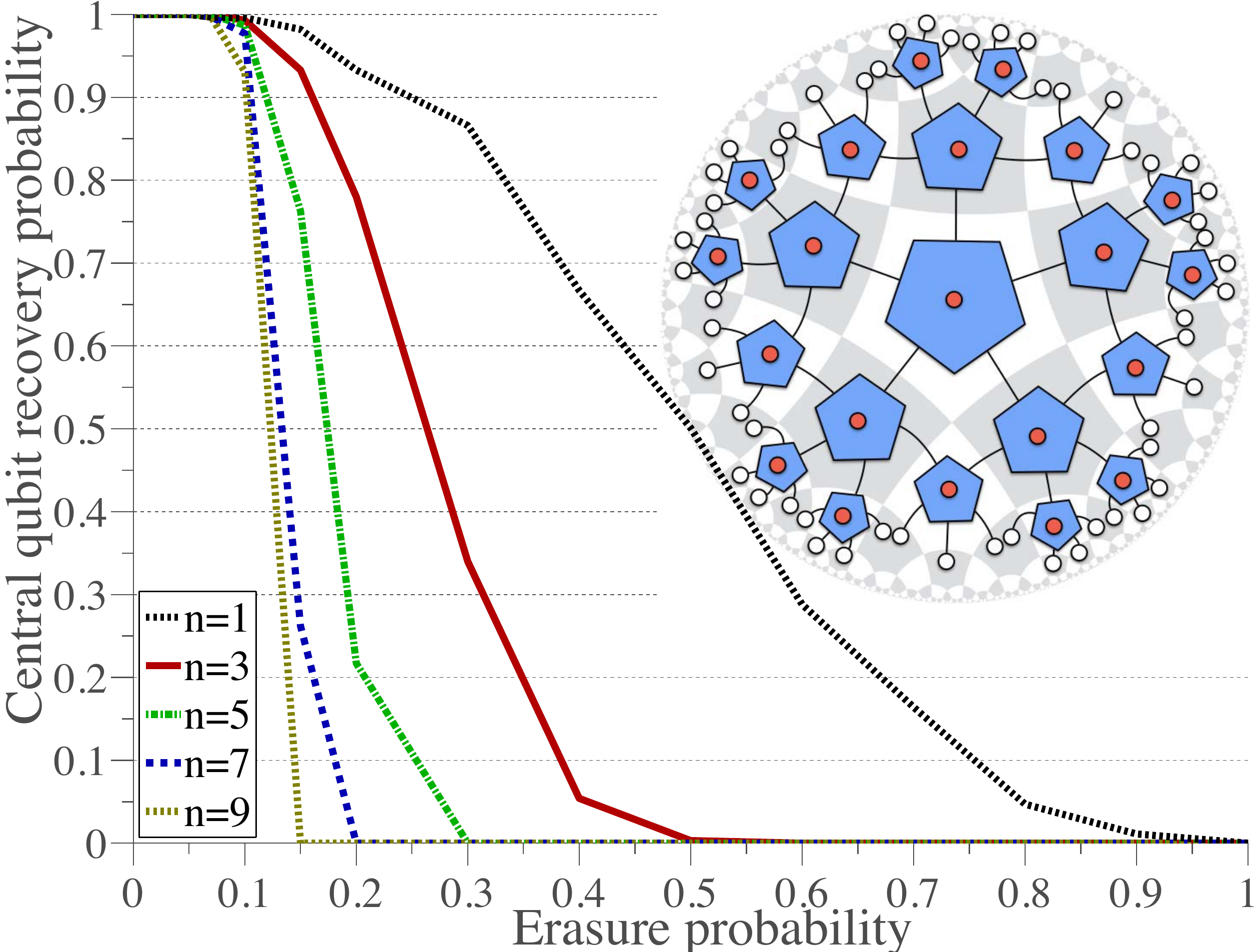} \label{fig:PentagonCodeNumerics}}\
   \subfloat[Pentagon/hexagon code]{
    \includegraphics[width=0.31\textwidth]{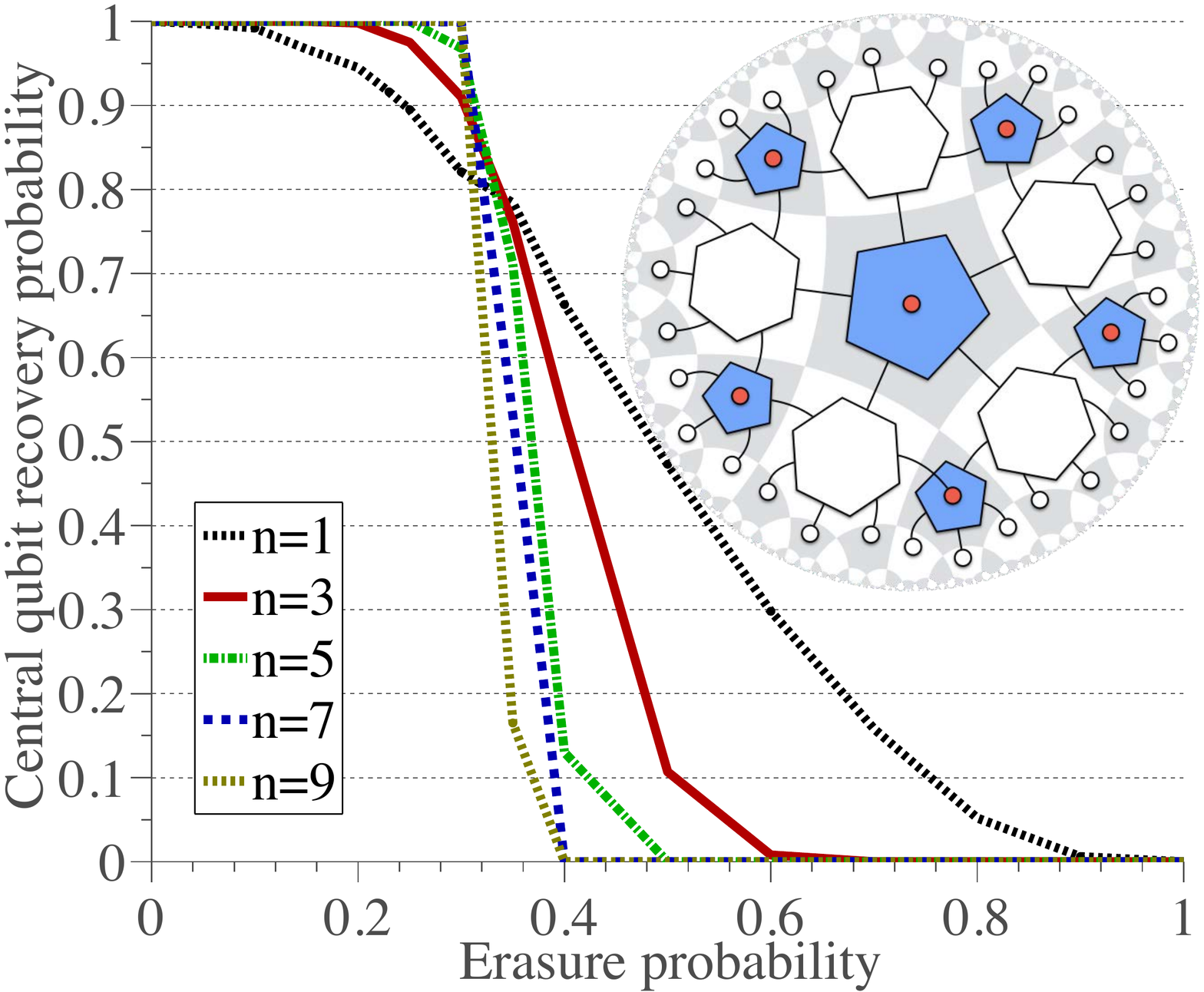}
     \label{fig:PentagonHexagonNumerics}}\
        \subfloat[Single qubit  hexagon code]{
         \includegraphics[width=0.31\textwidth]{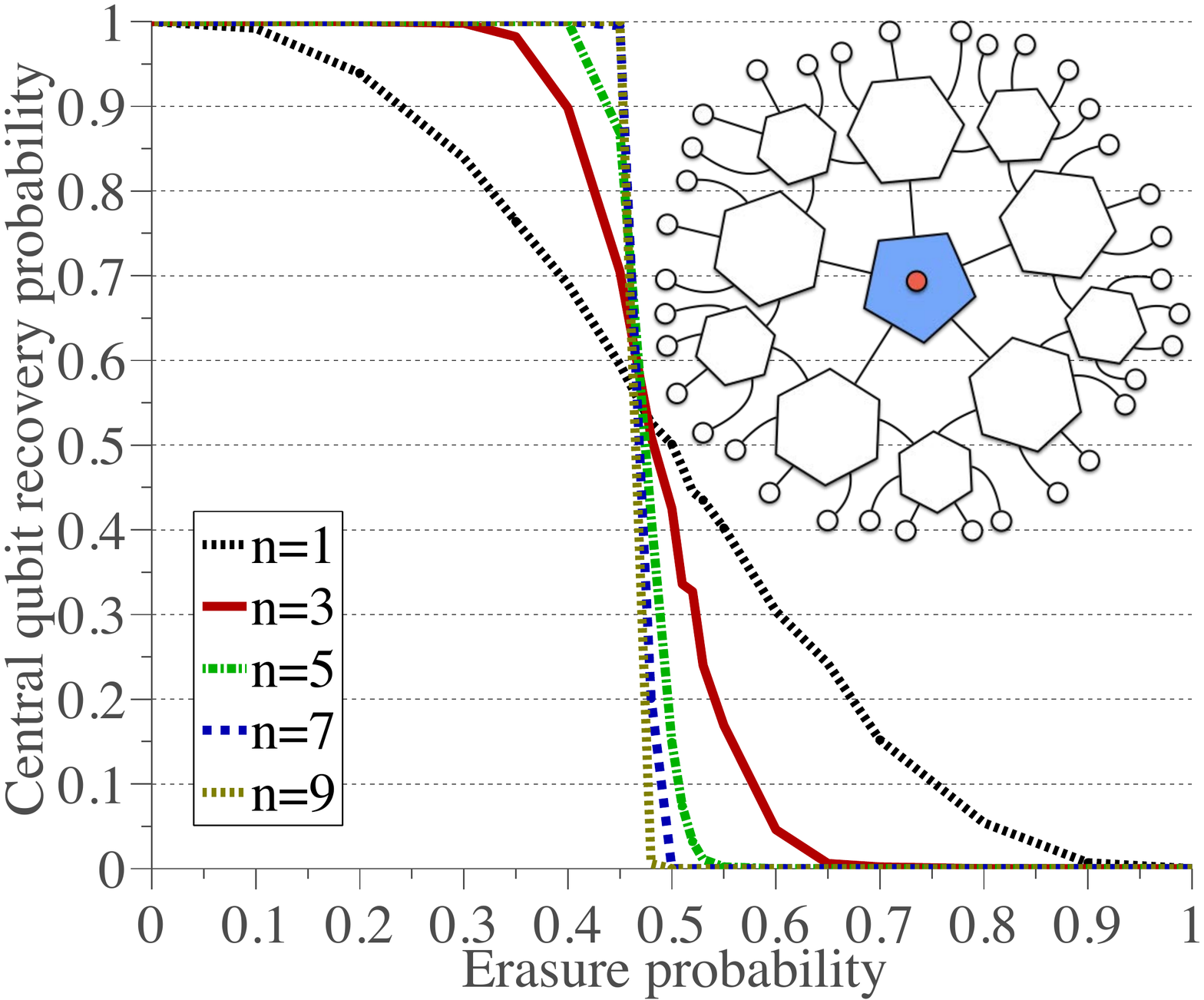}
    \label{fig:1qubitHexagonNumerics}}   
\caption{Here we present Monte Carlo simulation for the probability of the central tensor being incorporated by the greedy algorithm applied to a holographic code.
The greedy algorithm is applied to a region $A$ constructed by randomly taking each boundary physical index to belong to $A$ with probability $p$.
We plot the central tensor containment probability in the greedy wedge according to the lattice radius (i.e. the distance from the central tensor at which the a priory infinite network is truncated). 
(a) We consider the holographic pentagon code of figure \ref{fig:HolographicPentagonCode}.
Numerical results remains consistent with the existence of five possible weight $4$ representations of the string-like logical operators acting on the central qubit.
(b) We focus on the pentagon/hexagon code of figure \ref{fig:PentagonHexagonCode} which has an erasure threshold in terms of the recoverabilty.
We observe some oscillatory behavior due to the fact that tensor `layers' added alternate between pentagons and hexagons.
(c) We present numerical data for the greedy algorithm applied to the 1 qubit hexagon code  of figure \ref{fig:1qubitHexagon} which corresponds to a tensor network identical to the holographic hexagon state except for having a single pentagon at its center.
}
\label{fig:Numerics}
  \end{figure}

\section{Reconstructing beyond the greedy algorithm}\label{App:BeyondGreedy}

The greedy algorithm provides an explicit prescription for representing bulk logical operators on a specified region of the boundary.
A key virtue is that the region obtained does not depend in any way on the specific perfect tensors used to construct the network.
In this sense, it is analogous to the AdS/Rindler reconstruction, which is explicit and applicable to a large family of models satisfying a holographic correspondence.

\subsection{Reconstruction from symmetry guarantees}

The greedy entanglement wedge falls short of the expectations for a geometric entanglement wedge in certain ways.
For instance, in the scenario where the minimal surface separating $A$ from $A^c$ is well defined, we expect $\cE[A] \cup \cE[A^c]$ to contain the full lattice.
Our first example of reconstruction beyond the greedy wedge involves a family of holographic stabilizer codes with a single logical qudit.
In this case, perfect reconstruction on either $A$ or its complement $A^c$ can be guaranteed by exploiting a symmetry.

Particularly, the three qutrit stabilizer code of section \ref{sec:qutritcode} is of CSS \cite{Nielsen_Chuang} type, a property  which we can show is inherited by any derived holographic code by following the arguments of section \ref{subsec:stabilizer}.
Furthermore, the qutrit Hadamard operator $H$ is a symmetry of the qutrit code meaning that applying $H$ to all tensor indices preserves the tensor.
The Hadamard operator is symmetric, unitary but generally not Hermitian
 \footnote{
For general qudits, the Hadamard gate is given by $H=\frac{1}{\sqrt{d}}\sum_{i,j} \omega^{ij}| i\rangle\langle j |$, where $\omega=e^{2\pi i/d}$.}.
 and is specified by its action on the generators $X$ and $Z$
\begin{align}
H X H^\dagger = Z^\dagger \qquad H Z H^\dagger = X.
\end{align}

The local symmetry of the tensors gives rise to a global symmetry on the full tensor network\footnote{
Here we limit ourselves to provide the simplest example which conveys the general spirit of deriving global symmetries from local tensor symmetries \cite{Schuch2010}.
The state of the art for this line of reasoning in tensor networks can be found in \cite{Sahinoglu2014, Buerschaper2014}.}.
To see this, we multiply each tensor leg either by $H$ or its inverse $H^\dagger=H^*$.
Since each individual tensor is invariant under such an action, the full tensor network should be invariant.
Furthermore, if we assume that the tensor network graph  is bipartite we may  alternate multiplication by $H$ and $H^\dagger$ such that these may locally cancel on all contracted indices, as depicted in figure \ref{fig:Local2GlobalSymmetry}.
We are then left with a symmetry acting exclusively on the free bulk and boundary legs of the tensor network.
This symmetry guarantees a form of duality between $X$-type logical operator and $Z$-type logical operators where dual operators have exactly the same support. 

\begin{figure}
  \centering
    \subfloat[From local to global symmetry]{ 
    \includegraphics[width=0.3\textwidth]{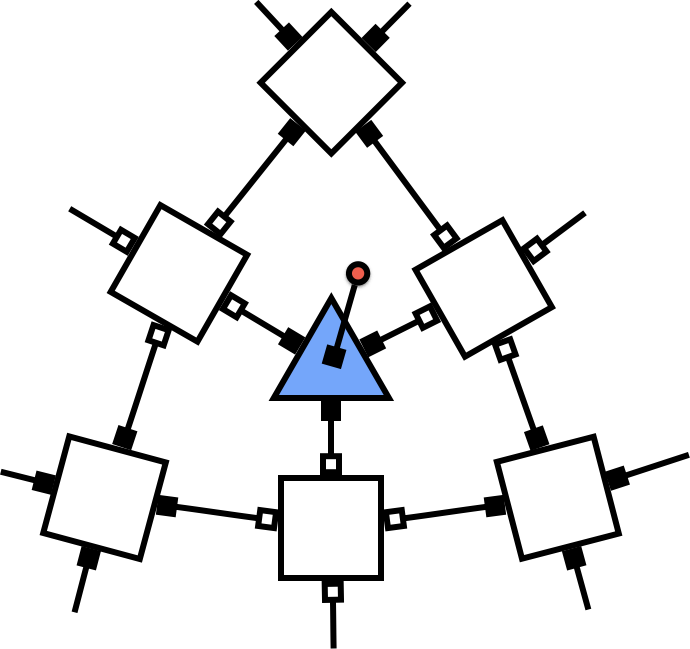} \label{fig:Local2GlobalSymmetry}}\
   \subfloat[Empty greedy wedge]{
    \includegraphics[width=0.3\textwidth]{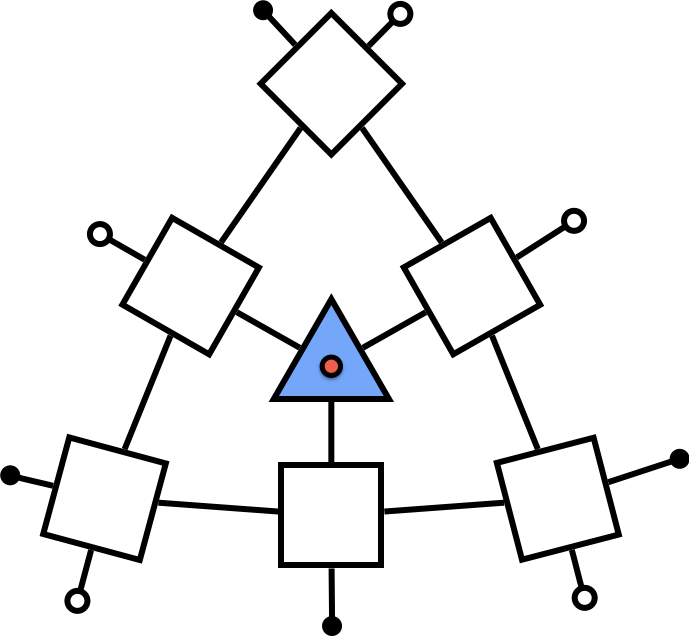}
     \label{fig:GreedyStuck}}
\caption{In (a) we represent a bipartite network composed of tensors with a symmetry $H$ and $H^*$.
We alternate applying $H$ (solid squares) and $H^*$ (hollow squares) to all legs of tensors in the bipartite tensor network.
In (b) we represent the same tensor network where a portion $A$ of the boundary was marked with hollow dots and its complement $A^c$ was marked with full dots.
The regions $A$ and $A^c$ have been chosen such that the greedy algorithm does not progress on either.
The greedy algorithm will recover the full network when initiated from the full boundary showing that the network indeed corresponds to a holographic code.
}
\label{fig:BeyondGreedy}
  \end{figure}

Theorem 1 in of Ref. \cite{Beni10} precisely relates the number of independent logical operators supported on complementary subsets of qudits.
Their result applies to general subsystem codes and includes a sharper claim for CSS codes~\cite{Haah2012}. 
In particular, for any subset $A$ of qudits, one may define $\ell(A)$ to be the number of independent Pauli logical operators supported exclusively on $A$.
\begin{lemma} Given a stabilizer code with $k$ logical qubits, $\ell(A) + \ell(A^c) = 2k$. 
Furthermore if the code is of CSS type, we have $\ell^Z(A) + \ell^X(A^c) = k = \ell^Z(A^c)+\ell^X(A)$, where $\ellˆX$ and $\ell^Z$ denote the number of $X$-type and $Z$-type generators respectively. 
\end{lemma}
This is called the cleaning lemma for stabilizer codes and applies to prime dimension qudits\footnote{Let us clarify that the operators representable on $A$ and on $A^c$ need not be mutually independent.}.

Assume that we are dealing with a CSS code with a single logical qudit, $k=1$ and a Hadamard type symmetry which guarantees $\ell^X(A)=\ell^Z(A)$.
From this we may exclude the case $\ell(A) = \ell(A^c) = 1$ and conclude that the full logical algebra may be reconstructed either on $A$ or on $A^c$.
This conclusion is analogous to $\cE[A] \cup \cE[A^c]$ covering the full bulk, which is  expected from the usual geometric entanglement wedge.
In contrast, the greedy algorithm does not provide such a guarantee for the greedy entanglement wedge.
Figure \ref{fig:GreedyStuck} illustrates a partition of the boundary of a tensor network into two regions such that the greedy algorithm does not make progress in either region.
The same tensor network may be associated to a CSS type stabilizer code with self-duality properties where all the previously exposed arguments apply.

The same cleaning lemma may be used to guarantee that when $|A^c|=4$ qubits are deleted, at least $2k-8$ independent logical Pauli operators can be reconstructed on $A$, where $k$ is the number of logical qubits in the code\footnote{In fact, a slightly more careful analysis allows us to conclude that at least $2k-2$ independent generators may be represented on $A$.}.
In the context of the example of figure \ref{fig:benipoints}, even though the greedy entanglement wedge of $A$ lacks a large number of tensors, the number of missing generators  to reconstructed the full algebra is small.

\subsection{Approximate reconstruction for typical tensors}

Consider a connected residual region $R$ obtained after removing the greedy entanglement wedge associated to boundary region $A$ and the one associated to its complement $A^c$.
A ``typical'' residual region will be composed of randomly chosen perfect tensors without any specific symmetry imposed.
In this case, we may average the entanglement entropy associated to boundaries $\gamma^\star_A$ and $\gamma^\star_{A^c}$.
We expect that for $|\gamma^\star_A| \geq | \gamma^\star_{A^c} | + n_R$ the map from the bulk logical indices and the smaller boundary onto the larger one will generically be full rank.
Furthermore, we conjecture that for random perfect tensors, the average value of $S_{A}$ will approach $ | \gamma^\star_{A^c} | + n_R$ exponentially with $|\gamma^\star_A| - | \gamma^\star_{A^c} | - n_R$.
We expect an argument analogous to that of Ref. \cite{Hayden07} to allow us to reach such a conclusion.
In turn this would imply that the logical operators in the residual region can   be reconstructed on $\gamma^\star_A$ (and in turn on $A$) to a good approximation.


\providecommand{\href}[2]{#2}\begingroup\raggedright\endgroup

\end{document}